%% file: main.tex
\documentclass[preprint]{elsarticle}
\input{aux/packages}
\input{aux/macros_global}

\input{aux/macros_local}

\begin{document}
\hypersetup{
    linkcolor=blue, 	
    anchorcolor=black, 	
    citecolor=green, 	
    filecolor=cyan, 	
    urlcolor=magenta, 
    pdftitle={DSEF2},
    pdfauthor={E. Emanuel Rapsch}
    }

\input{contents/frontmatter}

\tableofcontents

\input{contents/text}

\phantomsection
\addcontentsline{toc}{section}{References}
\bibliography{main}
\bibliographystyle{plain}

\newpage
\appendix
\input{contents/appendix}

\end{document}

%% file: aux/packages.tex
\usepackage[utf8]{inputenc}
\usepackage[T1]{fontenc}
\usepackage[english]{babel}
\usepackage{amsmath}
\usepackage{amssymb}
\usepackage{amsthm} 
\usepackage{amscd}
\usepackage{bm}
\usepackage{mathtools}
\usepackage{mathrsfs}
\usepackage{mathdots}
\usepackage{fancyhdr}
\usepackage{fancyvrb}
\usepackage{lmodern}
\usepackage[top=4cm]{geometry}
\usepackage[shortlabels]{enumitem}
\setlist[enumerate,1]{
{1.},
ref={\arabic*}
}
\setlist[enumerate,2]{
{(a)},
ref={\theenumi{}(\alph*)}
}
\setlist[itemize]{
label={--}
}
\usepackage{hyperref}
\hypersetup{
    colorlinks=true,
    linkcolor=blue, 	
    anchorcolor=black, 	
    citecolor=green, 	
    filecolor=cyan 	    
    urlcolor=magenta,
    pdfpagemode=FullScreen,
    }
\usepackage{graphicx}
\usepackage{tikz-cd}
\usetikzlibrary{positioning,arrows.meta,trees, shapes.geometric, fit, calc}

%% file: aux/macros_global.tex
\allowdisplaybreaks
\theoremstyle{plain}
\newtheorem{thm}{Theorem}[section]
\newtheorem{corollary}[thm]{Corollary}
\newtheorem{lemma}[thm]{Lemma}
\newtheorem{proposition}[thm]{Proposition}
\theoremstyle{definition}
\newtheorem{definition}[thm]{Definition}
\newtheorem{remark}[thm]{Remark}
\newtheorem{example}[thm]{Example}
\numberwithin{equation}{subsection}
\newcommand{\mf}{\mathfrak}
\newcommand{\mc}{\mathcal}
\newcommand{\ms}{\mathscr}
\newcommand{\op}{\operatorname}

\newcommand{\id}{\operatorname{id}}
\newcommand{\im}{\operatorname{im}}

\newcommand{\ev}{\operatorname{ev}}

\newcommand{\R}{\mathbb R}
\newcommand{\Z}{\mathbb Z}
\newcommand{\N}{\mathbb N}
\newcommand{\Nast}{{\N^{\ast}}}

\newcommand{\D}{\mathbb D}

\newcommand{\F}{\mathbb F}

\renewcommand{\Pr}{\mathrm{Pr}}
\newcommand{\e}{\varepsilon}
\newcommand{\p}{\varphi}

\newcommand{\inj}{\hookrightarrow}
\newcommand{\surj}{\twoheadrightarrow}

\newcommand{\E}{\mathbb{E}}

\newcommand{\bigmid}{~ \Big |~}

\newcommand{\A}{\mathbb{A}}

\newcommand{\T}{\mathbb T}

\newcommand{\Out}{\op{Out}}

\renewcommand{\P}{{\mathbb P}}

\newcommand{\X}{{\bm{\mathsf X}}}
\newcommand{\x}{{\bm{\mathsf x}}}

\newcommand{\Tr}{{\bm{\mathsf T}}}
\newcommand{\y}{{\bm{\mathsf y}}}

\renewcommand{\H}{{\bm{\mathsf H}}}
\newcommand{\h}{{\bm{\mathsf h}}}

%% file: aux/macros_local.tex
\makeatletter
\def\ps@pprintTitle{
  \let\@oddhead\@empty
  \let\@evenhead\@empty
  \let\@oddfoot\@empty
  \let\@evenfoot\@oddfoot
  }
\makeatother
\newcommand{\Today}{26th November 2024}
\usepackage{etoolbox}
\patchcmd{\MaketitleBox}{\footnotesize\itshape\elsaddress\par\vskip36pt}{\footnotesize\itshape\elsaddress\par\parbox[b][36pt]{\linewidth}{\vfill\hfill\textnormal{\Today}\hfill\null\vfill}}{}{}%
\patchcmd{\pprintMaketitle}{\footnotesize\itshape\elsaddress\par\vskip36pt}{\footnotesize\itshape\elsaddress\par\parbox[b][36pt]{\linewidth}{\vfill\hfill\textnormal{\Today}\hfill\null\vfill}}{}{}%

%% file: contents/frontmatter.tex
\begin{frontmatter}

\title{Decision making in stochastic extensive form II: \\Stochastic extensive forms and games}
\date{\today}
\author[1]{E.\ Emanuel Rapsch}
\ead{rapsch@math.tu-berlin.de}
\affiliation[1]{organization={Institut für Mathematik, Technische Universität Berlin},
                addressline={Straße des 17.\ Juni 136}, 
                city={10623 Berlin},
                country={Germany}
                }

\begin{abstract}
    A general theory of stochastic extensive forms is developed to bridge two concepts of information flow: decision trees and refined partitions on the one side, filtrations from probability theory on the other. Instead of the traditional ``nature'' agent, this framework uses a single lottery draw to select a tree of a given decision forest. Each ``personal'' agent receives dynamic updates from an own oracle on the lottery outcome and makes partition-refining choices adapted to this information. 
    This theory addresses a key limitation of existing approaches in extensive form theory, which struggle to model continuous-time stochastic processes, such as Brownian motion, as outcomes of ``nature'' decision making. Additionally, a class of stochastic extensive forms based on time-indexed action paths is constructed, encompassing a wide range of models from the literature and laying the groundwork for an approximation theory for stochastic differential games in extensive form.
\end{abstract}



\end{frontmatter}

%% file: contents/text.tex
\section*{Introduction}
\addcontentsline{toc}{section}{Introduction}

This paper is the second in a series aimed at developing a unified theory of stochastic games and decision problems in extensive form. Representing a decision problem in extensive form means specifying it in terms of what the author suggests calling ``extensive form characteristics'', namely, the flow of information about past choices and exogenous events, along with the set of choices available to decision makers. Although classical theory, as established by von Neumann and Morgenstern in \cite{Neumann1944} and furthered by Kuhn in \cite{Kuhn1950Extensive,Kuhn1953Extensive}, relies on strong finiteness assumptions, the concept itself is very general and broadly applicable. In a series of papers, including \cite{AlosFerrer2005,AlosFerrer2008,AlosFerrer2011Comment}, and in the monograph \cite{AlosFerrer2016}, Alós-Ferrer and Ritzberger develop an abstract, highly general theory of extensive form games and decision problems and give a concrete order-theoretic characterisation of its own boundaries. Beyond these boundaries, the notions of strategy, outcome, and equilibrium lose rigorous decision-theoretic meaning when applied to the extensive form characteristics of a given decision problem.

Hence, any decision problem and game exhibiting extensive form characteristics lies within these boundaries or is, at least, some sort of limit of objects within these boundaries. The inclusion, or the precise meaning of this limit, allows us to rigorously determine the decision-theoretic meaning of strategies, outcomes, equilibrium, etc.\ for the given decision problem. This opens a new perspective on problems in continuous time in particular -- a domain where the subtlety of this issue has long been recognised (see, e.g.\ \cite{Simon1989,Stinchcombe1992} and the references therein). Based on the maximality result in \cite{Stinchcombe1992}, systematically emphasised in \cite{AlosFerrer2008,AlosFerrer2011Comment}, which advocates restricting outcomes to certain piecewise constant paths, \cite{AlosFerrer2015} introduces, for the first time, a rigorous extensive form foundation of continuous-time games and decision problems involving such outcomes.

Having said that, an important question remains to be addressed: How can the ``extensive form characteristics'' of stochastic games, decision problems exposed to randomness, or those endowed with randomisation devices be modelled? The standard approach in game and decision theory is to introduce a ``nature'' agent executing a behaviour strategy (described in, e.g.\ \cite[Subsection~2.2.3]{AlosFerrer2005} or \cite{Fudenberg1991,AlosFerrer2016}, with roots in the works of Shapley \cite{Shapley1953} and Harsanyi \cite{Harsanyi1967,Harsanyi1968a,Harsanyi1968b}). However, a fundamental problem arises in the continuous-time case: many relevant modelling applications -- from economics (see, e.g.\ \cite{Riedel2017,Grenadier1996Strategic}) and finance (see \cite{Delbaen2006Mathematics,Pham2009Continuous,Bayer2023Rough,Lasry2007}) to engineering (see \cite{Cohen2023Optimal,Huang2006} and the references therein) and reinforcement learning (see, for example, \cite{Guo2023Reinforcement}) -- require exogenous noise represented by stochastic processes whose paths are anything but piecewise constant, with Brownian motion serving as a paradigmatic case (see, for instance, \cite{Pham2009Continuous,Karatzas1998Methods,Cohen2015Stochastic} for an overview of both theory and applications). By no means this implies ``the world to be Brownian'' or any similar assumption; rather, such models are widely used, be it because it appears reasonable to allow for unpredictably small time lags until the next exogenous information revelation, be it because certain quantities arguably evolve in decision paths of low regularity, be it for mathematical and computational convenience. One can certainly argue that these models exhibit ``extensive form characteristics'' without fitting into the Alós-Ferrer--Ritzberger framework (\cite{AlosFerrer2016}). This creates a fundamental issue: the lack of a rigorous decision-theoretic foundation of a large class of relevant models. 

Despite this, stochastic control and differential game theory, based on stochastic analysis, have found a pragmatic way to deal with the ``extensive form characteristics'' of a continuous-time decision problem (see, e.g.\ \cite{Cohen2015Stochastic,Carmona2018}). While extensive form theory is based on graph theory and refined partitions, stochastic analysis models exogenous information through filtrations on a given measurable space of scenarios. These are strictly less restrictive objects because $\sigma$-algebras need not be generated by partitions of the sample space. This works as if a one-shot lottery draw selects a scenario $\omega$ at the beginning, without communication, and then, as time progresses, more and more properties of $\omega$ become known to the decision makers. For instance, this could be knowledge about the realised path of a Brownian motion. From a decision-theoretic perspective, there is no disadvantage in using this weaker structure for modelling exogenous information. 

However, when it comes to the actual decision makers, the stochastic analysis-inspired approach is not compatible with the extensive form paradigm in the strict sense. This is because it does not explain strategies, randomisation, outcomes, or (dynamic, e.g.\ subgame-perfect) equilibria in terms of a decision tree-like object and choices locally available at moves. In this approach, ``strategies'' are typically defined as stochastic processes satisfying a certain measurability condition (e.g.\ progressive measurability) and such that a given stochastic differential equation, depending on them in a non-anticipating way, has a unique solution in some sense. This contradicts the problem's own extensive form characteristics in the strict sense, due to the \emph{ex post} restriction of the strategy space, which implies that the availability of a choice depends on future decisions (for a detailed discussion of this issue, see the introduction in \cite{AlosFerrer2015}). The ``outcome'' is taken to be the solution to the differential equation, and ``equilibria'' (or ``optimal controls'') are defined with implicit reference to the dynamic programming principle, but without an extensive form foundation. This not only fundamentally undermines the use of the equilibrium concept, but also introduces potential confusion, particularly regarding the decision-theoretic meaning of these ``equilibria'' (for example, concerning the notions of ``closed'' and ``open loop'' ``equilibria'' in stochastic differential games, see the discussion in \cite[p.\ 72--76]{Carmona2018}), as well as the precise definition of subgames (for example, in stochastic timing games, see \cite{Riedel2017} for further discussion).

Thus, there is room for an extensive form theory that models exogenous information through filtration-like objects, while modelling endogenous decision making with decision tree-like objects and partition-refining choices. This approach synthesises both frameworks in a productive way, which the author refers to as \emph{stochastic extensive forms}. This is the aim of this project, with the present paper forming its second part. The theory can be used to construct extensive form games involving general stochastic processes as noise. Moreover, it provides a way to precisely characterise how stochastic control and differential games can be approximated by stochastic extensive form decision problems. For instance, it applies to the stochastic timing game, which is typically not at all formulated in extensive form (see, in particular, \cite{Riedel2017}, where strategic form games are stacked according to a notion of ``consistency'', which is justified \emph{a priori} by an analogy to discrete time and \emph{a posteriori} in \cite{Steg2018} by a ``discrete time with an infinitesimal grid'' approximation argument reminiscent of \cite{Simon1987,Simon1989}). While a central motivation for this project stems from continuous-time problems, the theory is not confined to this domain. It is fundamentally about pushing the boundaries of extensive form theory as far as logically possible, with a focus set on probability. This makes it possible to formally establish links between uses of concepts such as equilibrium, which otherwise might seem ad hoc. It also allows for a clear outline of the decision-theoretic structure that arises naturally from the combination of these two ideas, distinguishing it from the additional structure sometimes added in modelling based on the modeller's discretion or specific perspectives on a given problem. For example, stochastic extensive form theory is independent of the representation of noise through a ``nature'' agent's virtual decision making, and in that sense, it is naturally stochastic. \smallskip

The first step in this endeavour, which is taken in the first paper \cite{Rapsch2024DecisionA} and briefly recalled in this paper for the reader's convenience, is to develop a theory of stochastic decision forests. These can be understood as forests of decision trees, where each tree corresponds to exactly one exogenous scenario, equipped with a similarity structure across trees, which identifies moves. Decision trees, as graph-theoretical objects, are the traditional base model for extensive form games and decision problems. However, as pointed out in \cite{AlosFerrer2005}, the refined partitions-based representation, using the set of all maximal chains, exhibits not only strong duality properties but also makes dynamic decision making amenable to the traditional decision-theoretic paradigm of choice under uncertainty: acts which assign a consequence to any state, following Savage's framework (\cite{Savage1972Foundations}). In this refined partitions approach, acts translate into strategies, mapping each move to a choice determined by a set of outcomes still possible at that move from the agent's perspective. This is how uncertainty about choices (whether past or future choices of oneself or others, and present choices of opponents) -- that is, endogenous information -- is handled. 
The modelling of uncertainty about exogenous noise, or exogenous information, in this paper, however, is fundamentally different. It is the aforementioned similarity structure, consisting in so-called random moves, that can be equipped with a filtration-like object which dynamically reveals information about the realised exogenous scenario. This allows for general stochastic processes to model noise without running into outcome generation problems along that dimension. The adaptedness of strategies with respect to exogenous information can be based on a concept of adapted choices.

In the present paper, the above-mentioned ingredients are combined to align with the modelling framework of choice under uncertainty, using refined partitions along decision paths and general probabilistic structure across exogenous scenarios. This yields, \emph{inter alia}, the notion of stochastic extensive forms, a strict generalisation of the corresponding concept in \cite{AlosFerrer2008,AlosFerrer2011Comment}. Strategies are simply Savage acts (\cite{Savage1972Foundations}), and randomisation occurs with respect to exogenous signals (following Aumann \cite{Aumann1974Subjectivity}). We demonstrate that the existence and uniqueness of induced outcomes for strategy profiles can be classified on a scenariowise level, so that the strong results in \cite{AlosFerrer2008} can be directly applied. We extend the concepts of sequential rationality and perfect Bayesian equilibrium to general well-posed stochastic extensive forms, showing that these concepts are not restricted to certain finite games or to specific formulations of more general situations. We conclude by discussing why stochastic extensive forms cannot be subsumed into the existing theory through the typical process of representing exogenous noise by means of a nature agent (as, especially, put forward by Harsanyi, in \cite{Harsanyi1967,Harsanyi1968a,Harsanyi1968b}).

The fundamental motivation for this theory arises in a large class of concrete applications. Thus, its presentation is tightly accompanied by examples. While some of them remain rather pedagogical, a general class of stochastic extensive forms based on time-indexed action paths is also constructed. A small set of readily verifiable axioms allows for many different time regimes, many different specifications of the outcomes (e.g.\ paths of timing games, also in the case of the scenario-dependent expiration of certain options), and general stochastic noise. Of course, it also includes the deterministic case. To the best of the author's knowledge, this theoretical unification of action path-based decision problems with extensive form characteristics within a single framework is a second new contribution of this paper in its own right.\smallskip

Before concluding the introduction, three remarks remain to be made. First, this paper is the second in a series of three and serves as the central part, introducing the key concepts from an abstract and general perspective. The main notions and results of the first papers are revisited as needed, with detailed references provided. Second, as pointed out by Aumann (see \cite{Aumann2020}), game theory arguably is ``interactive decision theory'', while decision theory is, in a trivial sense, single-player game theory. Still, game theory and decision theory are not the same; the former is more concerned with the ``interactive'' aspect, while the latter focuses more on the ``decision'' facet. This paper is more concerned with the latter, so it primarily employs the corresponding terminology. However, when the context allows, ``game'' is used instead of ``decision problem'', or vice versa, and similarly, the terms ``decision maker'', ``agent'', and ``player''  are used interchangeably. Finally, the relevant proofs of all new theorems, propositions, lemmata, and claims in examples can be found in the corresponding subsection of the appendix.

\section*{Notation}
\begin{itemize}[label=--]
    \item $\N = \Z_+$ = the positive integers including zero, $\Nast = \N \setminus \{0\}$, $\R$ = the real numbers, $\R_+ = \{x\in\R \mid x \ge 0\}$;
    \item $\mc P(A) = \mc P A$ = the set of subsets of a given set $A$, $\mc P(f) = \mc P f$ = the function $\mc P A \to\mc P B, M \mapsto \{ f(m) \mid m\in M\}$ for a given function $f\colon A \to B$ between two sets $A$ and $B$;\footnote{$\mc P$ defines a covariant endofunctor on the category of sets.}
    \item $\im f = (\mc P f)(D)$ = the image of a set-theoretic function $f\colon D \to V$;
    \item $|M|$ = if $M$ is not a number, then this is the cardinality of the set $M$;
    \item $\ms E|_D = \{A \cap D \mid A\in\ms E\}$, for any $\sigma$-algebra $\ms E$ and any $D\in\ms E$.
\end{itemize}

\section{Stochastic extensive forms}\label{sec:Stochastic extensive forms}

The fundamental object of classical extensive form decision and game theory is the decision tree (\cite{Neumann1944,Kuhn1950Extensive,Kuhn1953Extensive,AlosFerrer2016}). The stochastic generalisation constructed in the present series of papers replaces the traditional ``nature'' agent with a device ``randomly'' selecting a decision tree. As noted in the first part \cite{Rapsch2024DecisionA}, this implies that we consider decision forests rather than decision trees. In that paper, the former have been formalised using the refined partitions framework and made amenable to exogenous noise in the sense of general probability theory, and it has been shown that these are in a strong sense decision-theoretically natural. The most important notions of \cite{Rapsch2024DecisionA}, elementary in what follows, namely that of stochastic decision forests, exogenous information structures, and adapted choices are recalled in the first subsection. Then, based on this, the central concept of the current paper and arguably of the whole series, that of \emph{stochastic extensive forms}, is introduced. This makes it possible, in turn, to give a definition of strategies in the sense of Savage acts (\cite{Savage1972Foundations}) and their extension due to Anscombe and Aumann in \cite{Anscombe1963Definition} and Aumann in \cite{Aumann1974Subjectivity}, respectively. Towards the end of this section, the new concepts are illustrated by simple examples, including the absent-minded driver story.

\subsection{Preliminaries: Stochastic decision forests}\label{subs:SDF}

Let us start with recalling some basic definitions from graph and order theory, thereby fixing conventions used in this text, which combine those from \cite{Bollobas2013Modern,AlosFerrer2005,Davey2002}, and are identical to those from the first part \cite{Rapsch2024DecisionA}. In a partially ordered set (in short, \emph{poset}) $(N,\ge)$ a \emph{chain} is a subset $c\subseteq N$ such that for all $x,y\in c$, $x\ge y$ or $y\ge x$ holds true. A \emph{maximal chain} is a chain that is maximal as a chain with respect to set inclusion in $\mc P(N)$. $x\in N$ is called a \emph{maximal element} iff there is no $y\in N$ other than $x$ such that $y\ge x$. $x\in N$ is called \emph{maximum} iff for all $y\in N$, $x \ge y$. For $x\in N$, the \emph{principal up-set} $\uparrow x$ and \emph{principal down-set} $\downarrow x$ are defined by
\[ \uparrow x = \{y\in N \mid y\ge x\}, \qquad \downarrow x = \{y\in N \mid x \ge y\}. \]

Moreover, in this text, a poset $(F,\ge)$ is called a \emph{forest} iff for every $x\in F$, $\uparrow x$ is a chain. A forest $(F,\ge)$ is called \emph{rooted} iff $F\neq\emptyset$ and for every $x\in F$, $\uparrow x$ contains a maximal element of $(F,\ge)$. A forest $(T,\ge)$ is called a \emph{tree} iff for every $x,y\in T$, $(\uparrow x) \cap (\uparrow y) \neq \emptyset$. Given a forest $(F,\ge)$, the elements $x\in F$ are called \emph{nodes}. Nodes $x\in F$ such that $\downarrow x = \{x\}$ are called \emph{terminal}. We recall the following lemma, fundamental for what follows. It can actually be seen as an explicitly order-theoretic reformulation of a basic result from graph theory (see the discussion in \cite[section I.1]{Bollobas2013Modern}).

\begin{lemma}[Lemma 1.1 \cite{Rapsch2024DecisionA}]\label{lemma:partion_of_forest}
    For any forest $(F, \ge)$ there exists a unique partition $\mc F$ of $F$ into trees such that for all $x, y \in F$ with $x \ge y$ there is $T \in \mc F$ with $x, y \in T$. If $(F,\ge)$ is rooted, then for any $T\in\mc F$, $(T,\ge)$ is a rooted tree and has a maximum.
\end{lemma}

The elements of $\mc F$ are called \emph{connected components} of $(F,\ge)$. The maximum of a rooted tree $(T,\ge)$ is called the \emph{root}. The \emph{roots} of a forest $(F,\ge)$ are the roots of its connected components. A \emph{decision forest} (\emph{decision tree}) is a rooted forest (tree, respectively) $(F,\ge)$ such that all $x,y\in F$ with $x\neq y$ can be separated by some maximal chain $c\subseteq F$, that is, $c\cap \{x,y\}$ is a singleton. A \emph{move} in a decision forest $(F,\ge)$ is a non-terminal node $x\in F$. Following \cite{Rapsch2024DecisionA}, a decision forest is called $(F,\ge)$ \emph{(everywhere) non-trivial} iff some (any, respectively) root is a move.

If $V$ is a set, a \emph{$V$-poset} is a subset $N\subseteq \mc P(V)$. The name derives from the fact that $(N,\supseteq)$ defines a poset of subsets of $V$ ordered by set inclusion.

Note that the conventions used in this text are chosen intentionally in view of its objective. For instance, the definition of a tree used in this text is an order-theoretic transcription of a graph-theoretical concept. The meaning of ``tree'' and ``root'' and also the sense of the order ($\ge$ instead of $\le$) might differ slightly from those in other texts. See \cite[Remark 1.2]{Rapsch2024DecisionA} for a detailed explanation and comparison.\smallskip

Next, we recall the concept of decision forests from \cite{Rapsch2024DecisionA}, defined as a transcription of a characterisation of (a subclass of) so-called ``game trees'' in \cite[Theorem~3]{AlosFerrer2005} (see \cite[Subsection~1.2]{Rapsch2024DecisionA} for a detailed comparison).

\begin{definition}[Definition 1.3 in \cite{Rapsch2024DecisionA}]\label{def:decision_forest}
    Let $V$ be a set. A \emph{decision forest on $V$} is a $V$-poset $F$ such that:
    \begin{enumerate}
        \item\label{def:decision_forest.rooted_forest} $(F,\supseteq)$ is a rooted forest;
        \item\label{def:decision_forest:repr_by_dec_paths} $F$ is \emph{its own representation by decision paths}, that is, if $W$ denotes the set of maximal chains in $(F,\supseteq)$, and for every $y\in F$, $W(y) = \{w\in W \mid y\in w\}$, then there is a bijection $f\colon V \to W$ such that for every $y\in F$, $(\mc P f)(y) = W(y)$.
    \end{enumerate}
    $F$ is called \emph{decision tree on $V$} iff, in addition, for all $x,y\in F$ there is $z\in F$ with $z\supseteq x\cup y$.

    The nodes, terminal nodes, and moves of $(F,\supseteq)$ are also called \emph{nodes}, \emph{terminal nodes}, and \emph{moves} of $F$, respectively, and the elements of $V$ are called \emph{outcomes}. The set of moves of $F$ is denoted by $X(F)$ or $X$ in short.
\end{definition}

Decision forests on sets can be seen as a complete, faithful, and self-dual representation of order-theoretic decision forests in terms of refining partitions of the set of outcomes. They are self-dual in the sense that outcomes and maximal chains of nodes (called decision paths in \cite{Rapsch2024DecisionA} and plays in \cite{AlosFerrer2005}) can be uniquely identified. Actually, there is one and only one $f$ as in the definition, and it satisfies, for all $x\in F$ and $v\in V$:
\[ x\in f(v) \qquad \Longleftrightarrow\qquad v\in x. \]
In that sense, we can identify $V$ and $W$, and for that reason, we use the notation $W$ for the set underlying $F$ in the sequel (as in \cite{AlosFerrer2005}).

We note that although phrased differently and with a different aim, the duality and representation concepts just reviewed constitute one of the essential innovations of \cite{AlosFerrer2005}. \cite{Rapsch2024DecisionA} presents them from the perspective of (stochastic) decision forests and therefore additionally raises the question whether, put simply, a forest of decision trees on sets is the same thing as a decision forest on the union of these sets. The answer is shown to be affirmative in \cite[Theorem~1.7]{Rapsch2024DecisionA}. As a consequence, the analysis of decision forests as such (abstracting from possible supplementary structure like the one introduced in stochastic decision forests) can be reduced to the analysis of decision trees (compare \cite[Remark 1.8]{Rapsch2024DecisionA}).\smallskip

We now come to recall stochastic decision forests. The main idea behind this is to weaken the traditional assumptions on exogenous information which is no more assumed to arise through dynamic decision making of a nature agent. Rather, an exogenous scenario is ``randomly'' selected within a given measurable space which determines the decision tree underlying the actual decision maker's problem. From a decision-theoretic point of view, it is indeed very meaningful to define an \emph{exogenous scenario space} as a measurable space $(\Omega,\ms E)$ with $\Omega\neq\emptyset$, in the sense of measure theory (see \cite[Subsection~2.1]{Rapsch2024DecisionA} for a detailed argument, \cite{Bogachev2007Measure} for a textbook on measure theory). However, decision makers need not know in which tree they are (a piece of exogenous information) while making choices based on knowledge about the current (``information set'' of) move(s) (a piece of endogenous information). As a consequence, there is a need of a structure of similarity among trees that can serve as a consistent basis for both exogenous and endogenous information revelation. This led to the following definitions.

\begin{definition}[Definition 2.1 in \cite{Rapsch2024DecisionA}]\label{def:sdf}
    A \emph{stochastic decision forest}, in short \emph{\textsc{sdf}}, on an exogenous scenario space $(\Omega,\ms E)$ is a triple $(F,\pi,\X)$ consisting of:
    \begin{enumerate}
        \item\label{def:sdf.df} a decision forest $F$ on some set $W$;
        \item\label{def:sdf.conn_comp} a surjective map $\pi\colon F\to \Omega$ such that the set $\mc F$ of connected components of $(F,\supseteq)$ is given by the fibres of $\pi$, that is,
        \[ \mc F = \{\pi^{-1}(\{\omega\}) \mid \omega\in \Omega\}; \]
        \item\label{def:sdf.X} a set $\X$ such that: 
        \begin{enumerate}
            \item\label{def:sdf.X.section} any element $\x\in\X$ is a section of moves defined on some non-empty event, that is, it is a map $\x\colon D_\x \to X$ satisfying $\pi\circ\x = \id_{D_\x}$ for some $D_\x\in\ms E\setminus\{\emptyset\}$;
            \item\label{def:sdf.X.cov} $\X$ induces a covering of $X$, that is, $\{\x(\omega) \mid \x\in\X,\,\omega\in D_\x\} = X$.
        \end{enumerate}
    \end{enumerate}
    The elements of $\X$ are called \emph{random moves}. For $\omega\in\Omega$, let $T_\omega = \pi^{-1}(\{\omega\})$ and $W_\omega$ be the root of $T_\omega$. For $E\subseteq\Omega$, let $W_E = \bigcup_{\omega\in E} W_\omega$ and $F_E = \bigcup_{\omega\in E} T_\omega$.
\end{definition}

\begin{definition}[Definition 2.2 in \cite{Rapsch2024DecisionA}]\label{def:sdf.addon}
    Given a stochastic decision forest $(F,\pi,\X)$ on an exogenous scenario space $(\Omega,\ms E)$, let $\ge_\X$ denote the partial on $\X$ defined by
    \[ \x_1 \ge_\X \x_2 \quad \Longleftrightarrow \quad \Big[ D_{\x_1} \supseteq D_{\x_2}~ \text{ and }~\forall \omega\in D_{\x_2}\colon \x_1(\omega) \supseteq \x_2(\omega)\Big]. \]
    A set $\tilde\X\subseteq\X$ of random moves is said \emph{order consistent} iff for all $\x_1,\x_2\in\tilde\X$:
    \begin{equation*}
        \Big[\exists\omega\in D_{\x_1}\cap D_{\x_2}\colon~ \x_1(\omega) \supseteq \x_2(\omega)\Big] \qquad \Longrightarrow \qquad \x_1 \ge_\X \x_2.
    \end{equation*}

    A stochastic decision forest $(F,\pi,\X)$
    \begin{enumerate}[label=4(\alph*),ref=4(\alph*)]
        \item\label{def:sdf.X.OC} is said \emph{order consistent} iff $\X$ is order consistent;
        \item\label{def.sdf.X.surely_NT} is said \emph{surely non-trivial} iff $(F,\supseteq)$ is everywhere non-trivial;
        \item\label{def:sdf.X.max} that is order consistent, is said \emph{maximal} iff for every set $\bar\X$ such that $(F,\pi,\bar\X)$ is an order consistent stochastic decision forest and that is \emph{refined by $\X$} in that for all $\bar\x\in\bar\X$ there is $P_{\bar\x}\subseteq \X$ with $\bar\x = \bigcup P_{\bar\x}$,\footnote{According to standard set-theoretic conventions, $\bar\x = \bigcup P_{\bar\x}$ means: $\bar\x$ is a map with domain $\bigcup_{\x\in P_{\bar\x}} D_\x$ and for all $\x\in P_{\bar\x}$ and $\omega\in D_\x$, $\bar\x(\omega) = \x(\omega)$.} we have $\bar\X = \X$.
    \end{enumerate}
\end{definition}

For a detailed discussion, the reader is referred to \cite[Subsection~2.2]{Rapsch2024DecisionA}. Here, only the following two results are briefly recalled. For this, one extends $(\X,\ge_\X)$ by letting $\Tr$ be the set $\X$ augmented with all maps $\y\colon D_\y = \{\omega\}\to \{\{w\}\}$ for all $(\omega,w)\in\Omega\times W$ such that $\{w\}\in T_\omega$, and letting, for all $\y,\y'\in\Tr$, $\y\le_\Tr \y'$ iff $D_{\y} \supseteq D_{\y'}$ and all $\omega\in D_{\y'}$ satisfy $\y(\omega)\supseteq\y'(\omega)$ (see \cite[Subsection~2.2]{Rapsch2024DecisionA}). Then, first, in the order consistent case, the evaluation map on $\Tr\bullet\Omega = \{(\y,\omega) \in\Tr \times \Omega \mid \omega\in D_\y\}$, equipped with the order induced by the product of $\ge_\Tr$ on $\Tr$ and equality on $\Omega$, defines an order-isomorphism onto $(F,\supseteq)$ (\cite[Proposition~2.4]{Rapsch2024DecisionA}). Second, in the order consistent, surely non-trivial, and maximal case, $(\Tr,\ge_\Tr)$ is a decision tree in its own right and the set of its moves equals the set $\X$ of random moves of $F$ (\cite[Theorem~2.5]{Rapsch2024DecisionA}).  See \cite[Subsections~2.2, 2.3, 2.4]{Rapsch2024DecisionA} for a detailed discussion, including illustrations and examples.\smallskip

Random moves provide the basis for exogenous information revelation. The crucial modelling idea is based on adapting the concept of filtrations from probability theory, which provide a different and in probabilistic hindsight far more flexible way of generalising refined partitions on countable sets of a nature agent's ``choices'' (see \cite[Subsections~3.1, 3.2]{Rapsch2024DecisionA}). The corresponding definition is recalled next.

\begin{definition}[Definition 3.1 in \cite{Rapsch2024DecisionA}]\label{def:EIS}
    Let $(F,\pi,\X)$ be a stochastic decision forest on an exogenous scenario space $(\Omega,\ms E)$ and let $\tilde\X\subseteq\X$. An \emph{exogenous information structure on $\tilde\X$} is a family $\ms F = (\ms F_\x)_{\x\in\tilde\X}$ such that for all $\x\in\tilde\X$, $\ms F_\x$ is a sigma-algebra on $D_\x$ with $\ms F_\x\subseteq\ms E$. An exogenous information structure $\ms F$ is said to admit \emph{recall} iff for all $\x'\in\tilde\X$ with $\x \ge_\X \x'$ and every $E\in \ms F_\x$, we have $E\cap D_{\x'} \in \ms F_{\x'}$.
\end{definition}

For further explanation and examples, see \cite[Section~3]{Rapsch2024DecisionA}.\smallskip

The combination of stochastic decision forests and exogenous information structures makes it possible to define a concept of choices that is able to reconcile the classical decision-theoretic model of choice under uncertainty in its dynamic interpretation studied in \cite{AlosFerrer2016} on the one hand and on the other the general theory of stochastic processes. In short, choosing remains an act of refining partitions of global consequences (alias outcomes), but must be compatible with random moves, the exogenous information revealed at these and a given system of reference choices. This implements the basic principle from extensive form theory that at any ``move'' it is ``known'' to decision makers whether a given ``action'' is available to them or not.

One of the important facts the refined partitions approach formalised in \cite{AlosFerrer2005} and subsequent papers clarifies, is that the availability of a choice at a given move can be completely described in terms of the underlying set-theoretic structure: A choice is available at a move iff the latter is an immediate predecessor of the former. More precisely, if $(F,\pi,\X)$ is a stochastic decision forest on an exogenous scenario space $(\Omega,\ms E)$, $W=\bigcup F$ and $c\subseteq W$ is some subset (for instance, a union of nodes representing a choice), then, with
\[ \downarrow c = \{ x\in F \mid c \supseteq x \}, \]
let, as in the classical setting of \cite{AlosFerrer2005}, the set of \emph{immediate predecessors of $c$} be defined by:
\[ P(c) = \{x\in F \mid \exists y \in \downarrow c\colon \uparrow x = \uparrow y \setminus \downarrow c \}. \]

We now recall the definition of reference choices and adapted choices. For this, let us recall the definition going back to \cite{AlosFerrer2005} and reformulated in the setting of stochastic decision forests in \cite{Rapsch2024DecisionA} of a choice: In a stochastic decision forest, a \emph{choice} is a non-empty union of nodes.

\begin{definition}[Definition 4.2 in \cite{Rapsch2024DecisionA}]
    Let $(F,\pi,\X)$ be a stochastic decision forest on an exogenous scenario space $(\Omega,\ms E)$, let $\tilde \X\subseteq\X$, and let $\ms F$ be an exogenous information structure on $\tilde \X$.
    \begin{enumerate}
        \item A choice is said
        \begin{enumerate}
            \item \emph{non-redundant} iff for any $\omega\in\Omega$ with $P(c) \cap T_\omega = \emptyset$, we have $c\cap W_\omega = \emptyset$;
            \item \emph{$\tilde\X$-complete} iff for every random move $\x\in \tilde\X$, $\x^{-1}(P(c))$ is either empty or equal to $D_\x$;
            \item \emph{complete} iff it is $\X$-complete.
        \end{enumerate}
        \item For any random move $\x\in\X$, a choice $c$ is said \emph{available at $\x$} iff $\x^{-1}(P(c)) = D_\x$.
        \item A \emph{reference choice structure on $\tilde\X$} is a family $\ms C = (\ms C_\x)_{\x\in\tilde\X}$ of sets $\ms C_\x$ of non-redundant and $\tilde\X$-complete choices available at $\x$.
        \item Let $\ms C$ be a reference choice structure on $\tilde\X$. An \emph{$\ms F$-$\ms C$-adapted choice} is a non-redundant and $\tilde\X$-complete choice $c$ such that for all $\x\in\tilde\X$ that $c$ is available at and all $c'\in\ms C_\x$:
        \[ \x^{-1}(P(c \cap c')) = \{\omega\in D_\x \mid \x(\omega) \in P(c\cap c') \} \in \ms F_\x. \]
    \end{enumerate}
\end{definition}

For a detailed motivation and discussion of this definition, including examples, see \cite[Section~4]{Rapsch2024DecisionA}.

\subsection{Definition: Stochastic extensive forms}\label{subs:SEF}

The first paper, summarised in the preceding subsection, is about a model of stochastic decision forests that consistently combines exogenous and endogenous information flow and allows for a notion of choices adapted to exogenous information. While exogenous information is modelled through an additional structure, endogenous information is given in terms of the sets of immediate predecessors. In that sense, choices are ``adapted to'' endogenous information by construction. Given a stochastic decision forest $(F,\pi,\X)$, a set $I$ of agents, and families $\ms F = (\ms F^i)_{i\in I}$, $\ms C = (\ms C^i)_{i\in I}$ and $C = (C^i)_{i\in I}$ of exogenous information structures, reference choice structures, and sets of adapted choices, respectively, additional criteria are required in order to define an extensive form. First, choices must be partition refining along the trees: alternative future nodes must be separable by alternative choices and alternative choices must be disjoint, in any scenario; endogenous information sets must be disjoint; any possible future node must be compatible with some choice. Second, endogenous and exogenous information must be compatible: exogenous information must be identical across endogenous information sets. Finally, as we allow for multiple agents to act ``simultaneously'' at the same move (as in \cite{AlosFerrer2005,AlosFerrer2008,AlosFerrer2011Comment}), any admissible combination of choices must be compatible with at least one outcome.

The notation and definition that follow are generalisations of those from \cite{AlosFerrer2005,AlosFerrer2008,AlosFerrer2011Comment}. While the latter are based on (what we call) decision trees, the former are based on (more general) stochastic decision forests.

Let $I$ be a set. If $C = (C^i)_{i\in I}$ is a family of sets of choices on some stochastic decision forest $(F,\pi,\X)$, then for any move $x\in X$, any random move $\x\in \X$, and any $i\in I$,
$$ A^i(x) = \{c\in C^i \mid x\in P(c) \}, \qquad A^i(\x) = \{c\in C^i \mid \x^{-1}(P(c)) \neq \emptyset \} $$
are the sets of choices in $C^i$ \emph{available at} $x$, $\x$, respectively.  The notation $A$, commonly used and linked to the term ``action'', is discussed later in this subsection.

For any move $x\in X$, any random move $\x\in\X$, let
$$ J(x) = \{ i\in I \mid A^i(x) \neq \emptyset\}, \qquad J(\x) = \{i\in I \mid A^i(\x) \neq \emptyset\}. $$
For any $i\in I$, let 
\[ X^i = \{x\in X \mid i \in J(x) \}, \qquad \X^i = \{\x\in\X \mid i \in J(\x)\}. \]
Let $\X^i \bullet \Omega = \{(\x,\omega)\in\X^i\times \Omega\mid \omega\in D_\x\}$. 

Clearly, if all $c\in C^i$ are $\X^i$-complete, then we have $A^i(\x) = \{c\in C^i \mid \x^{-1}(P(c)) = D_\x\} = A^i(\x(\omega))$ and $J(\x) = J(\x(\omega))$ for all $\x\in\X^i$ and $\omega\in D_\x$. It is also clear that if $c\in C^i$ is $\X^i$-complete, then it is already complete.

\begin{definition}\label{def:SEF}
    Let $(\Omega,\ms E)$ be an exogenous scenario space. A \emph{stochastic pseudo-extensive form}, in short: \emph{$\psi$-\textsc{sef}}, \emph{on $(\Omega,\ms E)$}, is a tuple $\F = (F,\pi,\X,I,\ms F, \ms C,C)$ such that $(F,\pi,\X)$ is a stochastic decision forest on $(\Omega,\ms E)$, $I$ is a set, $C=(C^i)_{i\in I}$ is a family of sets of choices, $\ms F = (\ms F^i)_{i\in I}$ is a family of exogenous information structures on $\X^i$, $i\in I$, $\ms C = (\ms C^i)_{i\in I}$ is a family of reference choice structures on $\X^i$, $i\in I$, such that all elements of $C^i$ are $\ms F^i$-$\ms C^i$-adapted and the evaluation map $\X^i\bullet\Omega\to X$ is injective, for all $i\in I$, and satisfying the following axioms:
    \begin{enumerate}
        \item\label{def:SEF.P(c)} For all $i\in I$, all $c,c'\in C^i$ such that $P(c) \cap P(c') \neq \emptyset$ we have $P(c) = P(c')$, and for all $\omega\in\Omega$, we have either $c \cap W_\omega = c' \cap W_\omega$ or $c\cap c' \cap W_\omega = \emptyset$.
        \item\label{def:SEF.outcomes_faithful} For all $x\in X$ and all $(c^i)_{i\in J(x)} \in \bigtimes_{i\in J(x)} C^i$, we have
        $$ x\cap\bigcap_{i\in J(x)} c^i \neq \emptyset. $$
        \item\label{def:SEF.weak_separation} For all $y,y'\in F$ with $\pi(y) = \pi(y')$ and $y\cap y' = \emptyset$ there are $i\in I$ and $c,c'\in C^i$ such that $y\subseteq c$, $y\subseteq c'$ and $c\cap c' \cap W_{\pi(y)} = \emptyset$.
        \item\label{def:SEF.enough_choices} For all $x\in X$, all $i\in J(x)$, all $y\in \downarrow x\setminus\{x\}$, there is $c\in A^i(x)$ with $c\supseteq y$.
        \item\label{def:SEF.endo_exo_compatible} For all $i\in I$, $\x,\x'\in \X^i$ such that $A^i(\x) \cap A^i(\x') \neq \emptyset$, we have $\ms F^i_\x = \ms F^i_{\x'}$ and $\ms C^i_\x = \ms C^i_{\x'}$.
        \item\label{def:SEF.choice_completeness} For all $i\in I$, all $\ms F^i$-$\ms C^i$-adapted choices $c'$ such that
        \begin{enumerate}[label=(\roman*),ref=\theenumi{}(\roman*)]
            \item\label{def:SEF.choice_completeness.i} any $\omega\in\Omega$ with $c'\cap W_\omega \neq \emptyset$ admits $c\in C^i$ with $c'\cap W_\omega = c \cap W_\omega$,
            \item\label{def:SEF.choice_completeness.ii} and there is $c\in C^i$ with $P(c') = P(c)$,
        \end{enumerate}
         are choices for $i$, that is, satisfy $c'\in C^i$.
    \end{enumerate}
    The elements of $I$ are called \emph{agents} or \emph{decision makers}. For each agent $i\in I$, $\ms F^i$ is called \emph{$i$'s exogenous information structure}, $\ms C^i$ is called \emph{$i$'s reference choice structure}, and the elements of $C^i$ are called \emph{$i$'s choices}. For $i\in I$ and $x\in X$ ($\x\in \X$) $i$ is said \emph{active at $x$ ($\x$, respectively)} iff $i\in J(x)$ ($i\in J(\x)$, respectively). For any agent $i\in I$, the sets $P(c)$, $c\in C^i$, are called \emph{immediate predecessor sets of $i$'s choices}. For any agent $i\in I$, the elements of $X^i$ and $\X^i$ are called \emph{$i$'s moves} and \emph{$i$'s random moves}, respectively.

    If $\F$ is a stochastic pseudo-extensive form on $(\Omega,\ms E)$, then its items are typically denoted by
    \[ \F = (F,\pi,\X,I,\ms F,\ms C,C). \]

    A \emph{stochastic extensive form}, in short: \emph{\textsc{sef}}, \emph{on $(\Omega,\ms E)$}, is a stochastic pseudo-extensive form satisfying the following stronger separation axiom:
    \begin{enumerate}[label=3'.,ref=3']
        \item\label{def:SEF.separation} For all $y,y'\in F$ with $\pi(y) = \pi(y')$ and $y\cap y'= \emptyset$, there are $x\in X$, $i\in I$ and $c,c'\in C^i$ with $x\cap c\supseteq y$, $x\cap c'\supseteq y'$, $c\cap c' \cap W_{\pi(y)} = \emptyset$ and $x\in P(c) \cap P(c') \cap T_{\pi(y)}$. 
    \end{enumerate}

    A \emph{classical (pseudo-) extensive form} is the data $(F,I,C)$ for a stochastic (pseudo-) extensive form $\F$ on the singleton exogenous scenario space, respectively. 
\end{definition}

In other words, a stochastic pseudo-extensive form specifies a stochastic decision forest $(F,\pi,\X)$, a set of agents $I$, and for each agent $i\in I$, a ``dynamically updating oracle'' (exogenous information structure) $\ms F^i$ along $i$'s random moves, a set of reference choices $\ms C^i$ describing how $i$ can measure choices locally at each of $i$'s random moves, and a set $C^i$ of choices adapted to this data satisfying six axioms. Five of these axioms have already been motivated roughly, and all of them are discussed in more detail in the sequel. 

At this point, we make a first comparison with the notion of the ``extensive decision problem'' in the sense of \cite{AlosFerrer2005} and of the ``extensive form'' in the sense of \cite{AlosFerrer2011Comment}, reproduced and further developed in the monographic version in \cite[Definition 5.2, p.\ 118]{AlosFerrer2016}. Indeed, upon consulting the latter reference it becomes evident that, respectively, a triple $(T,I,C)$ is a classical (pseudo-) extensive form in the sense of the preceding definition iff $(T,C)$ is an (``extensive decision problem'') ``extensive form'' with set of ``players'' $I$ according to (\cite[Definition 4.1]{AlosFerrer2016}) \cite[Definition 5.2, p.\ 118]{AlosFerrer2016}, the tree $(T, \supseteq)$ is rooted, and $C = (C^i)_{i\in I}$ is such that for all $i\in I$ and all $c\in C^i$, $P(c) \neq \emptyset$.
In that sense, the concept of \cite[Definition 5.2]{AlosFerrer2016}-``extensive forms'' is naturally equivalent to the concept of classical extensive forms, which is naturally embedded into the concept of stochastic extensive forms. An analogous statement holds true for \cite[Definition 4.1]{AlosFerrer2016}-``extensive decision problems'' and classical pseudo-extensive forms. 

In the present text, the term ``pseudo-extensive forms'' is preferred over that of ``extensive decision problems'' because the latter, just as the term ``extensive (form) game'', can easily be understood to include a given preference relation on outcomes. This, however, is not the case and we wish to avoid any confusion about this. In this framing, pseudo-extensive forms as defined above, whether stochastic or classical, provide a form describing all possibilities of evolution, the information agents have about it, and what they can do. But to make a decision problem or game out of it, individual preferences about outcomes must be added. While the stronger separability property of extensive forms makes them more relevant in the end, their relaxed version (whatever its name) has been found to be important for understanding the problem of well-posedness in the classical case in \cite{AlosFerrer2008,AlosFerrer2011Comment}. This is why here as well both versions are introduced, but the naming is chosen such as to underline the importance of the stronger version.

The term ``classical'' is used instead of ``deterministic'' which, at first sight, might seem to be more compelling. This is because a certain class of stochastic (pseudo-) extensive forms can be represented as classical (pseudo-) extensive forms, as is discussed in Section~\ref{sec:well-posedness_equilibrium}. In the case of classical extensive forms, Axioms \ref{def:SEF.endo_exo_compatible} and~\ref{def:SEF.choice_completeness} become redundant, and the formulation of the remaining four axioms slightly shorter. The interpretation of the first four axioms in the general stochastic case, as discussed next, therefore resembles the discussion in \cite{AlosFerrer2008,AlosFerrer2011Comment,AlosFerrer2005} (see \cite{AlosFerrer2016} for a monographic treatment). 

Axiom \ref{def:SEF.enough_choices} means that at any move and for any possible future node, any active agent can choose not to discard it. Put less rigorously, it ensures that at any move any future node is compatible with some choice which is the last part of the refined partitions model. Axiom \ref{def:SEF.weak_separation} means that any pair of non-consecutive nodes from the same exogenous scenario $\omega$ can be separated by a pair of choices disjoint on $W_\omega$ and jointly available to one and the same agent, while Axiom \ref{def:SEF.separation} in addition requires this to be possible jointly at one and the same, and preceding move in $T_\omega$. Put less rigorously, both axioms ensure that alternative future nodes are separable by alternative choices, but to different extents. Axiom \ref{def:SEF.outcomes_faithful} means that at any move, all profiles of admissible choices by active agents are compatible with at least one outcome. This is a minimal requirement on the decision forest $F$ to be a faithful descriptor of outcomes, and is linked to the possibility that multiple agents can choose at once. This modelling ansatz, though non-standard compared to the historic literature, has been pursued in \cite{AlosFerrer2005,AlosFerrer2008,AlosFerrer2011Comment,AlosFerrer2016}. We adopt this convention because it simplifies the presentation of interactive settings.

\subsection{Information sets}

For the understanding of Axioms~\ref{def:SEF.P(c)} and~\ref{def:SEF.endo_exo_compatible}, let us note the following.
\begin{proposition}\label{prop:information_sets}
    Let $\F$ be a stochastic pseudo-extensive form and $i\in I$ an agent.
    \begin{enumerate}
        \item\label{prop:information_sets.P(c)_partition} $\{P(c) \mid c\in C^i\}$ is a partition of $X^i$.
        \item\label{prop:information_sets.A(x)_partition} $\{A^i(x) \mid x\in X^i\}$ is equal to $\{ A^i(\x) \mid \x \in \X^i\}$ and is a partition of $C^i$.
        \item\label{prop:information_sets.P(c)=P(c')} For all $c,c'\in C^i$, we have $P(c) = P(c')$ iff there is $x\in X$ with $c,c'\in A^i(x)$.
        \item\label{prop:information_sets.A(x)=A(x')} For all $x,x'\in X$, we have $A^i(x) = A^i(x')$ iff there is $c\in C^i$ with $x,x'\in P(c)$.
        \item\label{prop:information_sets.exists_mfP} There is a unique partition ${\mf P^i}$ of $\X^i$ such that for all $\x,\x'\in \X^i$ we have $A^i(\x) = A^i(\x')$ iff there is ${\mf p}\in{\mf P^i}$ such that $\x,\x'\in{\mf p}$.
        \item\label{prop:information_sets.Bij_mfP_P(c)} The assignment \[{\mf P^i} \ni {\mf p} \mapsto \bigcup_{\x\in\mf p} \im \x\] defines a bijection ${\mf P^i} \to \{P(c) \mid c\in C^i\}$.
        \item\label{prop:information_sets.msC_msF_const_on_mfP} For all ${\mf p}\in{\mf P^i}$, all $\x,\x'\in{\mf p}$, we have $D_\x = D_{\x'}$, $\ms C^i_\x = \ms C^i_{\x'}$, and $\ms F^i_\x = \ms F^i_{\x'}$.
    \end{enumerate}
\end{proposition}

\begin{remark}\label{rmk:prop_information_sets}
    In the proof of the preceding proposition in Section~\ref{subsec:appendix.proofs.1}, a bit more is shown. Namely, let $\F = (F,\pi,\X,I,\ms F,\ms F,C)$ be a tuple as in Definition~\ref{def:SEF} satisfying Axioms~\ref{def:SEF}.$k$, $k=1,\dots,5$, but not necessarily Axiom~\ref{def:SEF}.\ref{def:SEF.choice_completeness}. Then the conclusions of Proposition~\ref{prop:information_sets} hold true.
\end{remark}

Thus, the immediate predecessor sets partition $i$'s moves, and the sets of available choices partition $i$'s choices. Two choices have identical immediate predecessor sets iff they are available at a common move; and two moves $i$ is active at have identical available choices iff they are jointly immediate predecessors to one and the same choice. Moreover, immediate predecessor sets as well as the preceding statements can equivalently be formulated on the level of random moves, which gives rise to a model of endogenous information sets for stochastic extensive forms. This model is consistent with respect to exogenous information in that at two random moves belonging to the same endogenous information set ${\mf p}$ the same exogenous information is revealed. 

As a reaction to Proposition~\ref{prop:information_sets}, let, for ${\mf p}\in{\mf P^i}$ and $\x\in{\mf p}$:
$$ A^i({\mf p}) = A^i(\x), \qquad D_{\mf p} = D_\x, \qquad \ms F^i_{\mf p} = \ms F^i_\x, \qquad \ms C^i_{\mf p} = \ms C^i_\x. $$
Previous definitions about moves can be lifted accordingly. For instance, we call $c$ \emph{available} at an endogenous information set $\mf p$ iff $c \in A^i(\mf p)$. 

To fix names for the discussed concept of information, let us note that in stochastic pseudo-extensive forms, information is revealed along two different channels: exogenous information is revealed via $\ms F$, endogenous information (about agents' behaviour) is revealed via the position in $(\X,\ge_\X)$. Therefore, it is natural to decompose the property of perfect recall accordingly.

\begin{definition}
    Let $\F$ be a stochastic pseudo-extensive form and $i\in I$ be an agent. 
    \begin{enumerate}
        \item The elements of $\mf P^i$ are called \emph{$i$'s endogenous information sets}. 

        \item The set of choices of agent $i$ is said to admit and agent $i$ itself is said to have \emph{perfect endogenous information} iff all $\mf p\in \mf P^i$ are singletons and for all $j\in I\setminus\{i\}$, all $\x\in\X^i$, $\x'\in\X^j$, we have $\im\x \cap\im\x' = \emptyset$. Agent $i$ is said to have \emph{perfect exogenous information} iff $\ms F^i_\x = \ms E|_{D_\x}$ for all $\x\in\X^i$. Agent $i$ is said to have \emph{perfect information} iff $i$ has both perfect endogenous and exogenous information, and $\F$ is said to be of \emph{perfect information} iff this holds true for all $i\in I$.

        \item The set of choices of agent $i$ and agent $i$ itself are said to admit \emph{perfect endogenous recall} iff all $c,c'\in C^i$ and $\omega\in\Omega$ with $c\cap c'\cap W_\omega \neq \emptyset$ satisfy $c\cap W_\omega \supseteq c'\cap W_\omega$ or $c\cap W_\omega \subseteq c'\cap W_\omega$. Agent $i$ is said to admit \emph{perfect exogenous recall} iff ${\ms F}^i$ admits recall. Agent $i$ is said to admit \emph{perfect recall} iff $i$ admits both perfect endogenous and exogenous recall, and $\F$ is said so iff this holds true for all $i\in I$.
    \end{enumerate}
\end{definition}

Perfect recall with respect to endogenous information is defined by the requirement that for two choices available at moves along a given decision path the earlier one cannot condition on less endogenous information (this criterion is compatible with many classical definitions of perfect recall for a large class of classical extensive forms, see \cite[Subsections~6.4.1, 6.4.2]{AlosFerrer2016}). Perfect information with respect to endogenous information is essentially defined in the classical way, namely, by requiring information sets to be minimally small, among and across agents. Note, however, that we formulate the notion both for individual agents and the stochastic (pseudo-) extensive form as a whole. Perfect information with respect to exogenous information is defined analogously. However, it is not a very interesting case as no ``nature'' agent is supposed to act dynamically. The stochastic component of stochastic (pseudo-) extensive forms is relevant just because there may be agents having imperfect exogenous information (about the realised scenario). For the sake of an illustration of the above-defined notions it is shown in Lemma~\ref{lemma:perfect_endo_information_implies_perfect_endo_recall} that, as to be expected, perfect information implies perfect recall.

Perfect endogenous recall can be analysed along the lines of \cite{Ritzberger1999Recall}. Perfect recall, however, has a richer structure in the present setting with general sigma-algebras.
We also note that in stochastic extensive forms, no information set can be visited twice by a given decision path (as in the classical case discussed in \cite{AlosFerrer2005} and \cite{AlosFerrer2016}), which the author proposes to call the \emph{Heraclitus property}:

\begin{lemma}[Heraclitus Property]\label{lemma:Heraclitus_property}
   Let $(F,\pi,\X,I,\ms F,\ms C,C)$ be a stochastic pseudo-extensive form on an exogenous scenario space $(\Omega,\ms E)$. Then we have for all agents $i\in I$:
   \begin{enumerate}
       \item\label{lemma:Heraclitus_property.X} for all $x,x'\in X$ with $A^i(x) \cap A^i(x') \neq \emptyset$ and $x\supseteq x'$ we have $x=x'$;
       \item\label{lemma:Heraclitus_property.rmX} for all $\x,\x'\in\X$ with $A^i(\x) \cap A^i(\x') \neq \emptyset$ and $\x \ge_\X \x'$ we have $\x = \x'$.
   \end{enumerate}
\end{lemma}

The possibility of crossing an information set twice is typically referred to as ``absent-mindedness'', as a response to the absent-minded driver story and the subsequent modelling idea both due to Piccione and Rubinstein in \cite{Piccione1997Interpretation}. Together with the preceding result, we might thus be tempted to reject either the hypothesis that ``absent-mindedness'' as a concept is compatible with classical decision theory, as brought up in \cite{Piccione1997Interpretation}, or the claim of the generality of extensive form modelling (compare the corresponding discussion in \cite{AlosFerrer2016}). The author believes that to resolve this seeming dilemma it might be helpful to distinguish between a phenomenon (such as the absent-minded driver story) and a model or formal attempt to analyse or describe a phenomenon.

Note that from a point of view of classical decision theory and in particular of refined partitions-based (stochastic) extensive forms, the \emph{phenomenon} of absent-mindedness as expressed in this story is not contradictory in itself. 
This has been discussed, for instance, in \cite{Aumann1997Absent,Gilboa1997Comment}. Gilboa, for instance, has proposed an alternative description which can be succinctly formulated in stochastic extensive form (see the third of the simple examples in \cite{Rapsch2024DecisionA}, rediscussed in Subsection~\ref{subs:SEF_simple-examples}). Moreover, both \cite{Aumann1997Absent} and \cite{Gilboa1997Comment} make disappear much of the paradoxical conclusions from \cite{Piccione1997Interpretation} which suggests that the latter arise rather from the model than from the phenomenon.

In that sense, the \emph{phenomenon} of absent-mindedness is not at odds with the refined-partitions theory of (stochastic) extensive forms. There is such a model describing the strategic phenomenon convincingly and concisely. To make that clear we have used the term ``Heraclitus property'' for the above-described \emph{formal property} as opposed to the word ``absent-mindedness'' which we use here only for the phenomenon (without denying its use in purely graph-based models that one cannot always make sense of from the rigorous decision-theoretic refined partitions-based perspective). This point of view also distinguishes the present treatment from the one in \cite{AlosFerrer2016}.\smallskip

Let us conclude this subsection with two remarks. The first is about the actual information flow perceived by a given agent. If an agent $i\in I$ is at random move $\x\in\X^i$, the three pieces of information the agent has are $A^i(\x)$, $\ms F_\x^i$, and $\ms C_\x^i$. From this, the agent can infer the current information set ${\mf p}\in{\mf P^i}$ alias $P(c)$, for $c\in A^i(\x)$, the fact that the realised scenario $\omega$ is an element of $D_\x$, and for any $E\in\ms F_\x^i$ the fact whether $\omega\in E$ or not, including $E = \x^{-1}(P(c \cap c'))$ for all $c'\in\ms C_\x^i$. But the consistency conditions imply that all of this does not reveal more information along the vertical tree axis (e.g.\ about which $\x'\in{\mf p}$ is the actual one) or about the horizontal scenario axis (e.g.\ about events not contained in $\ms F_\x^i$).

Also note that the refined partitions approach implemented through stochastic (pseudo-) extensive forms does not require action labels. Choices already implicitly contain the data specifying conditional on which endogenous information they can be made -- this point has been made in \cite{AlosFerrer2005,AlosFerrer2008,AlosFerrer2011Comment} already (for classical (pseudo-) extensive forms). This implies that there is no necessity to add action labels because they are already implicit in the definition of choices. In the present framework, a choice does not only tell whether to go left at one particular move, but also at which set of moves. So for instance, a choice can consist in going left at move $x_0$; but it can also consist in going left at any move at time $2$; or it can consists in going left if agent $j\neq i$ has gone left at time $1$. Moreover, the fact whether these choices are available or not determines the endogenous information the given agent has: knowing whether you are at move $x_0$ or not when you are actually there; knowing nothing about what agents did before time $2$; knowing whether $j$ has gone left at time $1$, respectively. This has partly been discussed in \cite[Section~4]{Rapsch2024DecisionA} and is further detailed in the upcoming examples. Consequently, for a stochastic (pseudo-) extensive form denoted as above, one could call the elements $c\in C^i$ \emph{actions} of agent $i$. Although this is avoided for reasons of simplicity, the standard notion $A^i(x)$ for the set of actions at move $x\in X$ is retained.

Moreover, note that Axiom \ref{def:SEF.P(c)} also includes the statement that for any agent $i\in I$, any pair of choices $c,c'\in C^i$ with $P(c) = P(c')$ when ``evaluated'' in a particular scenario $\omega\in\Omega$ is either equal or disjoint. Along any tree $T_\omega$, two choices available at the same move are either identical or disjoint (alias strict alternatives). 

\subsection{Completeness}

Finally, let us consider Axiom~\ref{def:SEF.choice_completeness}. This is a completeness axiom. Indeed, any tuple of the form $(F,\pi,\X,I,\ms F,\ms C,C)$ satisfying the conditions from Definition~\ref{def:SEF} except Axiom~\ref{def:SEF.choice_completeness} can be modified by extending the set of choices for any agent such that the result satisfies Axiom~\ref{def:SEF.choice_completeness} and is equivalent to $(F,\pi,\X,I,\ms F,\ms C,C)$. More precisely:

\begin{lemma}\label{lemma:completeness}
    Let $\F = (F,\pi,\X,I,\ms F,\ms C,C)$ be a tuple satisfying the conditions defining a stochastic extensive form on some exogenous scenario space $(\Omega,\ms E)$ possibly except Axiom~\ref{def:SEF.choice_completeness}, according to Definition~\ref{def:SEF}. For any $i\in I$, let $\hat C^i$ be the set of all $\ms F^i$-$\ms C^i$-adapted choices $\hat c$ such that 
    \begin{enumerate}[label=(\roman*),ref=(\roman*)]
        \item\label{lemma:completeness.def:hatC.i} any $\omega\in\Omega$ with $\hat c\cap W_\omega\neq \emptyset$ admits $c\in C^i$ satisfying $\hat c\cap W_\omega = c \cap W_\omega$,
        \item\label{lemma:completeness.def:hatC.ii} there is $c\in C^i$ such that $P(\hat c) \subseteq P(c)$.
    \end{enumerate}
    Let $\hat C = (\hat C^i)_{i\in I}$. \smallskip

    Then, $\hat \F = (F,\pi,\X,I,\ms F,\ms C,\hat C)$ defines a stochastic extensive form on $(\Omega,\ms E)$ such that
    \begin{enumerate}
        \item\label{lemma:completeness.property_1} for all $i\in I$, all $\omega\in\Omega$,
        \[ \{ \hat c \cap W_\omega \mid \hat c\in\hat C^i\} \setminus \{\emptyset\} = \{c \cap W_\omega \mid c\in C^i\} \setminus \{\emptyset\}; \]
        \item\label{lemma:completeness.property_2} for all $i\in I$, 
        \[ \{P(\hat c) \mid \hat c\in \hat C^i\} = \{ P(c) \mid c\in C^i\}. \]
    \end{enumerate}

    The lemma remains true if ``stochastic extensive form'' is replaced with ``stochastic pseudo-extensive form'' everywhere.
\end{lemma}

By this lemma, completeness implies that choices are fully determined by their structure along any tree $T_\omega$, $\omega\in\Omega$, their immediate predecessor sets, and the informational structure along $\Omega$ given by $\ms F$ and $\ms C$. More precisely, any $C^i$, $i\in I$, is determined by the set of all $c\cap W_\omega$, $c\in C^i$, $\omega\in\Omega$, the family of sets $P(c)$, $c\in C^i$, and the families $\ms F^i$ and $\ms C^i$.

\subsection{Strategies}\label{subs:strategies}

In classical game and decision theory, a strategy is an agent's complete contingent plan of action (see e.g.\ \cite{MasColell1995}). As explained in \cite{AlosFerrer2005}, this can be viewed as a sequential version of acts in the theory of choice under uncertainty going back to Savage (\cite{Savage1972Foundations}). Acts map (given) ``states'' to (chosen) ``consequences''. 

\cite{AlosFerrer2005} takes the view that states are given by moves and consequences by available choices. Concerning ``states'', there is, however, a point of possible confusion which implies that there are different ways acts are formulated in the dynamic setting. One may argue that the point is not really the set of moves, but the information an agent has about them. In \cite{AlosFerrer2005}, as ``states'' are identified with moves, acts are required to assign identical consequences on information endogenous information sets. Although the present paper formally takes the perspective that endogenous information sets appear as the most precise interpretation of states, it perceives both viewpoints as equally convincing and equivalent.

On the other side, the refined partitions approach sees (local) ``consequences'' at a state as members of a partition of the set of (global) consequences (alias outcomes), that is, as choices in the sense of the present text. We see again that action is already implicitly described by choices, and no further structure of action labels or the like is needed, compare the discussion in the previous subsection. 

In order to ensure consistency (alias adaptedness) with respect to endogenous information in the mentioned sense, a strategy needs to assign to any of the agent's endogenous information sets a choice that is available at it. In addition, by restricting to adapted choices the local compatibility, or measurability, of acts with respect to exogenous information can be assured. This leads to the following definition.

\begin{definition}\label{def:strategy}
    Let $(F,\pi,\X,I,\ms F,\ms C,C)$ be a stochastic pseudo-extensive form and $i\in I$ an agent. 
    A \emph{strategy for $i$} is a map $s^i\colon {\mf P^i} \to C^i$ such that for all ${\mf p}\in{\mf P^i}$, $s^i({\mf p}) \in A^i({\mf p})$. 
    
    Let $S^i$ denote the set of strategies for $i$, and let $S = \bigtimes_{i\in I} S^i$. A \emph{strategy profile} is an element of $S$.
\end{definition}

As mentioned before, there are other ways of formally defining strategies which moreover seem more traditional. These interpret moves as states and impose restrictions on strategies in terms of (endogenous) information sets. As the setting of stochastic pseudo-extensive forms includes both moves and random moves, we obtain the following two temporary definitions. Let $i\in I$. 

An \emph{$X$-strategy for $i$} is a map $s^i\colon X^i \to C^i$ such that:
\begin{enumerate}
    \item for all $x\in X^i$, we have $s^i(x) \in A^i(x)$;
    \item for all $x,x'\in X^i$ with $A^i(x) = A^i(x')$, $s^i(x) = s^i(x')$.
\end{enumerate}
Clearly, this is equivalent to saying that for all $c\in C^i$, 
$$ \{x\in X^i \mid s^i(x) = c\} \in \{\emptyset, P(c)\}. $$
This definition is a direct formal generalisation of the corresponding definition in \cite[Subsection~5.2]{AlosFerrer2005} and exactly coincides with it in case of singleton $\Omega$.

An \emph{$\X$-strategy for $i$} is a map $s^i\colon \X^i \to C^i$ such that:
\begin{enumerate}
    \item for all $\x\in\X^i$, we have $s^i(\x) \in A^i(\x)$;
    \item for all $\x,\x'\in\X^i$ with $A^i(\x) = A^i(\x')$, $s^i(\x) = s^i(\x')$.
\end{enumerate}
Denote the set of $X$-strategies for $i$ by $S_X^i$ and the set of $\X$-strategies for $i$ by $S_\X^i$.

Furthermore, consider the natural surjections 
\begin{equation}\label{eq:Xi_surj}
    X^i \surj \X^i \surj {\mf P^i},
\end{equation}
with respect to which any map with domain ${\mf P^i}$ induces a map with domain $\X^i$ and any map with domain $\X^i$ induces a map with domain $X^i$, respectively. Then we have the following result.

\begin{proposition}\label{prop:strategies}
    Let $(F,\pi,\X,I,\ms F,\ms C,C)$ be a stochastic pseudo-extensive form, and $i\in I$ an agent. Then the maps in Equation~\ref{eq:Xi_surj} induce bijections
    $$ S^i \quad \stackrel{\cong}{\longrightarrow}\quad S_\X^i  \quad\stackrel{\cong}{\longrightarrow}\quad  S_X^i. $$  
\end{proposition}

Therefore, in the following, if $s^i$ is a strategy, both its corresponding $\X$- and $X$-strategy are denoted by $s^i$ as well.

\begin{remark}\label{rmk:strategies_adapted}
    Let $(F,\pi,\ms W,\X,I,\ms F,\ms C,C)$ be a stochastic pseudo-extensive form and $i\in I$ an agent. Any strategy for agent $i$ is adapted in the sense that for all $\x\in\X^i$ the choice $s^i(\x)$ is adapted, that is, for all reference choices $c'\in\ms C_\x$, and all $\x'\in\X^i$ that $s^i(\x)$ is available at: 
    $$ \x'^{-1}(P(s^i(\x)\cap c')) \in \ms F_{\x'}^i. $$
    This abstract adaptedness is a consequence of the structure of the underlying stochastic pseudo-extensive form. It is not part of the definition of a strategy and, in that sense, not a primitive of the theory. 

    In the case of action path stochastic pseudo-extensive forms this property is seen to correspond exactly to the adaptedness in the language of the theory of stochastic processes (see Subsection~\ref{subs:AP_SEF.information_histories_adapted_choices} and \cite[Subsection~4.4]{Rapsch2024DecisionA}).
\end{remark}

It is a standard procedure in game and decision theory to extend acts (alias strategies) so that they become maps from states to lotteries over consequences. As discussed in \cite{Anscombe1963Definition}, this is consistent with the theory of subjective probability. It is an elementary insight of game theory that, contrary to what one may naively infer from the single-agent situation, additional randomisation may improve coordination. In the dynamic setting, where states can be seen as endogenous information sets and consequences as available choices, this extension leads to the abstract notion of behaviour strategies. From an abstract point of view, a behaviour strategy is a complete contingent plan of lotteries over action. That is, the agent can draw an action at random at any endogenous information set. In contrast to this, an abstract mixed strategy is a lottery over complete contingent plans of action, that is, over strategies (which are then said ``pure'', in order to distinguish). That is, the agent can draw a ``pure'' strategy at random before the start and then commits to it from the beginning until the end. While for mixed strategies correlation over different endogenous information sets is possible, the lottery draws of a behaviour strategy at distinct endogenous information sets are independent.

Following von Neumann and Morgenstern, lotteries are interpreted as probability measures. The term ``lottery'' describes a typically unpredictable procedure of determining a consequence, for instance, the procedure consisting of observing the winner of a horse race. Yet, a probability measure describes only the statistical distribution of such a procedure's result whatever the meaning of ``statistical'' (frequentist, subjectivist, or other). It thus needs a way of transforming this abstract distribution into a procedure of the above-mentioned sort. In the literature, there are conflicting ways of doing so.\footnote{About the interpretation of randomised strategies, see the detailed discussion in \cite{Luce1989Games}.} 

The basic interpretation starts as follows: go to the horse race in question and observe the result. But then, one might either act blindly according to it or revise the own strategy in view of that new information. The point here is whether the horse race's result is stochastically independent of all exogenous information the agent has at that moment, and in particular, whether the agent has to commit to its result (as if another agent behaved within a mandate given by the original agent, or as if the horse race were some unconscious cognitive process). 

The second interpretation which does not require commitment can be implemented using an exogenous scenario space, as discussed in \cite{Aumann1974Subjectivity,Aumann1987Correlated}. In particular, it can be implemented in stochastic extensive forms by using strategies as defined above and profiles thereof. The randomisation is given in terms of exogenous scenarios, beliefs on these, and the dependence alias correlation structure across endogenous information sets and agents.\footnote{This is discussed in more detail in Section~\ref{sec:well-posedness_equilibrium}, once beliefs have been introduced.} In that sense, strategy profiles in stochastic extensive forms are profiles of correlated strategies whose correlation device and information structure are given by the exogenous scenario space, the exogenous information structures of the agents and beliefs. Also note that for the fundamental solution concept of Nash equilibrium (and its derivatives) the difference between the two interpretations above evaporates (there, commitment to the result of randomisation conforms to personal interest); hence, restricting on the second interpretation does not seem to imply a restriction.

\subsection{Simple examples}\label{subs:SEF_simple-examples}

In this subsection, we recall the three simple examples of stochastic decision forests, including the exogenous information structures admitting recall, reference choices structures, and adapted choices for them introduced in \cite{Rapsch2024DecisionA}. As we demonstrate here, these data do indeed yield stochastic extensive forms. The stochastic decision forest of the basic example is illustrated in \cite[Figure 1]{Rapsch2024DecisionA} which is reprinted as Figure~\ref{fig:simple_sdf} in the appendix.
It indicates \emph{pars pro toto} how finite stochastic extensive form problems can be formalised.

Formally, let $(\Omega,\ms E)$ be the discrete exogenous scenario space with exactly two scenarios, say $\Omega = \{\omega_1,\omega_2\}$, $\omega_1\neq\omega_2$, and $\ms E = \mc P\Omega$. 
Let $W = \Omega \times \{1,2\}^2$ and $\x_0,\x_1,\x_2\colon \Omega \to \mc P(W)$ given by $\x_0(\omega) = \{\omega\} \times \{1,2\}^2$ and $\x_k(\omega) = \{(\omega,k)\}\times \{1,2\}$, for $\omega\in\Omega$ and $k=1,2$. Further, let $F = \{\x_k(\omega) \mid \omega\in\Omega,~k=0,1,2\} \cup \{ \{w\} \mid w\in W\}$, and $\pi\colon F \to\Omega$ be the map sending any node to the first entry of an arbitrary choice among its elements.
The corresponding decision tree $(\Tr,\ge_\Tr)$ is illustrated in \cite[Figure 2]{Rapsch2024DecisionA} and reprinted as Figure~\ref{fig:simple_sdf_Tr} in the appendix. For the sake of simplicity, we consider only one agent $i$, i.e.\ $I = \{i\}$.

Regarding exogenous information, it is shown in \cite[Lemma~3.2]{Rapsch2024DecisionA} that there are exactly five exogenous information structure admitting recall, given by  the following families $\ms F^i = (\ms F^i_\x)_{\x\in\X}$:
\begin{enumerate}
    \item $\ms F^i_\x = \{\Omega,\emptyset\}$ for all $\x\in\X$: at all moves, it is unknown which scenario is realised;
    \item $\ms F^i_{\x_0} = \{\Omega,\emptyset\}$ and one of the following three cases is true:
    \begin{enumerate}
        \item $\ms F^i_{\x_1} = \ms F^i_{\x_2} = \mc P(\Omega)$: only at the second move, it becomes known which scenario is realised, irrespective of which one is the second move;
        \item $\ms F^i_{\x_1} = \mc P(\Omega)$, $\ms F^i_{\x_2} = \{\Omega,\emptyset\}$: $\x_1$ is the only move at that the realised scenario is revealed; an agent with this exogenous information may have interest in choosing (if possible) $\x_1$ rather than $\x_2$ in order to learn, modelling the trade-off \emph{exploration vs.\ exploitation}; that way, problems with partial information and adaptive control can be modelled;
        \item $\ms F^i_{\x_2} = \mc P(\Omega)$, $\ms F^i_{\x_1} = \{\Omega,\emptyset\}$: analogous to the preceding situation;
    \end{enumerate}
    \item $\ms F^i_\x = \mc P(\Omega)$ for all $\x\in\X$: at all moves, the realised scenario is known.
\end{enumerate}

Concerning choices, we formally generalise the following definitions made in \cite{Rapsch2024DecisionA}. Let $M$ be the set of maps $\Omega\to \{1,2\}$. For $k\in \{1,2\}$ and $f,g\in M$, \cite{Rapsch2024DecisionA} defines
\begin{align*}
    c_{f\bullet} =&~ \{(\omega,k',m')\in W \mid k' = f(\omega)\}, \\
    c_{k g} =&~  \{(\omega,k',m')\in W \mid k' = k, ~ m' = g(\omega)\}, \\
    c_{\bullet g} =&~\{(\omega,k',m')\in W \mid  m' = g(\omega)\}.
\end{align*}
$c_{k\bullet}$, $c_{\bullet m}$, and $c_{km}$ are defined by identifying $k,m\in\{1,2\}$ with the constant maps on $\Omega$ with values $k$ and $m$, respectively.  
Further, are defined $\ms C^i_{\x_0} = \{c_{1 \bullet},c_{2 \bullet}\}$, $\ms C^i_{\x_1} = \ms C^i_{\x_2} = \{c_{\bullet1},c_{\bullet2}\}$. We recall that the partitioned structure of these sets, reflecting the discreteness of the situation, should be noted. 

Next, consider the following table which is slightly different from the one in \cite[Subsection~4.3]{Rapsch2024DecisionA} because we are now interested in the consistency requirements defining stochastic extensive form, rather than just in providing a list of adapted choices. The table here reads as follows: Each line specifies a set of subsets of $W$ for one of the five exogenous information structures (\textsc{eis}) from \cite[Lemma~2.3]{Rapsch2024DecisionA} and recalled above; these subsets are classified according to whether they will correspond to choices at the beginning of the ``first period'' (at time $0$) or of the ''second period'' (at time $1$), if perceived as action path \textsc{sdf} according to \cite[Lemma~2.17]{Rapsch2024DecisionA}:
\begin{center}
    \begin{tabular}{r| c c}
     \textsc{eis}& 1st period & 2nd period  \\
     \hline
      1. & $c_{k \bullet}$ : $k\in\{1,2\}$ &$c_{k m}$ : $k,m\in\{1,2\}$ \\
      1. & $c_{k \bullet}$ : $k\in\{1,2\}$ &$c_{\bullet m}$ : $m\in\{1,2\}$ \\
      2.(a) &$c_{k \bullet}$ : $k\in\{1,2\}$ & $c_{k g}$ : $k\in\{1,2\}, g\in M$ \\
      2.(a) &$c_{k \bullet}$ : $k\in\{1,2\}$ & $c_{\bullet g}$ : $g\in M$ \\
      2.(b) & $c_{k \bullet}$ : $k\in\{1,2\}$ & $c_{1g}, c_{2m}$ : $m\in\{1,2\}, g\in M$ \\
      2.(c) & $c_{k \bullet}$ : $k\in\{1,2\}$ & $c_{1m}, c_{2g}$ : $m\in\{1,2\}, g\in M$ \\
      3. &  $c_{f \bullet}$ : $f\in M$ & $c_{k g}$ : $k\in\{1,2\}, g\in M$ \\
      3. &  $c_{f \bullet}$ : $f\in M$ & $c_{\bullet g}$ : $g\in M$ \\
    \end{tabular}
\end{center}

\begin{lemma}\label{lemma:simple_sef1}
    Let $I$ be a singleton, $i\in I$, $\ms F^i$ be any of the five preceding families, $\ms F = (\ms F^i)$, $\ms C^i$ as defined above, $\ms C = (\ms C^i)$, and $C^i$ be a set of choices corresponding to it via the preceding table, $C = (C^i)$. Then, the tuple $\F = (F,\pi,\X,I,\ms F,\ms C,C)$ defines a stochastic extensive form on $(\Omega,\ms E)$.
\end{lemma}

\begin{remark}
    Note that for the exogenous information structure 2.(b), the set
    \[ \tilde C^i = \Big\{c_{k \bullet} \mid k\in\{1,2\}\Big \} \cup \Big\{ c_{\bullet m} \mid m\in\{1,2\}\Big\} \]
    gives not rise to a stochastic extensive form because Axiom \ref{def:SEF.endo_exo_compatible} is violated. A similar remark can be made regarding 2.(c). In other words, the exogenous information available at time $1$ would reveal the current random move though the endogenous information at that node would not do so (the endogenous information set would contain both random moves at time $1$). Axiom \ref{def:SEF.endo_exo_compatible} stipulates that such an inconsistency must not arise. 
\end{remark}

Next, the variant of the simple example from \cite{Rapsch2024DecisionA} is recalled. It starts from the preceding example, but identifies the elements $(\omega_1,2,1)$ and $(\omega_1,2,2)$ in $W$ which provides a stochastic decision forest with a random move that is not defined on all of $\Omega$, as illustrated in \cite[Figure 3]{Rapsch2024DecisionA}, reprinted in Figure~\ref{fig:simple_sdf_variant} in the appendix. 
Formally, this leads one to consider $W' = W\setminus\{(\omega_1,2,1),(\omega_1,2,2)\} \cup \{(\omega_1,2)\}$, $\x'_0 = \x_0$, $\x'_1 = \x_1$, and $\x'_2\colon\{\omega_2\} \to \mc P(W')$ given by $\x'_2(\omega_2) = \{(\omega_2,2)\}\times \{1,2\}$. The set of random moves is denoted by $\X' = \{\x'_0,\x'_1,\x'_2\}$. The domains are given by $D_{\x'_0}=\Omega$, $D_{\x'_1} = \Omega$, $D_{\x'_2}=\{\omega_2\}$, and the set of nodes is given by $F'=\{\x'(\omega) \mid \x'\in\X',~\omega\in D_{\x'}\} \cup \{\{w'\}\mid w'\in W'\}$. The projection $\pi'\colon F'\to\Omega$ is the map sending any node to the first entry of an arbitrary choice among its elements. The corresponding decision tree $(\Tr',\ge_{\Tr'})$ is illustrated in \cite[Figure 4]{Rapsch2024DecisionA}, reprinted in Figure~\ref{fig:simple_sdf_variant_Tr} in the appendix.

It is shown in \cite[Lemma~3.3]{Rapsch2024DecisionA} that there are exactly three exogenous information structures admitting recall, given by the three following families $(\ms F^{\prime i}_{\x'})_{\x'\in\X'}$.
In all three cases, one has $\ms F^{\prime i}_{\x'_2} = \{\{\omega_2\},\emptyset\}$. Furthermore:
\begin{enumerate}
    \item $\ms F^{\prime i}_{\x'_0} = \ms F^{\prime i}_{\x'_1} = \{\Omega,\emptyset\}$: again, there could be an exploitation vs.\ exploration trade-off (compare the cases 2(c) above);
    \item $\ms F^{\prime i}_{\x'_0} = \{\Omega,\emptyset\}$, $\ms F^{\prime i}_{\x'_1} =\mc P(\Omega)$: this is similar to case 2(a) above;
    \item $\ms F^{\prime i}_{\x'_0} = \ms F^{\prime i}_{\x'_1} = \mc P(\Omega)$: the realised scenario is known at any move (similar to case 3 above).
\end{enumerate}

Concerning choices, the following definitions are made in \cite{Rapsch2024DecisionA}. Again, $M$ denotes the set of maps $\Omega\to \{1,2\}$. For $k\in \{1,2\}$ and $f,g\in M$, \cite{Rapsch2024DecisionA} defines
\begin{align*}
    c'_{f\bullet} =&~ \{(\omega,k',m'), (\omega,k')\in W \mid k' = f(\omega)\}, \\
    c'_{k g} =&~ \{(\omega,k',m')\in W \mid k' = k, ~ m' = g(\omega)\}, \\
    c'_{\bullet g} =&~ \{(\omega,k',m') \in W \mid  m' = g(\omega)\}.
\end{align*}
$c'_{k\bullet}$, $c'_{\bullet m}$, and $c'_{km}$ are defined by identifying $k,m\in\{1,2\}$ with the constant maps on $\Omega$ with values $k$ and $m$, respectively.  
Further, are defined $\ms C^{\prime i}_{\x'_0} = \{ c'_{1 \bullet},c'_{2 \bullet}\}$, $\ms C^{\prime i}_{\x'_1} = \ms C^{\prime i}_{\x'_2} = \{c'_{\bullet1},c'_{\bullet2}\}$. 

Next, consider the following table. Its interpretation is completely analogous to the one from the basic version of the simple example.
\begin{center}
    \begin{tabular}{r| c c}
     \textsc{eis}& 1st period & 2nd period  \\
     \hline
      1. & $c'_{k \bullet}$ : $k\in\{1,2\}$ & $c'_{k m}$ : $k,m\in\{1,2\}$ \\
      2. & $c'_{k \bullet}$ : $k\in\{1,2\}$ & $c'_{k g}$ : $k\in\{1,2\}, g\in M$ \\
      3. & $c'_{f \bullet}$ : $f\in M$ & $c'_{k g}$ : $k\in\{1,2\}, g\in M$ \\
    \end{tabular}
\end{center}   

\begin{lemma}\label{lemma:simple_sef2}
    Let $I'$ be a singleton, $i\in I'$, $\ms F^{\prime i}$ be any of the three preceding families, $\ms F' = (\ms F^{\prime i})$, $\ms C^{\prime i}$ as defined above, $\ms C = (\ms C^{\prime i})$, and ${C^{\prime i}}$ be the set of choices corresponding to it via the preceding table, $C' = (C^{\prime i})$. Then, the tuple $\F' = (F',\pi',\X',I',\ms F',\ms C',C')$ defines a stochastic extensive form on $(\Omega,\ms E)$.
\end{lemma}

\begin{remark}
    Note that for all three exogenous information structures and any $g\in M$, $c'_{\bullet g}$ cannot be the choice of an agent for some stochastic extensive form on $(F',\pi',\X')$. This is because the exogenous information available at time $1$ necessarily reveals the current random move, while the endogenous information at any move at time $1$ does not for $c'_{\bullet g}$ (the endogenous information set of that choice would contain both random moves at time $1$). Axiom \ref{def:SEF.endo_exo_compatible} stipulates that such an inconsistency must not arise. 
\end{remark}

The third example is a representation of Gilboa's interpretation of the absent-minded driver phenomenon in \cite{Gilboa1997Comment} whose \textsc{sdf}, exogenous information structures and adapted choices are well known from \cite{Rapsch2024DecisionA}. For this, let $(\Omega,\ms E)$ an exogenous scenario space and $\rho\colon \Omega \to \{1,2\}$ be an $\ms E$-$\mc P\{1,2\}$-measurable surjection. Further, let $\xi^1,\xi^2$ be $[0,1]$-valued random variables. Suppose that $(\Omega,\ms E)$ is rich enough to admit a probability measure under that $(\rho,\xi^1,\xi^2)$ is independent whose marginals are uniformly distributed on $\{1,2\}$ and $[0,1]$, respectively. 
Further, let $D,H,M$ be three different symbols meaning ``disastrous region'', ``home'' and ``motel'' as in the original story from \cite{Piccione1997Interpretation}. Let $W = \Omega \times \{D,H,M\}$ and $\x_1,\x_2$ be maps defined on the whole of $\Omega$ given by
\[ \x_k(\omega) = \begin{cases} \{\omega\} \times \{D,H,M\}, &\text{ if } \rho(\omega) = k, \\ \{\omega\} \times \{H,M\}, & \text{ if } \rho(\omega) = {3-k}, \end{cases} \qquad k\in\{1,2\}. \]
$F$ is defined as the union of the images of $\x_1$ and $\x_2$ and the set of all singleton sets in $W$. $\pi\colon F \to \Omega$ maps any element of $F$ to the first component of its elements. $(F,\supseteq)$ and $(\Tr,\ge_\Tr)$ are illustrated in Figure \ref{fig:absent_minded_driver_Gilboa_sdf}. Note that $(\Tr,\ge_\Tr)$ is not even a forest, though $\x_1$ and $\x_2$ are maximal elements.

Further, let $I =\{1,2\}$ and $\ms F^i_{\x_i} = \sigma(\xi^i)$, that is, the agents have no information other than their private signals $\xi^i$. 
Let, for $i\in I $:
\[ \op{Ex}_i = \underbrace{[\rho^{-1}(i)\times\{D\}] \cup [\rho^{-1}(3-i)\times\{H\}]\}}_{=\text{``exit''}}, \qquad \op{Ct}_i = \underbrace{[\rho^{-1}(i)\times\{H,M\}] \cup [\rho^{-1}(3-i) \times \{M\}]}_{\text{=``continue''}},\]
and $\ms C_{\x_i} = \{ \op{Ex}_i,\op{Ct}_i\}$.
Futhermore, for any $E\in\ms F^i_{\x_i}$, let
\[ c_i(E) = (W_E \cap \op{Ex}_i)\cup(W_{E^\complement} \cap \op{Ct}_i), \]
the choice of ``agent'' to exit in the event $E$ and to continue in the opposite event $E^\complement$. $E$ might be thought about as an event independent of $\rho$, allowing for individual ``randomisation''. Let $C^i = \{c_i(E) \mid E\in\ms F^i_{\x_i}\}$. That is, at both random moves $\x_i$, $i\in I=\{1,2\}$, the active agent $i$ has two basic choices: ``exit'' and ``continue'', between that $i$ can randomise according to private information. $\ms C^i$ defines a reference choice structure on $\{\x_i\}$ and that $C^i$ is a set of $\ms F^i$-$\ms C^i$-adapted choices, for both $i\in I$. Let $\ms F = (\ms F^i)_{i\in I}$, $\ms C = (\ms C^i)_{i\in I}$, and $C = (C^i)_{i\in I}$. 

\begin{thm}\label{thm:absent_minded_driver_Gilboa_sef}
    $(F,\pi,\X,I,\ms F,\ms C,C)$ defines a stochastic extensive form with perfect recall and imperfect information. 
\end{thm}

\section{Action path stochastic extensive forms}\label{sec:AP_SEF}

In most pieces of the literature, dynamic games are defined by supposing a notion of \emph{time} and specifying outcomes as certain paths of action at instants of time. \cite[Subsection~2.2]{AlosFerrer2005} provides a broad overview for this, including classical textbook definitions as in \cite{Fudenberg1991}, infinite bilateral bargaining in discrete time as in \cite{Rubinstein1982Perfect}, repeated games, the long cheap talk game in \cite{Aumann2003}, and a decision-theoretic interpretation of differential games as in \cite{Dockner2000}. In this series of papers the author desires to generalise this by, first, allowing for a large class of outcome paths and, second, adding a truly stochastic dimension. As an application, studied in the third paper, this will be used to add, in a precise approximate sense, stochastic control in not only in discrete, but also in continuous time (see, e.g.\ \cite{Pham2009Continuous,Bertsekas1996Stochastic,Karatzas1998Methods}) and stochastic differential games (see, e.g.\ \cite{Carmona2018}) without restrictions on the noise in question.

In \cite[Subsections~2.4, 3.4, 4.4]{Rapsch2024DecisionA}, a first part of this approach has been formalised in one abstract and general framework in terms of stochastic decision forests, that allows for general exogenous stochastic noise, therefore going strictly beyond the ``nature'' agent setting. The author insists on what has been said there, namely that ``this framework is based on a specific structure pertaining to all of these examples, namely \emph{time}. Interestingly, time is not included in the abstract formulation of decision forests, and it serves as a particularly strong similarity structure for trees and even branches of one and the same tree'' (\cite[Subsection~2.4, p.\ 15]{Rapsch2024DecisionA}). In this section, the construction is recalled and then it is shown under what conditions, how, and in what sense it gives rise to a stochastic extensive form.\smallskip

\subsection{Preliminaries: Action path stochastic decision forest data}

We start with recalling the construction by summarising \cite[Subsections~2.4, 3.4, 4.4]{Rapsch2024DecisionA}. Let $(\T,\le)$ be a total order admitting a minimum which we denote by $0$. Further, let $I$ be a non-empty finite set, $(\A^i)_{i\in I}$ be a family of non-empty metric spaces, and $\A = \prod_{i\in I} \A^i$ be their topological product, with canonical projections $p^i\colon\A \to \A^i$, $i\in I$. Of course, the case of singleton $I$ is included in this setting.

Let $(\Omega,\ms E)$ be an exogenous scenario space.
Let $W \subseteq \Omega\times\A^\T$ be such that for all $\omega\in\Omega$, there is $f\in\A^\T$ with $(\omega,f)\in W$. An outcome will thus be a pair of an exogenous scenario and a path $f\colon \T \to \A$ of ``action'', and any scenario is required to admit at least one outcome.
For any $t\in\T$ and $\tilde w = (\omega,f)\in\Omega\times\A^\T$, let \[x_t(\tilde w) = x_t(\omega,f) = \{ (\omega',f')\in W \mid \omega' = \omega, ~ f'|_{[0,t)_\T} = f|_{[0,t)_\T}\}.\] 
Let $F=\{x_t(w) \mid t\in\T, ~w\in W\} \cup \{\{w\} \mid w\in W\}$. Further, let $\pi\colon F\to \Omega$ be the unique function mapping any $x\in F$ to the first item of one choice of its elements.

For $(t,f)\in\T\times\A^\T$, let $D_{t,f} = \{\omega \in\Omega\mid |x_t(\omega,f)| \ge 2\}$. This will turn out as the event that $x_t(.,f)$ is a move.
We consider the following assumptions:
\begin{itemize}[label=--]
    \item\hypertarget{Ass:AP.SDF0}{\textbf{Assumption~AP.SDF0.}}~For all $t\in\T$ and $f\in\A^\T$, $D_{t,f}\in\ms E$. 
    \item\hypertarget{Ass:AP.SDF1}{\textbf{Assumption~AP.SDF1.}}~For all $w\in W$ and all $t,u\in\T$ with $t\neq u$ and $x_t(w) = x_u(w)$, we have $x_t(w) = \{w\}$.
    \item\hypertarget{Ass:AP.SDF2}{\textbf{Assumption~AP.SDF2.}}~For all $\omega\in\Omega$ and $\tilde f\in\A^\T$, and for all subsets $\T'\subseteq \T$ satisfying $x_t(\omega,\tilde f) \in F$ for all $t\in \T'$, there is $f\in\A^\T$ with $(\omega,f)\in W$ and $f|_{[0,t)_\T} = \tilde f|_{[0,t)_\T}$ for all $t\in \T'$.
    \item\hypertarget{Ass:AP.SDF3}{\textbf{Assumption~AP.SDF3.}}~For all $t\in\T$ and $f,g\in\A^\T$ such that $D_{t,f},D_{t,g}\neq\emptyset$ and $D_{t,f} \cap D_{t,g} = \emptyset$, there is $u\in [0,t)_\T$ such that $D_{u,f} \cap D_{u,g} \neq\emptyset$ and $f|_{[0,u)_\T} \neq g|_{[0,u)_\T}$.
\end{itemize}
For an interpretation of these conditions, the reader is referred to \cite[Subsection~2.4]{Rapsch2024DecisionA}. 
Further, let $\X$ be the set of maps
\[ \x_t(f) \colon D_{t,f} \to F, \omega \mapsto \x_t(f)(\omega) = x_t(\omega,f), \]
ranging over all $t\in\T$, $f\in \A^\T$ with $D_{t,f} \neq \emptyset$.

The tuple $(I,\A,\T,W)$ is called \emph{action path \textsc{sdf} data} on $(\Omega,\ms E)$ iff Assumptions \hyperlink{Ass:AP.SDF0}{AP.SDF$k$}, $k=0,1,2,3$, are satisfied. 
By \cite[Theorem~2.15]{Rapsch2024DecisionA}, given action path \text{sdf} data denoted as above, $(F,\pi,\X)$ defines an order consistent and maximal \textsc{sdf} on $(\Omega,\ms E)$. $(F,\pi,\X)$ is said to be the \emph{action path \textsc{sdf} induced} by the given action path \textsc{sdf} data. $\A = \prod_{i\in I} \A^i$ is called \emph{action space}, $\T$ its \emph{time axis}, $I$ the \emph{action index set} and $\A^i$ the \emph{$i$-th action space factor} of the action path \textsc{sdf} data, and \textit{a fortiori}, of the induced action path \textsc{sdf}.

For the remainder of this subsection, suppose that Assumptions \hyperlink{Ass:AP.SDF0}{AP.SDF$k$}, $k=0,1,2,3$, are satisfied so that $(F,\pi,\X)$ is an action path \textsc{sdf} on the exogenous scenario space $(\Omega,\ms E)$, with time axis $\T$, action index set $I$ and action space factors $(\A^i)_{i\in I}$. For any $x\in F$, let 
\[\T_x = \{t\in\T \mid \exists w\in x\colon x = x_t(w)\}. \]
By \cite[Lemma~2.9]{Rapsch2024DecisionA}, $\T_x$ is a singleton for any move $x\in X$. Let $\mf t\colon X\to \T$ be the map assigning to any move $x\in X$ the unique element $\mf t(x)$ of $\T_x$. By \cite[Lemma~2.16]{Rapsch2024DecisionA}, it is strictly decreasing and constant on the images of all random moves. Hence, it induces a strictly decreasing map $\X\to\T$ which is also denoted by $\mf t$. Sometimes, it is useful to extend these maps to $F$ and $\T$ in the following way. Let $\hat \T = \T \cup \{\T\}$ and extend $\le$ by letting $t \le \T$ for all $t\in\hat\T$.\footnote{We recall that this is a standard set-theoretic construction, introduced by von Neumann. In our context, it might seem suggestive and traditional to write $\infty$ for the element $\T\in\hat\T$. This, however, can be misleading at times, since $\T$ could already contain elements defined by limits, such as limit ordinals.} Further, let $\mf t(y) = \mf t(\y)  = \T$ for all $y\in F\setminus X$ and $\y\in\Tr \setminus\X$. Clearly, these extensions are still strictly decreasing.

The reader is referred to \cite[Lemma~2.17, Example 2.13, 2.14, 2.18]{Rapsch2024DecisionA} for a list of examples of action path stochastic decision forests (including timing problems, stochastic versions of ``differential games'' in the sense of \cite{AlosFerrer2005}, discrete-time problems, the simple examples from Subsection~\ref{subs:SEF_simple-examples}).

Concerning exogenous information structures for the action path \textsc{sdf} $(F,\pi,\X)$, let us recall here that given an exogenous information structure admitting recall and $f\in\A^\T$ such that all $t\in\T$ satisfy $D_{t,f} = \Omega$, $(\ms F_{\x_t(f)})_{t\in \T}$ defines a filtration in the usual sense of probability theory (a special case of \cite[Lemma~3.4]{Rapsch2024DecisionA}). Conversely, exogenous information structures with recall can be constructed by taking the filtration generated by a process depending in a non-anticipative way on the action path. For formal details, the reader is referred to \cite[Subsection~3.4]{Rapsch2024DecisionA}.

In stochastic extensive forms, action is described by choices and action labels are not necessary. In action path \textsc{sdf}, however, outcomes are paths of points in a space $\A$ which we call actions, and so choices are to be defined as certain sets of action paths. In order to close this loop and make both operations dual to each other, a choice is defined by a specific range of values in $\A$ at a given time, given a certain set of past trajectories of ``actions'' (as elements of $\A$). This is what has been proposed in \cite[Subsection~4.4]{Rapsch2024DecisionA} and the following builds upon.

Let $t\in\T$. For any set $A_{<t}\subseteq \A^{[0,t)_\T}$ and any family $A_t = (A_{t,\omega})_{\omega\in\Omega}\in \mc P(\A)^\Omega$ of subsets of $\A$, let
\[ c(A_{<t},A_t) = \{(\omega,f)\in W \mid f|_{[0,t)_\T} \in A_{<t},~ f(t) \in A_{t,\omega} \}. \]
Following \cite[Subsection~4.4]{Rapsch2024DecisionA}, for any $t\in\T$, $\ms C_t$ is defined as the set of all $c(A_{<t},A_t)$ ranging over all $A_{<t}\subseteq\A^{[0,t)_\T}$ and all families $A_t = (A_{t,\omega})_{\omega\in\Omega}\in \mc P(\A)^\Omega$ of subsets of $\A$ satisfying the following assumptions: 
\begin{itemize}[label=--]
    \item \hypertarget{Ass:AP.C0}{\textbf{Assumption AP.C0.}}~$c(A_{<t},A_t)\neq\emptyset$.
    \item \hypertarget{Ass:AP.C1}{\textbf{Assumption AP.C1.}}~For all $w\in c(A_{<t},A_t)$, there is $w'\in x_t(w)\setminus c(A_{<t},A_t)$. 
    \item \hypertarget{Ass:AP.C2}{\textbf{Assumption AP.C2.}}~For all $f\in\A^\T$ with $f|_{[0,t)_\T} \in A_{<t}$, we have \[x_t(\omega,f)\cap c(A_{<t},A_t) \neq \emptyset\] for all or for no $\omega\in D_{t,f}$.
\end{itemize}
For a detailed interpretation and analysis of these conditions, see \cite[Subsection~4.4, especially Lemmata 4.7, 4.8, 4.9]{Rapsch2024DecisionA}, where it has been shown that the set of immediate predecessors of elements of $\ms C_t$ can be easily calculated, that all elements of $c\in \ms C_t$ define complete and non-redundant choices, and that the time $\mf t(\x)$ of all random moves $\x\in\X$ that $c$ is available at is equal to $t$.

\subsection{Information, history structures, and adapted choices}\label{subs:AP_SEF.information_histories_adapted_choices}

The next steps consist in determining reference choice structures and adapted choices for any $i\in I$. While reference choice structures and adapted choices in \cite[Subsection~4.4]{Rapsch2024DecisionA} have been studied with respect to their consistency with respect to exogenous information, it remains to clarify their consistency with endogenous information. In action path \textsc{sdf}s, endogenous information at time $t$ can be modelled via partitions of histories prior to that $t$, as is clarified in the sequel.

For this, we introduce new notation. For $i\in I$, $t\in\T$, any subset $\T_t\subseteq \T$ with $[0,t)_\T \subseteq \T_t$, and all $f\in\A^{\T_t}$, let $D_{t,f}^i$ be the set of $\omega\in \Omega$ such that there are $f',f''\in\A^\T$ with $(\omega,f'), (\omega,f'')\in W$, $f'|_{[0,t)_\T} = f|_{[0,t)_\T} = f''|_{[0,t)_\T}$ and $p^i \circ f'(t) \neq p^i \circ f''(t)$.  This set will turn out as the event that, given the historic path $f|_{[0,t)_\T}$, agent $i$ can choose. Clearly, $D_{t,f}^i \subseteq D_{t,f}$, if $\T = \T_t$. In Proposition~\ref{prop:Dfti} this is stated in higher generality and it is seen that equality holds true unless $D_{t,f}^i = \emptyset$. Given action path \textsc{sdf} data $(I,\A,\T,W)$ and $i\in I$, let $\tilde\X^i = \{\x_t(f) \mid t\in\T,\,f\in\A^\T\colon D_{t,f}^i \neq \emptyset\}$.

\begin{definition}\label{def:mcH}
    Let $(I,\A,\T,W)$ be action path \textsc{sdf} data on an exogenous scenario space $(\Omega,\ms E)$ as before and $\ms F = (\ms F^i)_{i\in I}$ be a family of exogenous information structures $\ms F^i = (\ms F^i_\x)_{\x\in\tilde\X^i}$ on $\tilde\X^i$, $i\in I$. A \emph{history structure for $i$, given $(I,\A,\T,W,\ms F)$} is a family $\mc H^i = (\mc H^i_t)_{t\in\T}$ such that:
    \begin{enumerate}
        \item\label{def:mcH.partition} for each $t\in\T$, $\mc H^i_t$ is a partition of the set of all $f\in \A^{[0,t)_\T}$ such that $D_{t,f}^i \neq \emptyset$, 
        \item\label{def:mcH.msF_compatible} for all $t\in\T$, $A_{<t}\in \mc H^i_t$, and $f,f'\in A_{<t}$ we have $\ms F_{\x_t(f)}^i = \ms F_{\x_t(f')}^i$.
    \end{enumerate}
\end{definition}

The elements of $\mc H^i$ serve as a basis for the model of endogenous information sets in terms of partitions of histories, consistent with exogenous information. Axiom \ref{def:mcH}.\ref{def:mcH.partition} is meant to ensure that for any $i\in I$ and $t\in\T$, any element of $\mc H^i_t$ partitions the set of possible histories for moves of $i$ at time $t$. Axiom \ref{def:mcH}.\ref{def:mcH.msF_compatible} postulates the consistency with respect to exogenous information. Note that $\ms F_{\x_t(f)}^i = \ms F_{\x_t(f')}^i$ implies $D_{t,f} = D_{t,f'}$.

In the following, we fix, for any $i\in I$, an exogenous information structure $\ms F^i = (\ms F^i_\x)_{\x\in\tilde\X^i}$ and let $\ms F = (\ms F^i)_{i\in I}$. In addition, we fix a family $\mc H = (\mc H^i)_{i\in I}$ of history structures $\mc H^i$ for $i$, given $(I,\A,\T,W,\ms F)$, $i\in I$. For all $t\in\T$, $i\in I$, and $\x\in\tilde\X^i$ with $\mf t(\x) = t$, let $\ms C^i_{\x}$ the set of \emph{all} sets $c(A_{<t},A_t)$ as above such that
\begin{enumerate}
    \item\label{def:msC.1} $A_{<t}\in\mc H^i_t$;
    \item\label{def:msC.2} $A_t = (A_{t,\omega})_{\omega\in\Omega}$ such that there is $A_t^i\in \ms B(\A^i)$ satisfying, for all $\omega\in \Omega$, \[A_{t,\omega} = \begin{cases} (p^i)^{-1}(A^i_t), &\quad \omega\in D_\x, \\ \emptyset, &\quad \omega\notin D_\x; \end{cases}\]
    \item\label{def:msC.3} $c(A_{<t},A_t)\in\ms C_t$; and
    \item\label{def:msC.4} for all $\omega\in D_{\x}$, $\x(\omega) \cap c(A_{<t},A_t) \neq\emptyset$.
\end{enumerate}
These properties are referenced as ($\ms C^i_\x$.\ref{def:msC.1}) etc. In contrast to \cite[Subsection~4.4]{Rapsch2024DecisionA}, we have now specified the history structure implicit in $\ms C^i_\x$. $c(A_{<t},A_t)$ allows for choosing a measurable set of actions in the $i$-th action space factor at the random move $\x$ given an endogenous past $A_{<t}$ at time $t$, as specified by $\mc H$. By \cite[Proposition~4.10]{Rapsch2024DecisionA}, $\ms C^i = (\ms C_\x^i)_{\x\in\tilde\X^i}$ defines a reference choice structure for any $i\in I$. 

\begin{remark}
    By \cite[Lemma~2.17]{Rapsch2024DecisionA}, the ``simple examples'' for stochastic decision forests discussed also in Subsection~\ref{subs:SEF_simple-examples} can be represented as action path \textsc{sdf}. However, the reference choice structures considered in \cite[Lemmata 4.3, 4.5]{Rapsch2024DecisionA}, and recalled in Subsection~\ref{subs:SEF_simple-examples}, are not those obtained by the previous construction. Actually, the previous construction is obtained from the former by element-wise intersecting with the known endogenous past. 
    
    Let us spell this out for the basic version. Denote by $\ms C^{\text{AP}}$ the action path construction and by $\ms C^{\text{orig}}$ the ``original'' definition from \cite{Rapsch2024DecisionA}, recalled in Subsection~\ref{subs:SEF_simple-examples}. Then, recalling the table in Subsection~\ref{subs:SEF_simple-examples}, describing the different combinations of exogenous information structures, reference choice structures and sets of choices:
    \begin{enumerate}
        \item $\ms C^{\text{AP}}_{\x_0} = \ms C^{\text{orig}}_{\x_0}$, and
        \item for both $k=1,2$:
        \begin{enumerate}
            \item if there is a bullet in the second period, that is, for the second, fourth or eighth line of the table, then: $\ms C^{\text{AP}} = \ms C^{\text{orig}}$ 
            \item else, that is, for the first, third, fifth, sixth, seventh line, then: 
            \[ \ms C^{\text{AP}}_{\x_k} = \{ c \cap \im \x_k \mid c\in \ms C^{\text{orig}}_{\x_k} \} = \{ c \cap c_{k\bullet} \mid c\in \ms C^{\text{orig}}_{\x_k} \}. \]
        \end{enumerate}
    \end{enumerate}

    While the action path construction is more accurate in that it precisely conditions on the piece of endogenous information the agent really has, the sets of adapted choices are identical under both specifications of the reference choice structure. After all, reference choice structures are about $\ms B(\A^i)$ rather than about $\A^{[0,t)_\T}$. In \cite{Rapsch2024DecisionA}, therefore, it was not relevant to make that distinction, and for the simple examples it does not really matter. In general, however, the general construction given for action path \textsc{sdf} data given above is preferable, because available actions may vary across different pasts $A_{<t}$ so that a reference choice for a given past may no more be a choice for another past.
\end{remark}

Next, the actual choices prospective agents make are introduced, and, again, we do so with explicit reference to $\mc H$. Let $i\in I$, $A_{<t}\in\mc H^i_t$, $D\in\ms E$, and $g\colon D\to\A^i$. Let $A_t^{i,g} = (A_{t,\omega}^{i,g})_{\omega\in\Omega}$ be given by
\[ A_{t,\omega}^{i,g} = \begin{cases} \{a\in \A \mid p^i(a) = g(\omega) \}, &\quad\omega\in D, \\ \emptyset, &\quad \omega\notin D. \end{cases} \]
Let $c(A_{<t},i,g) = c(A_{<t},A_t^{i,g})$. Very similarly to what has been said in \cite[Subsection~4.4]{Rapsch2024DecisionA}, if this set is an element of $\ms C_t$, it models the choice to take, given an endogenous history in $A_{<t}$, the action $g(\omega)$ in the $i$-th action space factor in scenario $\omega\in D$ at time $t$.

Further assumptions are needed to construct a stochastic (pseudo-) extensive form out of all this. First, as we wish to interpret action indices as agents, simultaneous actions in different action space factors should be independent from each other. Furthermore, we require \emph{in fine} that scenariowise choices available at a move $x = \x(\omega)$, for some $(\x,\omega) \in \X\bullet\Omega$ can be extended a) to a representative class of elements of $\ms C^i_\x$ and b) to adapted choices of the form $c(A_{<t},i,g)$. In this sense, this requires $W$ to be measurable with respect to the Borel $\sigma$-algebra on $\A$. In addition, separation assumptions with respect to $W$ are proposed that makes it possible to separate outcomes via choices in the relevant way.

In the following, let us call a set $\mc M$ \emph{stable under non-trivial intersections} iff for all $A,B\in\mc M$ with $A\cap B\neq\emptyset$, we have $A\cap B\in\mc M$. This property is equivalent to saying that $\mc M \cup \{\emptyset\}$ is stable under intersections. For example, any partition has this property. 

\begin{itemize}[label=--]
    \item\hypertarget{Ass:AP.SEF0}{\textbf{Assumption AP.SEF0.}}~For all subsets $J\subseteq I$, all $t\in\T$, all $(f_j)_{j\in J} \in(\A^{\T})^J$ and $\omega\in \Omega$ with $(\omega,f_j)\in W$ for all $j\in J$ and $f_j|_{[0,t)_\T} = f_{j'}|_{[0,t)_\T}$ for all $j,j'\in J$, there is $f\in \A^\T$ such that 
    \[ \forall j\in J\colon \quad \Big[ p^j \circ f(t) = p^j \circ f_j(t), \quad (\omega,f) \in x_t(\omega,f_j)\Big]~. \]
    \item\hypertarget{Ass:AP.SEF1}{\textbf{Assumption AP.SEF1.}}~For all $i\in I$, $t\in\T$, $f\in\A^\T$ such that $D_{t,f}^i \neq \emptyset$ and the unique $A_{<t}\in\mc H^i_t$ with $f|_{[0,t)_\T}\in A_{<t}$, there is a generator $\ms G(\A^i)$ of $\ms B(\A^i)$ stable under non-trivial intersections such that for all $G\in\ms G(\A^i)$, upon letting $A_t^{i,G} = (A_{t,\omega}^{i,G})_{\omega\in\Omega}$ be given by $A_{t,\omega}^{i,G} = (p^i)^{-1}(G)$ for $\omega\in D_{t,f}$ and $A_{t,\omega}^{i,G} = \emptyset$ for $\omega\notin D_{t,f}$, we have
    \[ c(A_{<t},A_t^{i,G}) \in \ms C_{\x_t(f)}^i. \]
    \item\hypertarget{Ass:AP.SEF2}{\textbf{Assumption AP.SEF2.}}~For all $i\in I$, $t\in\T$, $f\in\A^\T$, $\omega\in D_{t,f}^i$ with $(\omega,f)\in W$, and the unique $A_{<t}\in\mc H^i_t$ satisfying $f|_{[0,t)_\T}\in A_{<t}$, there is a map $g\colon D_{t,f} \to \A^i$ such that $p^i\circ f(t) = g(\omega)$, $c(A_{<t},i,g)\in \ms C_t$, $c(A_{<t},i,g)$ is $\ms F^i$-$\ms C^i$-adapted, and for all $(\omega',{f'}_{<t})\in D_{t,f}\times A_{<t}$ there is $f'\in\A^\T$ satisfying $(\omega',f')\in c(A_{<t},i,g)$ and $f'|_{[0,t)_\T} = {f'}_{<t}$.
    \item\hypertarget{Ass:AP.psi-SEF3}{\textbf{Assumption AP.$\psi$-SEF3.}}~ For all $f,f'\in\A^\T$ and $t_0\in\T$ with $f(t_0)\neq f'(t_0)$ and $\omega\in\Omega$ such that $(\omega,f),(\omega,f')\in W$, there are $t\in\T$ and $i\in I$ such that $t\le t_0$, $p^i\circ f(t) \neq p^i \circ f'(t)$ and $\omega\in D_{t,f}^i \cap D_{t,f'}^i$.
    \item\hypertarget{Ass:AP.SEF3}{\textbf{Assumption AP.SEF3.}}~ For all $f,f'\in\A^\T$ with $f\neq f'$ and $\omega\in\Omega$ such that $(\omega,f),(\omega,f')\in W$, the set
    \[ \{ t\in\T \mid f(t) \neq f'(t)\}\]
    has a minimum.
\end{itemize}

Assumption \hyperlink{Ass:AP.SEF1}{AP.SEF1} is tightly linked to Assumption AP.C3 in \cite[Subsection~4.4]{Rapsch2024DecisionA} (see Theorem~\ref{thm:adapted_choices_mb_functions} later in this section). Hence, if we assume \hyperlink{Ass:AP.SEF1}{AP.SEF1}, then by \cite[Theorem~4.13]{Rapsch2024DecisionA}, provided that $c = c(A_{<t},i,g)\in\ms C_t$, the $\ms F^i$-$\ms C^i$-adaptedness of the choice $c(A_{<t},i,g)$ is equivalent to the $\ms F_\x^i$-measurability of $g|_{D_\x}$ for all $\x\in\tilde\X^i$ that $c$ is available at. Note that Assumption \hyperlink{Ass:AP.SEF3}{AP.SEF3} is satisfied if all $f\in\A^\T$ with $(\omega,f)\in W$ for some $\omega\in\Omega$ are locally right-constant, with respect to the order topology on $\T$.\footnote{By definition, $f$ is locally right-constant iff for every non-maximal $t\in\T$ there is $u\in\T$ with $t<u$ such that $f|_{[t,u)_\T}$ is constant.} This is necessarily true if $(\T,\le)$ is a well-order.

It has been discussed in some detail in \cite[Subsection~4.4]{Rapsch2024DecisionA} that these choices can model dependence on exogenous and endogenous information independently from each other. For instance, if $A_{<t} = \{f|_{[0,t)_\T}\}$ for some $f\in\A^\T$ such that $D_{t,f}^i\neq\emptyset$, for an action index $i\in I$, the prospective agent corresponding to $i$ who can choose $c(A_{<t},i,g)$ can make her or his choice dependent on whether the past actions are described by $f$ or not -- while $\ms F_{\x_t(f)}$ could well be very coarse allowing only a small amount of functions $g$, that is, a weak dependence on the exogenous scenario $\omega$. Similar remarks can be made in the opposite case and in mixed regimes. This means that ``open'' and ``closed loop'' decision making and control can be understood with respect to endogenous and exogenous information independently: in the example just mentioned, the loop could be closed with respect to endogenous information, but need not be with respect to exogenous information. See \cite[Subsection~4.4]{Rapsch2024DecisionA} for more details, Proposition~\ref{thm:link_endogenous_information_H} below for the link between $\mc H$ and the endogenous information structure, and \cite[Section~3]{Rapsch2024DecisionA} for more details regarding exogenous information structures.\smallskip

For each $i\in I$, let $C^i$ be the set of all sets of the form $c(A_{<t},i,g)$ where $t\in\T$, $A_{<t} \in \mc H^i_t$, $D\in\ms E$, $g\colon D\to \A^i$ such that $c(A_{<t},i,g)\in\ms C_t$, $c(A_{<t},i,g)$ is $\ms F^i$-$\ms C^i$-adapted, and for all $(\omega,f_{<t})\in D \times A_{<t}$ there is $f\in \A^\T$ with $(\omega,f)\in c(A_{<t},i,g)$ and $f|_{[0,t)_\T} = f_{<t}$. Let $C=(C^i)_{i\in I}$. This concludes the construction of the necessary data.

\subsection{Action path stochastic extensive form data}

We briefly fix some names for the objects just constructed, before showing that they provide stochastic (pseudo-) extensive forms whose properties are directly related to corresponding properties of the data.

\begin{definition}\label{def:SEF_data}
    Let $(\Omega,\ms E)$ be an exogenous scenario space.
    \begin{enumerate}
        \item \emph{Action path stochastic pseudo-extensive form ($\psi$-\textsc{sef}) data} on $(\Omega,\ms E)$ are defined to be a tuple
        \[ (\ast) \qquad \D = \big (I,\A,\T,W,\ms F,\mc H\big) \]
        such that $(I,\A,\T,W)$ is action path \textsc{sdf} data on $(\Omega,\ms E)$, $\ms F = (\ms F^i)_{i\in I}$ is a family of exogenous information structures $\ms F^i$ on $\tilde\X^i$, $i\in I$, and $\mc H = (\mc H^i)_{i\in I}$ is a family of history structures $\mc H^i$ for $i\in I$, such that Assumptions \hyperlink{Ass:AP.SEF0}{AP.SEF$k$}, $k=0,1,2$, and \hyperlink{Ass:AP.psi-SEF3}{AP.$\psi$-SEF3} are satisfied.
        If $\D$ is action path $\psi$-\textsc{sef} data on $(\Omega,\ms E)$, then its entries are denoted as in $(\ast)$.
        \item \emph{Action path stochastic extensive form (\textsc{sef}) data} on $(\Omega,\ms E)$ are action path $\psi$-\textsc{sef} data satisfying Assumption \hyperlink{Ass:AP.SEF3}{AP.SEF3}.
        \item Suppose that $\D$ is action path $\psi$-\textsc{sef} data on $(\Omega,\ms E)$. Associate to it $\ms C = (\ms C^i)_{i\in I}$ and $C = (C^i)_{i\in I}$ in the way defined above. The tuple
            \[ \F = (F,\pi,\X,I,\ms F,\ms C,C) \]
        is said to be \emph{the \textsc{sef} candidate induced by $\D$}. The term ``\textsc{sef} candidate'' can be replaced with ``(pseudo-) stochastic extensive form'' if $\F$ satisfies the respective property.
    \end{enumerate}
\end{definition}

Before clarifying whether \textsc{sef} data induce \textsc{sef}, whether pseudo or not, respectively, we explain the meaning of the sets $D_{t,f}^i$. For this define, for all $t\in\T$, all sets of time $\T_t \subseteq \T$ with $[0,t)_\T\subseteq \T_t$, all $f\in\A^{\T_t}$, the set
\[ \hat D_{t,f} = \{ \omega\in\Omega \mid \exists f',f'' \in\A^\T\colon (\omega,f'), (\omega,f'')\in W,~ f'|_{[0,t)_\T} = f|_{[0,t)_\T} = f''|_{[0,t)_\T},~ f'(t) \neq f''(t)\}. \]
\begin{proposition}\label{prop:Dfti}
    Let $(\Omega,\ms E)$ be an exogenous scenario space, $\D$ be action path $\psi$-\textsc{sef} data on it and $\F$ be the induced action path $\psi$-\textsc{sef} candidate. Then, the following statements hold true.
    \begin{enumerate}
        \item\label{prop:Dfti.hatDft=unionDfti} For all $t\in\T$, sets of time $\T_t \subseteq \T$ with $[0,t)_\T\subseteq \T$,  $f\in\A^{\T_t}$, we have
    \[ \hat D_{t,f} = \bigcup_{i\in I} D_{t,f}^i. \]
        \item\label{prop:Dfti.hatDft_subseteq_Dft} For all $t\in\T$ and $f\in\A^\T$, we have $\hat D_{t,f} \subseteq D_{t,f}$.
        \item\label{prop:Dfti.Dfti_nonempty} For all $t\in\T$, $f\in\A^\T$, and $i\in I$, we have:
        \[ \Big[ D_{t,f} \neq \emptyset \text{ and } \x_t(f) \in \X^i \Big] \quad \Longleftrightarrow \quad D_{t,f}^i \neq \emptyset, \]
        and if either side of the equivalence holds true, then $D_{t,f} = D_{t,f}^i$.
    \end{enumerate}
\end{proposition}

\begin{corollary}\label{cor:tildeXi=Xi}
    Let $(\Omega,\ms E)$ be an exogenous scenario space, $\D$ be action path $\psi$-\textsc{sef} data on it and $\F$ be the induced action path $\psi$-\textsc{sef} candidate. Then, for all $i\in I$, $\tilde\X^i = \X^i$. \hfill \qed
\end{corollary}

Note, however, that it is easy to construct action path \textsc{sef} data with $\hat D_{t,f} \subsetneq D_{t,f}$ for some $(t,f)\in\T\times\A^\T$, see Example \ref{ex:SEF_with_Dtf_neq_hatDtf}. Hence, the induced \textsc{sef} candidate satisfies $\bigcup_{i\in I} \X^i \subsetneq \X$, that is, provided it defines an \textsc{sef} (which it indeed does, by the following theorem), there are random moves no agent is active at. 

After these preparations, we can now state the central result of this section.

\begin{thm}\label{thm:AP_sef}
    Let $(\Omega,\ms E)$ be an exogenous scenario space and consider action path stochastic (pseudo-) extensive form data $\D$ on it and let $\F$ be the \textsc{sef} candidate induced by $\D$ as in Definition~\ref{def:SEF_data}. Then $\F$ is a stochastic (pseudo-) extensive form, respectively.
\end{thm}

The reader is referred to Subsection~\ref{subs:SDF} as well as \cite[Examples 2.13, 2.14, 2.15, 3.5, 4.11, 4.12, 4.14]{Rapsch2024DecisionA} for concrete examples of stochastic decision forests, exogenous information structures, and adapted choices in the action path case. In the third paper of the present series, specific applications in the realm of continuous time will be discussed.

Next, we explicitly compute the sets of available choices and information sets for action path $\psi$-\textsc{sef}.
\begin{proposition}\label{prop:compute_Ai(x)_mfp}
    Let $(\Omega,\ms E)$ be an exogenous scenario space, $\D$ be action path $\psi$-\textsc{sef} data on it, $\F$ be the induced action path $\psi$-\textsc{sef}, and $i\in I$.
    \begin{enumerate}
        \item\label{prop:compute_Ai(x)_mfp.Ai(x)} For all $\x\in\X$ and $t\in\T$, $f\in\A^\T$ with $\x = \x_t(f)$, we have:
        \[ A^i(\x) = \{c\in C^i \mid \exists A_{<t}\in \mc H^i_t~\exists g\colon D_{t,f}\to \A^i\colon c = c(A_{<t},i,g), \, f|_{[0,t)_\T}\in A_{<t}\}. \]
        \item\label{prop:compute_Ai(x)_mfp.mfp} The function mapping any pair $(t,A_{<t})$, where $t\in\T$ and $A_{<t}\in\mc H^i_t$, to
        \[ \mf p = \{\x_t(f) \mid f\in\A^\T\colon f|_{[0,t)_\T} \in A_{<t} \} \]
        is well-defined, injective and has image $\mf P^i$.
    \end{enumerate}
\end{proposition}

In the next theorem, which can be seen as a direct consequence of \cite[Theorem~4.13]{Rapsch2024DecisionA}, we clarify how the structural assumption of adaptedness of choices relates to the measure-theoretic concept of measurability of maps. 
\begin{thm}\label{thm:adapted_choices_mb_functions}
    Let $(\Omega,\ms E)$ be an exogenous scenario space, $\D$ be action path $\psi$-\textsc{sef} data on it, $\F$ be the induced action path $\psi$-\textsc{sef}, and $i\in I$. Further, let $t\in\T$, $A_{<t}\in\mc H^i_t$, $D\in\ms E$ and $g\colon D\to \A^i$ be a map such that $c = c(A_{<t},i,g)\in\ms C_t$. 
    Then we have:
    \begin{enumerate}
        \item\label{thm:adapted_choices_mb_functions.non_red_compl} $c$ is a non-redundant and $\X^i$-complete choice.
        \item\label{thm:adapted_choices_mb_functions.Dx_subseteq_D} For all $\x\in\X^i$ that $c$ is available at, we have $D_\x \subseteq D$.
        \item\label{thm:adapted_choices_mb_functions.mb} $c$ is $\ms F^i$-$\ms C^i$-adapted iff for all $\x\in\X^i$ that $c$ is available at, $g|_{D_\x}$ is $\ms F^i_\x$-measurable.
    \end{enumerate}
\end{thm}
The point of this theorem is to explain the usual measurability assumption on random action in terms of the decision-theoretic concept of adapted choices.\footnote{Of course, this measurability requirement is not sufficient to imply $c\in C^i$ in general, because, in addition, the latter requires that for all $(\omega,f_{<t})\in D \times A_{<t}$ there is $f\in \A^\T$ with $(\omega,f)\in c(A_{<t},i,g)$ and $f|_{[0,t)_\T} = f_{<t}$. Without this requirement, first, $c$ could be available at $\x_t(f)$, but not at $\x_t(f')$, for some $f,f'\in\A^\T$ with $f|_{[0,t)_\T},f'|_{[0,t)_\T}\in A_{<t}$, and, second, $D$ could be too large.}\smallskip

We conclude this subsection (and section) by classifying the endogenous information structure of an action path \textsc{sef} in terms of $\mc H$.
For $t,u\in\T$ with $t<u$ let $p_{u,t}$ be the restriction
\[ \A^{[0,u)_\T} \to \A^{[0,t)_\T}, ~ f\mapsto f|_{[0,t)_\T}. \]

\begin{thm}\label{thm:link_endogenous_information_H}
    Let $(\Omega,\ms E)$ be an exogenous scenario space, $\D$ be action path $\psi$-\textsc{sef} data on it, $\F$ be the induced action path $\psi$-\textsc{sef}, and $i\in I$. Then the following statements hold true:
    \begin{enumerate}
        \item\label{thm:link_endogenous_information_H.perfect_recall} $i$ admits perfect endogenous recall iff for all $t,u\in\T$ with $t<u$, all $A_{<t}\in\mc H^i_t$, all $A_{<u}\in\mc H^i_u$, we have 
        \begin{itemize}[label=--]
            \item either: a) $(\mc P p_{u,t})(A_{<u}) \subseteq A_{<t}$ and b) all $f,f'\in A_{<u}$ satisfy $p^i \circ f(t) = p^i\circ f'(t)$,
            \item or: $(\mc P p_{u,t})(A_{<u}) \cap A_{<t} = \emptyset$.
        \end{itemize}
        \item\label{thm:link_endogenous_information_H.perfect_information} $i$ has perfect endogenous information iff for all $t\in\T$, all $A\in\mc H^i_t$ are singletons such that $A\cap A' = \emptyset$ for all $j\in I\setminus\{i\}$ and $A'\in \mc H^j_t$.
    \end{enumerate}
\end{thm}

That is, the finer $\mc H^i$ is, the more precise endogenous information agent $i$ has, and vice versa; and $i$ admits perfect recall iff, under the identification introduced by the projection operator $p$, $\mc H^i_t$ is essentially a flow of refined partitions in $t$ and $i$ recalls past decisions.

\section{Well-posedness and equilibrium}\label{sec:well-posedness_equilibrium}

A crucial property of extensive forms is what the author suggests calling ``well-posedness'', namely, that conditional on any history, any strategy profile induces a unique outcome and that any outcome can be attained this way. In this section, the concept of histories from \cite[p.\ 219]{AlosFerrer2008} is reformulated within the setting of decision forests. We relate this to the new concept of random histories in the case of stochastic decision forests $(F,\pi,\X)$ which, in the order consistent, surely non-trivial, and maximal case, turn out as the histories of the associated decision tree $(\Tr,\ge_\Tr)$ and corresponds to the scenario-wise principal up-sets of random moves. Then, well-posedness is formulated following \cite{AlosFerrer2008} and characterised in terms of well-posedness of the scenario-wise classical extensive forms. As a result, notions of preferences on strategy profiles and equilibrium concepts can be defined in well-posed stochastic extensive forms because consequences of strategy profiles always exist and are unique.

In the three final subsections, examples of well- and ill-posed stochastic extensive forms are discussed, first of rather pedagogic nature, second in the case of action path outcomes. Moreover, it is emphasised using examples why the nature representation of dynamic noise is not sufficient in general, an issue that the theory of games in stochastic extensive form can resolve.

\subsection{Histories}

Let us start with extending the definition of histories from \cite[p.\ 219]{AlosFerrer2008} to decision forests, following \cite[Section~4.4]{AlosFerrer2016}. For this, let us remind the reader that a subset $h\subseteq N$ of a poset $(N,\ge)$ is \emph{upward closed} iff for all $x\in h$, we have $\uparrow x \subseteq h$.

\begin{definition}[\cite{AlosFerrer2016}]\label{def:history}
    Let $(F,\ge)$ be a decision forest. A \emph{history in $(F,\ge)$} is a non-empty, non-maximal and upward closed chain. The set of histories in $(F,\ge)$ is generically denoted by $H$.
\end{definition}

Actually, two histories may be equivalent in that they have the same sets of maximal chains containing them, respectively. This motivates the following lemma and definition.

\begin{lemma}\label{lemma:closed_history}
    Let $(F,\ge)$ be a decision forest and let $h\in H$ be a history.
    \begin{enumerate}
        \item\label{lemma:closed_history.1} There is a unique upward closed chain $\overline h$ in $(F,\ge)$ a) satisfying $\overline h \subseteq w$ for all maximal chains $w$ in $(F,\ge)$ with $h\subseteq w$ and b) maximal with respect to set inclusion among all chains in $(F,\ge)$ satisfying Property a).
        \item\label{lemma:closed_history.2} We have
        \[ \overline h = \bigcap \{ \uparrow x\mid x\in F\colon h\subseteq \uparrow x\}. \]
        \item\label{lemma:closed_history.3} For any further history $h'\in H$, we have $\overline {h'} = \overline h$ iff for any maximal chain $w$ in $(F,\ge)$ we have
        \[ h'\subseteq w \quad \Longleftrightarrow\quad h\subseteq w. \]
        \item\label{lemma:closed_history.4} We have $\overline h = h \cup\{\inf h\}$ if $h$ admits an infimum, and $\overline h = h$ otherwise.
    \end{enumerate}
\end{lemma}

\begin{definition}\label{def:closed_history}
    Let $(F,\ge)$ be a decision forest and let $h\in H$ be a history. $\overline h$ is said the \emph{closure of $h$}. $h$ is said \emph{closed} iff $h=\overline h$.
\end{definition}

From the perspective of outcome generation, it seems to suffice to describe historical dependence in terms of closed histories. This is confirmed by Lemma~\ref{lemma:R_continuous_in_h} in the following subsection.\smallskip

Within the realm of stochastic decision forests the exogenous information agents can condition on is described based on random moves rather than on moves. This is particularly true in the order consistent case which motivates the following concept of histories compatible with random moves.

\begin{definition}\label{def:random_history}
    Let $(F,\pi,\X)$ be an order consistent stochastic decision forest on an exogenous scenario space $(\Omega,\ms E)$. A \emph{(closed) random history in $(F,\pi,\X)$ }is a map $\h$ with domain $D_\h\in\ms E \setminus \{\emptyset\}$ such that $\h(\omega)$ is a (closed, respectively) history in $(T_\omega,\supseteq)$ for all $\omega\in D_\h$ and such that for all $\x\in\X$ admitting $\omega\in D_\x\cap D_\h$ with $\x(\omega)\in\h(\omega)$, we have $D_\x \supseteq D_\h$ and $\x(\omega') \in \h(\omega')$ for all $\omega'\in D_\h$. The set of random histories in $(F,\pi,\X)$ is generically denoted by $\H$.
\end{definition} 

\begin{lemma}\label{lemma:histories_canonical_inj}
    Let $(F,\pi,\X)$ be a stochastic decision forest on an exogenous scenario space $(\Omega,\ms E)$. Then is an injection $X\inj H$ associating to any move $x\in X$ the closed history $\uparrow x$. Moreover, if $(F,\pi,\X)$ is order consistent, there is an injection $\X\inj \H$ associating to any random move $\x\in\X$ the closed random history 
    \[ \h \colon D_\x \to H, \, \omega \to \uparrow \x(\omega). \]
\end{lemma}

Via these injections, we consider $X$ as a subset of $H$ and, if applicable, $\X$ as a subset of $\H$. Moreover, if $(F,\pi,\X)$ is surely non-trivial, then, via the natural injection $\Omega\to X,\,\omega\mapsto W_\omega$, $\Omega$ can be seen as a subset of $X$ and thus of $H$, too.\smallskip

The following proposition establishes that random histories in an order consistent, surely non-trivial, and maximal \textsc{sdf} and histories in the induced decision tree compatible with at least one scenario naturally correspond to each other. 

\begin{proposition}\label{prop:random_histories}
    Let $(F,\pi,\X)$ be an order consistent, surely non-trivial, and maximal stochastic decision forest on an exogenous scenario space $(\Omega,\ms E)$ and $(\Tr,\ge_\Tr)$ be the induced decision tree. Denote the set of histories in $(\Tr,\ge_\Tr)$ by $H_\Tr$. Let $f$ be the map with domain $\H$ associating to any $\h\in\H$ the set of $\x\in\X$ such that there is $\omega\in D_\x\cap D_\h$ with $\x(\omega)\in \h(\omega)$. Then, the following statements hold true.
    \begin{enumerate}
        \item\label{prop:random_histories.image_H_Tr} The image of $f$ is given by the set of a) all $h_\Tr\in H_\Tr$ that are non-maximal chains in $(\X,\ge_\X)$, and b) all maximal chains $h_\Tr$ in $(\X,\ge_\X)$ admitting non-empty $D\in\ms E\setminus\{\emptyset\}$ with $D\subseteq \bigcap_{\x\in h_\Tr} D_\x$ such that for any $\omega\in D$ there is $w\in W_\omega$ with $w\in\bigcap_{\x\in h_\Tr} \x(\omega)$. 
        \item\label{prop:random_histories.faithful} $f$ is faithful in that for all $\h_1,\h_2\in\H$ with $f(\h_1) = f(\h_2)$ there is $\h\in\H$ with $D_{\h_1}\cup D_{\h_2}\subseteq D_\h$ and $\h|_{D_{\h_k}} = \h_k$ for both $k=1,2$.
        \item\label{prop:random_histories.closed} For any closed random history $\h\in\H$, $f(\h)$ is a closed history in $(\Tr,\ge_\Tr)$. 
    \end{enumerate}
\end{proposition}

\subsection{Induced outcomes and well-posedness}

In extensive form theory, the outcome of strategic interaction is defined as the consequence of local decision making compatible with the rules defined by the tree and choice structures. More precisely, ``the'' outcome $w$ compatible with a strategy profile $s$, given a history $h$, is characterised by the fact that it is not discarded by any of the strategy profile's choices at moves $x$ succeeding the history $h$ and containing $w$. Yet, beyond the realm of standard finite or discrete-time situations, such an outcome need not exist, or there may be several of them, and there may even be outcomes that can never be attained that way (see, i.a.\ \cite{Simon1989,Stinchcombe1992,AlosFerrer2008}). In a well-posed extensive form model of decision making, these three problems must not arise.

From a decision-theoretic point of view, this implies two things. First, it is relevant to formally define well-posedness of (stochastic) extensive forms. This is done in the following definition. As the theory exposed here is based on the refined partitions approach developed by Alós-Ferrer and Ritzberger (see, e.g.\ \cite{AlosFerrer2016}), this definition here is necessarily structurally similar to the definitions in \cite[p.\ 228]{AlosFerrer2008} and \cite[pp.\ 102, 105, 106]{AlosFerrer2016}. For the sake of accessability, a notation similar to that used in these texts is used here as well. 

Second, additional assumptions are required to make an extensive form well-posed. This is discussed below the following definition.

\begin{definition}\label{def:sef_well-posed}
    Let $\F = (F,\pi,\X,I,\ms F,\ms C,C)$ be a stochastic pseudo-extensive form on an exogenous scenario space $(\Omega,\ms E)$. 

    \begin{enumerate}
        \item\label{def:sef_well-posed.induced_outcome} Let $R = R(\F)$ be the map assigning to all triples $w\in W$, $s\in S$, $h\in H$ with $w\in \bigcap h$ the set
        \[ R(w,s \mid h) = \bigcap \Big\{ s^i(x) \mid x \in X,\,i\in J(x)\colon~ w\in x \subseteq \bigcap h \Big\}. \]

        Given such data, $w$ is said \emph{compatible with $s$ given $h$} iff $ w \in R(s,w \mid h) $.
        \item\label{def:sef_well-posed.well_posed} $\F$ is called \emph{well-posed} iff 
        \begin{enumerate}
            \item\label{def:sef_well-posed.well_posed.onto_outcomes} for all $h\in H$ and all $w\in\bigcap h$, there is $s\in S$ such that $w$ is compatible with $s$ given $h$;
            \item\label{def:sef_well-posed.well_posed.existence} for all $s\in S$ and all $h\in H$, there is $w\in \bigcap h$ that is compatible with $s$ given $h$;
            \item\label{def:sef_well-posed.well_posed.uniqueness} for all $s\in S$ and all $h\in H$, there is at most one $w\in \bigcap h$ that is compatible with $s$ given $h$, and in this case, $R(w,s \mid h) = \{w\}$.
        \end{enumerate}

        \item\label{def:sef_well-posed.well_posed.outcome_map} Suppose that $\F$ is well-posed. Then, for any strategy profile $s\in S$ and history $h\in H$, the unique outcome $w\in \bigcap h$ compatible with $s$ given $h$ is said \emph{outcome induced by $s$ given $h$}, and denoted by $w = \Out(s\mid h)$. The corresponding map $\Out(.\mid .)\colon S\times H\to W$ is called \emph{(induced) outcome map}. 
    \end{enumerate}
\end{definition}

Note that $R(w,s \mid h)$ describes the set of outcomes that are compatible with both the outcome -- or decision path associated to -- $w$ and the strategy profile $s$ given the history $h$. By definition, $w$ is compatible with $s$ given $h$ iff it is compatible with itself and $s$ given $h$. A minimum requirement on the result of strategic interaction in terms of $s$ and given history $h$ is not to be discarded by $s$, given $h$, that is, to be compatible with $s$ given $h$.

Property \ref{def:sef_well-posed.well_posed.onto_outcomes} means that given any history $h$, any outcome yet undiscarded (by $h$) can be compatible with some strategy profile. If the other two existence and uniqueness properties are satisfied so that the outcome map can be defined at all, Property \ref{def:sef_well-posed.well_posed.onto_outcomes} is equivalent to that for any history $h\in H$ we have $\im\Out(. \mid h) = \bigcap h$. Property \ref{def:sef_well-posed.well_posed.existence} describes existence of compatible outcomes for all strategy profiles given any history. Property \ref{def:sef_well-posed.well_posed.uniqueness} describes uniqueness of compatible outcomes $w$ for all strategy profiles $s$ given any history $h$ and, moreover, of outcomes compatible with such triplets $(w,s,h)$.

Note that in the special case of classical extensive forms this corresponds to the definition of ``induced'' outcomes and conditional versions of conditions (A0), (A1), (A2) in \cite[Chapter~5, pp.\ 102, 105, 106]{AlosFerrer2016} though the term ``well-posed'' is not used therein, to the best of the author's knowledge. As explained above, it is of utmost importance to understand when a stochastic extensive form is well-posed, and a main contribution of \cite{AlosFerrer2008,AlosFerrer2011Comment}, comprehensively presented in \cite{AlosFerrer2016}, is to characterise this in the case of classical extensive forms, namely in terms of order-theoretic properties of the underlying decision tree.

The next lemma confirms formally that, as claimed in the preceding subsection, it is sufficient to restrict to closed histories in what concerns outcome generation.

\begin{lemma}\label{lemma:R_continuous_in_h}
    Let $\F = (F,\pi,\X,I,\ms F,\ms C,C)$ be a stochastic pseudo-extensive form on an exogenous scenario space $(\Omega,\ms E)$. Furthermore, let $s\in S$ be a strategy profile, $h\in H$ be a history, and $w\in W=\bigcup F$ be an outcome. Then,
    \[ R(w,s \mid h) = R(w,s \mid \overline h).\]
\end{lemma}

Next, we discuss under what conditions an extensive form is well-posed. First, we consider the property whether any undiscarded outcome can be attained.
\begin{thm}\label{thm:well_posed.onto_outcomes}
    For any stochastic extensive form $\F$, Property \ref{def:sef_well-posed}.\ref{def:sef_well-posed.well_posed.onto_outcomes} is satisfied.
\end{thm}
In case of classical extensive forms and the unconditional version of Property \ref{def:sef_well-posed}.\ref{def:sef_well-posed.well_posed.onto_outcomes}, this basic fact has already been established in \cite[Theorem~5.1]{AlosFerrer2016}.

Furthermore, we classify existence and uniqueness. For this, the classification from the classical case in \cite{AlosFerrer2008,AlosFerrer2011Comment} and \cite{AlosFerrer2016} is applied. Therefore, it remains to describe existence and uniqueness in terms of adequate existence and uniqueness properties for suitable classical extensive forms. In fact, these classical pseudo-extensive forms are given by the scenario-wise truncations of the given stochastic pseudo-extensive form:

\begin{proposition}\label{prop:scwise_SEF}
    Let $(F,\pi,\X,I,\ms F,\ms C,C)$ be a stochastic (pseudo-) extensive form on some exogenous scenario space $(\Omega,\ms E)$ and $\omega\in\Omega$. Then, $(T_\omega,I,C_\omega)$ is a classical (pseudo-) extensive form, respectively, where $C_\omega = (C_\omega^i)_{i\in I}$ and
    $$ C_\omega^i = \{ c\cap W_\omega \mid c\in C^i\} \setminus \{\emptyset\}, \qquad i\in I. $$
\end{proposition}

We turn to the central theorem of this subsection translating well-posedness properties of stochastic pseudo-extensive forms into the language of classical extensive forms.

\begin{thm}\label{thm:SEF_well-posed}
    Let $\F = (F,\pi,\X,I,\ms F,\ms C,C)$ be a stochastic pseudo-extensive form on an exogenous scenario space $(\Omega,\ms E)$. Then $\F$ satisfies Properties \ref{def:sef_well-posed}.\ref{def:sef_well-posed.well_posed.onto_outcomes}, \ref{def:sef_well-posed}.\ref{def:sef_well-posed.well_posed.existence}, \ref{def:sef_well-posed}.\ref{def:sef_well-posed.well_posed.uniqueness}, well-posedness, iff for all $\omega\in\Omega$, the classical pseudo-extensive form $(T_\omega,I,C_\omega)$ does so\footnote{Precisely, some and any stochastic pseudo-extensive form with set of nodes $T_\omega$, set of agents $I$, and family of sets of individual choices $C_\omega = (C^i_\omega)_{i\in I}$ \emph{does so}. Clearly, these properties do only depend on $T_\omega$, $I$, and $C_\omega$ in the singleton-$\Omega$ case.}, respectively.
\end{thm}

This theorem makes it possible to apply the classification results from \cite{AlosFerrer2008,AlosFerrer2011Comment}, see \cite{AlosFerrer2016} for a comprehensive monographic treatment. For this, we recall four important notions from these works (see, e.g.\ \cite[Definitions~4.2, 5.1]{AlosFerrer2016}).\footnote{Without loss of generality, we reformulate them for decision forests.} Let $(F,\ge)$ be a decision forest. For any history $h$ in $(F,\ge)$, a \emph{continuation} is the complement $w\setminus h$ of $h$ in a maximal chain alias decision path $w$ containing $h$. $(F,\ge)$ is said
\begin{enumerate}
    \item \emph{weakly up-discrete} iff for all non-terminal nodes $x\in F$ any maximal chain in $\downarrow x \setminus \{x\}$ has a maximum;
    \item \emph{up-discrete} iff any non-empty chain has a maximum;
    \item \emph{coherent} iff every history without minimum has at least one continuation with a maximum;
    \item \emph{regular} iff for all non-maximal $x\in F$, the history $\uparrow x \setminus \{x\}$ has an infimum.
\end{enumerate}
See \cite{AlosFerrer2008,AlosFerrer2011Comment} and \cite{AlosFerrer2016} for a discussion and examples. We only give a brief overview here. Weak up-discreteness requires the existence of successor nodes, while up-discreteness even demands all chains to be well-ordered with respect to $\ge$.\footnote{This might be confusing because typically well-orders are considered with respect to $\le$, rather than $\ge$.} The latter is equivalent to the existence of a maximum for every continuation of every history (\cite[Lemma 5.4]{AlosFerrer2016}). Roughly speaking, coherence ensures the existence of (inductive) limit candidates and regularity requires such a limit to be uniquely identifiable. Weak up-discreteness, coherence, and regularity are mutually independent of each other. See \cite{AlosFerrer2008} and \cite{AlosFerrer2016} for more details.

We obtain the following two corollaries:

\begin{corollary}\label{cor:SEF_existence=~order-theoretic_properties}
    Let $F$ be a decision forest over a set $W$ such that for any $x\in X$ and $w\in x$, we have
    \begin{equation}\label{eq:available_choices}
        x \supsetneq \bigcup \Big\{ y\in \downarrow x \setminus \{x\} \bigmid w\in y\Big\}. 
    \end{equation}
    Then, the following statements are equivalent:
    \begin{enumerate}
        \item\label{cor:SEF_existence=~order-theoretic_properties.order_properties.existence} Every stochastic pseudo-extensive form $\F$ with decision forest $F$ satisfies Property \ref{def:sef_well-posed}.\ref{def:sef_well-posed.well_posed.existence}.
        \item\label{cor:SEF_existence=~order-theoretic_properties.order_properties} $(F,\supseteq)$ is weakly up-discrete and coherent.
    \end{enumerate}
\end{corollary}

In the context of classical extensive forms in \cite{AlosFerrer2008,AlosFerrer2016}, the fact that for all $x\in X$ and $w\in x$ Equation~\ref{eq:available_choices} is described by the sentence that $F$ has \emph{available choices}. It characterises the mere possibility of defining a stochastic pseudo-extensive form on it. See \cite[Subsection~4.5]{AlosFerrer2016} for more details. By \cite[Corollary 1]{AlosFerrer2008} (or identically, \cite[Corollary 4.1]{AlosFerrer2016}), weak up-discreteness of $(F,\supseteq)$ implies that property. In particular, the implication ``\ref{cor:SEF_existence=~order-theoretic_properties.order_properties} $\Rightarrow$ \ref{cor:SEF_existence=~order-theoretic_properties.order_properties.existence}'' holds true without additionally requiring available choices.

\begin{corollary}\label{cor:SEF_well-posedness=~order-theoretic_properties}
    Let $\F$ be a stochastic extensive form on an exogenous scenario space $(\Omega,\ms E)$. Then, the following statements are equivalent:
    \begin{enumerate}
        \item\label{cor:SEF_well-posedness=~order-theoretic_properties.well-posed} $\F$ is well-posed.
        \item\label{cor:SEF_well-posedness=~order-theoretic_properties.order_properties} $(F,\supseteq)$ is weakly up-discrete, coherent, and regular.
        \item\label{cor:SEF_well-posedness=~order-theoretic_properties.order_properties2} $(F,\supseteq)$ is up-discrete and regular.
    \end{enumerate}
\end{corollary}

Hence, the consistency requirements of stochastic extensive forms suffice to fully characterise well-posedness in terms of easily verifiable order-theoretic properties of the underlying decision forest $F$.\smallskip

While in classical extensive forms subgames are defined in terms of moves, the obvious analogon in stochastic extensive form is given by random moves, that is, the sections of moves that, potentially, exogenous information is revealed and choices are available at. 
From this perspective, it is important to understand the link between random moves and histories, similarly to the classical theory (as discussed also in, e.g.\ \cite[p.\ 106]{AlosFerrer2016}). Moreover, in the order consistent case, the link between random moves and random histories is interesting since the latter correspond to the relevant histories in the induced decision tree $(\Tr,\ge_\Tr)$, by Proposition~\ref{prop:random_histories}. 

\begin{proposition}\label{prop:random_histories=~random_moves}
    Let $\F = (F,\pi,\X,I,\ms F,\ms C,C)$ be a well-posed stochastic extensive form on an exogenous scenario space $(\Omega,\ms E)$. Then for any closed history $h\in H$ there is $x\in X$ with $h = \uparrow x$. Moreover, if $(F,\pi,\X)$ is order consistent, then for any closed random history $\h$ there is $\x\in\X$ such that $D_\h \subseteq D_\x$ and for all $\omega\in D_\x$, $\h(\omega) = \uparrow \x(\omega)$.
\end{proposition}

We see that in well-posed stochastic extensive forms, closed histories always have minima. From the perspective of outcome generation, it suffices to consider closed histories, by Lemma \ref{lemma:R_continuous_in_h}. Hence, sections of closed histories whose minima constitute random moves are the crucial object for representing subgames. In the order consistent case, moreover, this corresponds exactly to closed random histories.

\subsection{Expected utility preferences}\label{subs:EU_preferences}

In a well-posed stochastic extensive form $\F$, dynamic decision making can be analysed by comparing consequences of strategy profiles, given any random move $\x$. By Proposition~\ref{prop:random_histories=~random_moves} and Lemma~\ref{lemma:R_continuous_in_h}, this is equivalent to conditioning on collections of closed histories whose infima constitute random moves, and in the order consistent case, this is equivalent to conditioning on closed random histories. Actually, an agent can only condition on its information set. Hence, agent $i\in I$ is ultimately interested in the maps
\[ \Out(s\mid .)\colon \mf p\times D_\mf p \to W,\, (\x,\omega) \mapsto \Out(s\mid \x(\omega)), \]
ranging over all strategy profiles $s\in S$ and all endogenous information sets $\mf p\in\mf P^i$.

In the following, a general concept of analysis is presented. We begin with a concept of comparing consequences. This concept is only an example, though a quite central and general one, and not generic. It implements the idea of expected utility in its basic form (which itself exhibits strong links to preference-based rational choice under uncertainty, \cite{Savage1972Foundations,Anscombe1963Definition,Gilboa1989Maxmin, Hara2023Multiple}). On it, as is discussed later, a theoretical equilibrium concept can be based, that covers (subgame-) perfect versions of Nash, correlated and Bayesian equilibrium. For this, dynamic consistency of ``comparison'' is central which is why this is discussed in some detail in the following. 

\begin{definition}\label{def:EU_pref_str}
    Let $\F$ be a well-posed stochastic extensive form on an exogenous scenario space $(\Omega,\ms E)$.
    \begin{enumerate}
        \item A \emph{belief system on $\F$} is a family $\Pi = (p_{i,\mf p},\ms P_{i,\mf p},\P_{i,\mf p})_{i\in I,\,\mf p\in\mf P^i}$ such that, for any $i\in I$ and any $\mf p\in\mf P^i$, $\ms P_{i,\mf p}$ is a $\sigma$-algebra in $\mf p$, $\P_{i,\mf p}$ is a probability measure on $\ms E|_{D_\mf p}$, and $p_{i,\mf p} \colon D_{\mf p} \to \mf p$ is an $\ms E|_{D_\mf p}$-$\ms P_{i,\mf p}$-measurable map.
        \item A \emph{taste system on $\F$} is a family $U = (u_{i,\mf p})_{i\in I,\,\mf p\in\mf P^i}$ of maps $u_{i,\mf p} \colon W \to \R$.
        \item An \emph{expected utility (\textsc{eu}) preference structure on $\F$} is a tuple $\Pr = (\Pi,U,\ms W)$ where
        \begin{itemize}[label=--]
            \item $\Pi$ is a belief system on $\F$,
            \item $U$ is a taste system on $\F$, and
            \item $\ms W$ is a $\sigma$-algebra on $W$,
        \end{itemize}
        such that, 
        we have, for all $i\in I$ and $\mf p\in\mf P^i$:
        \begin{enumerate}
            \item\label{def:EU_pref_str.uip_mb} $u_{i,\mf p}$ is $\ms W$-Borel-measurable;
            \item\label{def:EU_pref_str.Out_mb} $\Out_{i,\mf p}^s\colon D_\mf p \to W,\, \omega \mapsto \Out(s \mid p_{i,\mf p}(\omega),\omega)$ is $ \ms E|_{D_\mf p}$-$\ms W$-measurable for all $s\in S$;
            \item\label{def:EU_pref_str.uip_Out_intb} $u_{i,\mf p} \circ \Out_{i,\mf p}^s$ is Lebesgue-quasi-integrable\footnote{That is, the negative part or the positive part is Lebesgue-integrable.} with respect to $\P_{i,\mf p}$ for all $s\in S$;
            \item\label{def:EU_pref_str.W_rich_enough} the map $\psi_\mf p\colon W\to \mf p \cup \{\emptyset\}$ assigning to any $w\in W$ the unique\footnote{If two immediate predecessors in $X$ of a choice can be compared, then they are equal, see Lemma~\ref{lemma:Heraclitus_property}. As the evaluation on $\X^i \bullet \Omega$ is injective, $\x$ must be unique.} random move $\x\in \mf p$ such that $w\in \bigcup \im\x$ if it exists and $\emptyset$ else, is measurable with respect to $\ms W$ and the $\sigma$-algebra on $\mf p \cup \{\emptyset\}$ generated by $\ms P_{i,\mf p} $.
        \end{enumerate}
    \end{enumerate}
\end{definition}
In other words, an \textsc{eu} preference structure fixes beliefs (probability measures and random draws $p_{i,\mf p}$ on endogenous information sets $\mf p$, inducing probabilities $(p_{i,\mf p})_\ast \P_{i,\mf p}$ on $\mf p$), tastes (utility functions, or payoffs), and a measurability structure on outcomes, such that tastes and outcome generation is measurable (\ref{def:EU_pref_str.uip_mb}, \ref{def:EU_pref_str.Out_mb}), expected utility can be computed (\ref{def:EU_pref_str.uip_Out_intb}, which is always satisfied if $u_{i,\mf p}$ is bounded below), and the sets of outcomes in $W$ describing endogenous information sets are measurable (\ref{def:EU_pref_str.W_rich_enough}). Moreover, these data are given conditional on all agents and all of their endogenous information sets because beginning at any of these, a well-posed stochastic extensive form is induced, defining a decision situation in its own right.\footnote{We omit the formal argument behind this statement, in the interest of brevity.} In the perfect information case, these exactly correspond to classical subgames. An analysis of the original stochastic extensive form might necessitate an analysis of these induced \textsc{sef}, in order to rule out non-credible threats, for instance (which motivated the concept of subgame-perfect equilibrium in \cite{Selten1965}).

Once fixed, an \textsc{eu} preference structure is considered common knowledge among all agents, at least in the following example of a rationality and equilibrium concept.\footnote{Without that assumption, one would want to assume beliefs of higher order etc.} For a rational agent whose rationality is common knowledge, then, it seems plausible that beliefs on two different information sets are consistent with respect to the outcome generation map induced by the given strategy profile. This is a version of the Harsanyi doctrine (\cite{Harsanyi1967, Harsanyi1968a, Harsanyi1968b, Aumann1974Subjectivity}), according to which differences in posterior beliefs should only result from differences in information. We suggest using a dynamic variant asking for common priors ``locally'' between any two given information sets, therefore taking into account consistency off the equilibrium path. 

One can also argue that rational agents whose rationality is common knowledge, tastes are necessarily consistent for any individual agent across endogenous information sets. Indeed, consider an agent $i$ performing equilibrium analysis. Why should $i$ make contingent plans for action at two different, but possibly related endogenous information sets at which tastes differ meaning that the consequences of action are evaluated differently? Put equivalently, supposing a specific form of rationality might lead us to split dynamically inconsistent ``agents'' into ``multiple selves'' (see \cite{Gilboa1997Comment,Strotz1955Myopia}). These notions of consistency are formalised as follows.

\begin{definition}\label{def:consistent_EU_pref_str}
    Let $\F$ be a well-posed stochastic extensive form on an exogenous scenario space $(\Omega,\ms E)$, $\Pi$ be a belief system and $U$ be a taste system on $\F$. Moreover, let $\ms W$ be a $\sigma$-algebra on $W$.
    \begin{enumerate}
        \item Let $\Pi$ and $\ms W$ satisfy Assumptions~\ref{def:EU_pref_str.Out_mb} and~\ref{def:EU_pref_str.W_rich_enough} in Definition~\ref{def:EU_pref_str}, and let $s\in S$ be a strategy profile. The pair $(\Pi,s)$ is said \emph{dynamically consistent} iff for all sets $J \subseteq \bigcup_{i\in I} \{i\} \times \mf P^i$ with at most two elements,
        \begin{enumerate}
            \item\label{def:consistent_EU_pref_str.beliefs.p} for all $(i,\mf p_i),(j,\mf p_j)\in J$, all $\omega\in D_{\mf p_i}$ with
            \[ (\ast) \qquad \psi_{j,\mf p_j} \circ \Out^s_{i,\mf p_i}(\omega) \neq\emptyset, ~\text{ and }~ p_{i,\mf p_i}(\omega)(\omega) \supseteq \psi_{j,\mf p_j} \circ \Out^s_{i,\mf p_i}(\omega)(\omega) \]
            satisfy
            \[ p_{j,\mf p_j}(\omega) = \psi_{j,\mf p_j} \circ \Out^s_{i,\mf p_i}(\omega); \]
            \item for all $(i,\mf p_i),(j,\mf p_j)\in J$ the set 
            \[ E_{i,j} = \{ \omega\in D_{\mf p_i} \cap D_{\mf p_j} \mid p_{i,\mf p_i}(\omega)(\omega) \supseteq  p_{j,\mf p_j}(\omega)(\omega) \} \]
            is an event, i.e.\ satisfies $E_{i,j} \in \ms E$;
            \item\label{def:consistent_EU_pref_str.beliefs.common_prior} there is a \emph{common prior for $J$}, that is, a probability measure $\P$ on $(\Omega,\ms E)$ such that for all $((i,\mf p_i),(j,\mf p_j))\in J^2$ with $(i,\mf p_i)\neq(j,\mf p_j)$ we have $\P(D_{\mf p_i}\cup D_{\mf p_j}) = 1$, and the events
            \[  E_{\neg i,j} = D_{\mf p_j} \setminus E_{i,j}, \quad E_{i,j}^s = \{ \omega \in E_{i,j} \mid \psi_{j,\mf p_j} \circ \Out^s_{i,\mf p_i}(\omega) \neq \emptyset \} \]
            and any $E\in \ms E|_{D_{\mf p_j}}$ satisfy
            \[ \P_{j,\mf p_j}(E) \cdot \P(E_{\neg i,j} \cup E_{i,j}^s) = \P(( E_{\neg i,j} \cup E_{i,j}^s)\cap E). \]
        \end{enumerate}
        \item $U$ is said \emph{dynamically consistent} iff for all $i\in I$, there is a map $u_i\colon W \to \R$ such that for any $\mf p\in\mf P^i$, we have $u_{i,\mf p} = u_i$.
        \item Provided $\Pr = (\Pi,U,\ms W)$ is an \textsc{eu} preference structure, and given a strategy profile $s\in S$, $(\Pr,s)$ it is said \emph{dynamically consistent} iff both $(\Pi,s)$ and $U$ are dynamically consistent.
    \end{enumerate}
\end{definition}

Some of the technical details of the definition may require some verbal explanation. Given $i, j\in I$, $\mf p_i\in\mf P^{i}$, $\mf p_j\in\mf P^j$ as above, then $\psi_{j,\mf p_j}\circ \Out^s_{i,\mf p_i}(\omega)$ equals a) the random move in $\mf p_j$ reached by strategy profile $s$ when starting in information set $\mf p_i$ in scenario $\omega$ and random move $p_{i,\mf p_i}(\omega)$, if it exists (which is the case precisely iff $(\ast)$ holds true); b) the unique $\x\in\X^{j}$ such that $\x(\omega)$ strictly precedes $p_{i,\mf p_i}(\omega)(\omega)$ in case it exists (which is the case of the left-hand side in $(\ast)$ being true, but not the right-hand side); c) the empty set meaning that neither of the preceding two things exist, that is, in scenario $\omega$, $\mf p_j$ does not precede $\mf p_i$ nor is $\mf p_j$ reached through $s$ starting from $\mf p_i$, given the belief of being at move $p_{i,\mf p_i}(\omega)$ in $\mf p_i$. 

Moreover, $E_{i,j}$ is the set of scenarios at that information set $\mf p_i$ precedes information set $\mf p_j$, according to the ``random move beliefs'' $p_{i,\mf p_i}$ and $p_{j,\mf p_j}$; $E_{\neg i,j}$ is the set of scenarios in $D_{\mf p_j}$ at that $\mf p_i$ does not precede $\mf p_j$ in that sense, and thus, where probabilities cannot be explained by outcome generation given $\mf p_i$; and $E_{i,j}^s$ is the set of scenarios at that $\mf p_j$ is actually reached via $s$ given $\mf p_i$, according to $p_{i,\mf p_i}$. Note that, if Property \ref{def:consistent_EU_pref_str.beliefs.p} is satisfied, then, by the Heraclitus Property in Lemma \ref{lemma:Heraclitus_property},
\[ E_{i,j}^s = \{\omega\in D_{\mf p_i} \cap D_{\mf p_j} \mid \exists \x\in\mf p_j\colon\, p_{i,\mf p_i}(\omega)(\omega) \supseteq \x(\omega), ~ \psi_{j,\mf p_j} \circ \Out^s_{i,\mf p_i}(\omega) \neq \emptyset\}. \]

Thus, dynamic consistency of $(\Pi,s)$ requires that for all two information sets of some agents, which can but need not be identical, first, in all scenarios $\omega$ where the second is reached via $s$ starting from the first and given the belief about actual position within the first, the belief at the second evaluated in $\omega$ equals exactly the attained random move; second, the set $E_{i,j}$ of scenarios at that information set $\mf p_i$, according to the ``belief'' $p_{i,\mf p_i}$, precedes information set $\mf p_j$ in any way, is an event; third, there is a common prior $\P$ on $(\Omega,\ms E)$ such that for both distinct pairs $((i,\mf p_i),(j,\mf p_j))$ the posterior $\P_{j,\mf p_j}$ is equal to the conditional probability given the event that $\mf p_i$ does not precede $\mf p_j$, or it does and $s$ leads from $\mf p_i$ to $\mf p_j$, always according to the belief $p_{i,\mf p_i}$ on the realised random move in $\mf p_i$. Note that this can be a trivially void statement for one of the two pairs, but not for both because
\[ E_{\neg i,j} \cup E_{i,j}^s \cup E_{\neg j,i} \cup E_{j,i}^s = D_{\mf p_i} \cup D_{\mf p_j},\]
and $\P$ is concentrated on $D_{\mf p_i} \cup D_{\mf p_j}$. However, Part \ref{def:consistent_EU_pref_str.beliefs.common_prior} can be a void statement in some situations, for instance, if
\[ \P_{i,\mf p_i}(E_{i,j}) = 1, \qquad \P_{i,\mf p_i}(E_{i,j}^s) = 0. \]
In that case, $\P = \iota_\ast \P_{i,\mf p_i}$, where $\iota$ is the inclusion map $D_{\mf p_i} \inj \Omega$, does the job.
This corresponds to the case where, essentially, at $\mf p_i$, $\mf p_i$ is believed to precede $\mf p_j$ although the outcome of $s$ is believed to almost never reach $\mf p_j$ out of $\mf p_i$. Then, the relation between the beliefs $\P_{i,\mf p_i}$ and $\P_{j,\mf p_j}$ is unaffected by condition \ref{def:consistent_EU_pref_str.beliefs.common_prior} applied to $\{(i,\mf p_i),(j,\mf p_j)\}$.

The definition above is very general, and it simplifies in specific situations. For instance, in many situations endogenous information sets can be ordered in the weak sense that for all $i,j\in I$ and $\mf p_i\in \mf P^i$ and $\mf p_j\in\mf P^j$ admitting $\omega\in D_{\mf p_i} \cap D_{\mf p_j}$, $\hat\x_i\in\mf p_i$ and $\hat\x_j\in\mf p_j$ such that $\hat\x_i(\omega) \supseteq \hat\x_j(\omega)$, then $D_{\mf p_i}\supseteq D_{\mf p_j}$ and for all $\omega\in D_{\mf p_j}$ and $\x_j\in\mf p_j$, there is $\x_i\in\mf p_i$ with $\x_i(\omega) \supseteq \x_j(\omega)$. In that case, for distinct $(i,\mf p_i),(j,\mf p_j)$ as in the hypothesis of \ref{def:consistent_EU_pref_str.beliefs.common_prior}, we have for $(i,\mf p_i)$ preceding $(j,\mf p_j)$ as in the preceding sentence that $E_{\neg j,i} \cup E_{j,i}^s = D_{\mf p_i}$. Hence, for any $E\in\ms E|_{D_{\mf p_i}}$, 
\[ \P(E) = \P_{i,\mf p_i}(E) \cdot \P(D_{\mf p_i}) = \P_{i,\mf p_i}(E). \]
Thus, condition \ref{def:consistent_EU_pref_str.beliefs.common_prior} reduces to the unidirectional Bayesian updating rule given by
\[ \P_{j,\mf p_j}(E) \cdot \P_{i,\mf p_i}(E_{\neg i,j} \cup E_{i,j}^s) = \P_{i,\mf p_i}(( E_{\neg i,j} \cup E_{i,j}^s)\cap E), \qquad E\in\ms E|_{D_{\mf p_j}}. \]
If, moreover, perfect endogenous information is given, and $\mf p_i = \{\x_i\}$, $\mf p_j = \{\x_j\}$ with $\x_i >_\X \x_j$, we have $E_{i,j} = D_{\x_j}$, $E_{\neg i,j} = \emptyset$, so that condition \ref{def:consistent_EU_pref_str.beliefs.common_prior} reduces to the statement:
\[ \P_{j,\mf p_j}(E) \cdot \P_{i,\mf p_i}(\psi_{j,\mf p_j} \circ \Out^s_{i,\mf p_i}\neq \emptyset) = \P_{i,\mf p_i}(\{\psi_{j,\mf p_j} \circ \Out^s_{i,\mf p_i}(\omega) \neq \emptyset\}\cap E), \qquad E\in\ms E|_{D_{\mf p_j}}. \]

From this discussion it becomes obvious that dynamic consistency of belief systems is not a very strong notion if endogenous information sets are ``small'' relative to the amount of random moves eligible for forming an endogenous information set, provided the exogenous scenario space is large and belief probability measures are non-atomic, because then $\P_{i,\mf p_i}(\psi_{j,\mf p_j} \circ \Out^s_{i,\mf p_i}\neq \emptyset)$ is likely to be zero in most of the cases. 

If these conditions are satisfied, the problem is actually even deeper: there may be essentially no non-trivial way for making $\Out_{i,\mf p_i}^s$ measurable. Consider action path \textsc{sef} data with $\A = [0,1]$, $\T = \{0,1\}$, and perfect endogenous information on the exogenous scenario space given by the unit interval $\Omega = [0,1]$ with Borel $\sigma$-algebra $\ms E = \ms B([0,1])$. Suppose that the agent $i$ active at the root $\x_0$ has full information available, i.e.\ $\ms F^i_{\x_0} = \ms E$. Let $g\colon \Omega \to \A$ be the identity on $[0,1]$ and $h\colon \A\to\A$ be a non-measurable function. Then, the strategy profile $s$ mapping $\x_0$ to the random action $g$ and the move $\x_1(f)$, where $f\in\A^\T$ with $f(0) = a$, to the deterministic (alias constant ``random'') action $h(a)$, for any $a\in[0,1]$, has a problematic outcome map: $p_{i,\{\x_0\}}$ maps any $\omega\in\Omega$ to $\x_0$, and hence $\Out^s_{i,\{\x_0\}}(\omega) = (\omega,(0,\omega),(1,h(\omega)))$. If we wished to allow for utility functions on $W$ depending in a non-trivial, Borel-measurable on the third component of an outcome $w\in W = \Omega\times \A^\T$, for instance, the map $u\colon W\to [0,1]$ mapping any $w = (\omega,f)$ to $f(1)$, then $h = u\circ \Out^s_{i,\{\x_0\}}$ would have to be measurable -- which it is not.\footnote{From the purely mathematical perspective, this problem has been known for long from the perspective of randomisation. Mathematically, the preceding example is largely similar to Aumann's ``attacker--defender'' example in \cite{Aumann1964Mixed} and its version in \cite[Example 6.7]{AlosFerrer2016}. We refer to the Subsection~\ref{subs:randomisation} for a discussion of this aspect and the related literature. However, we take a somewhat different view on the problem on the level of decision theory, as explained in the following paragraph.} 

We thus clearly see that in well-posed stochastic extensive forms there is a tension between non-triviality of endogenous information (agents having information about past behaviour) and non-triviality of \textsc{eu} preference structures (non-atomic beliefs thereby requiring uncountable $\Omega$, and complex utility functions). The main point, however, is not about the existence of such a tension, but the fact that it is precisely described in terms of formal conditions and existing non-trivial examples, like action path \textsc{sef} data. Precisely, the tension does not imply that discrete structures are necessary for extensive form theory to be sensible and non-empty. On the contrary, from the preceding discussion and the remainder of the present section it becomes clear that relevant well-posed action path \textsc{sef} data with coarse endogenous information sets may well admit non-trivial \textsc{eu} preference structures, with non-atomic beliefs and interesting utility functions. Although this eventually breaks downs if endogenous information becomes finer, these observations suggests an approximation theory, by approximating \textsc{sef} with fine endogenous information with refining sequences of \textsc{sef} with coarse endogenous information. 

\subsection{Dynamic rationality and equilibrium}\label{subs:equilibrium}

Having provided a model of beliefs and tastes with natural and sufficient measurability properties, a general concept of ``dynamic rationality'' for well-posed stochastic extensive forms can be formulated (see, e.g.\ \cite[Section~9.C]{MasColell1995}).\footnote{Note that the term ``sequential'' is replaced with the more suitable and general one of ``dynamic''.} It is based on the classical idea of best response to correct beliefs about other agents' strategies, conditional on all induced decision situations given through endogenous information sets, given the specified beliefs and tastes about endogenous and exogenous events which are common knowledge. As mentioned earlier, the conditioning on induced stochastic extensive forms is supposed to model the dynamic aspect of decision making (decisions are taken at any information set). As a prime example of this, this may rule out empty threats. 
As outlined in the motivation of dynamic consistency, by adding the latter, an equilibrium concept of general scope obtains.

\begin{definition}\label{def:dynamic_rationality}
    Let $\F$ be a well-posed stochastic extensive form on an exogenous scenario space $(\Omega,\ms E)$, and let $\Pr = (\Pi,U,\ms W)$ be an \textsc{eu} preference structure on $\F$. 
    For any $i\in I$, any $\mf p\in\mf P^i$, let $\E_{i,\mf p}$ denote the conditional expectation operator with respect to $\P_{i,\mf p}$ given $\ms F_\mf p^i$, that is, 
    \[ \E_{i,\mf p} = \E^{\P_{i,\mf p}}[. \mid {\ms F}_{\mf p}^i], \]
    where $\E^\mu$ denotes the Lebesgue integral operator with respect to a given measure $\mu$. For any strategy profile $s\in S$, $i\in I$ and $\mf p\in\mf P^i$, let $\pi_{i,\mf p}(s) = \E_{i,\mf p}[u_{i,\mf p}\circ\Out_{i,\mf p}^s]. $
    
    A strategy profile $s\in S$ is said \emph{dynamically rational given $\Pr$} iff for all $i\in I$, all $\mf p\in\mf P^i$, and all $\tilde s\in S$ with $\tilde s^{-i} = s^{-i}$, we have
    \[ \pi_{i,\mf p}(s) \ge \pi_{i,\mf p}(\tilde s). \]

    Let $s\in S$ be a strategy profile. Then, $(s,\Pr)$ is said \emph{in equilibrium} iff it is dynamically rational given $\Pr$ and $(\Pr,s)$ is dynamically consistent. 
\end{definition}

As discussed in the following remark, this equilibrium concept is a generalised model of perfect Bayesian equilibrium, including several other well-known classical concepts within one extensive form framework.
\begin{remark}[Perfect Bayesian, subgame-perfect Nash and correlated equilibrium]
    Let $\F$ be a well-posed stochastic extensive form on an exogenous scenario space $(\Omega,\ms E)$ and $\Pr = (\Pi,U,\ms W)$ be an \textsc{eu} preference structure. Let $s\in S$.
    \begin{enumerate}
        \item Let us consider the static case, that is, $\X$ is a singleton. Denote its unique element by $\x$. Dynamic consistency and perfect recall is trivially satisfied, of course. In Harsanyi's setting of \cite{Harsanyi1968a}, $s$ is a ``Bayesian equilibrium point'' iff $(s,\Pr)$ is in equilibrium according to the preceding definition. In the current game-theoretic language, $(\F,\Pr)$ yields a generalisation of Bayesian games and the equilibrium property corresponds to a generalisation of Bayesian Nash equilibrium\footnote{According to \cite{MasColell1995,Escude2023Advanced}. In Aumann's setting of \cite{Aumann1974Subjectivity}, this is simply called ``equilibrium point''. Following the framework of \cite{Aumann1987Correlated}, this might be called ``Bayes rational at almost all states of the world''.} with respect to the ``information structure'' given by the exogenous scenario space $(\Omega,\ms E)$, with common prior $\P = \P_{i,\{\x\}}$ shared by all agents $i\in I$, and the individual ``signal'' $\sigma$-algebras $(\ms F_\x^i)_{i\in I}$. By weakening the dynamical consistency requirement on $(\Pi,s)$ to hold only if $i=j$, see Definition~\ref{def:consistent_EU_pref_str}, one obtains the more general case with subjective priors $\P_i = \P_{i,\{x\}}$, $i\in I$. In particular, with $(\Omega,\ms E)$ serving as a correlation device, $s$ can be seen as a correlated equilibrium, with respect to the subjective priors $\P_i$ -- though this equilibrium is typically formulated with respect to a common prior (see the discussion in \cite{Aumann1987Correlated}). As a special case, we obtain the Nash equilibrium (\cite{Nash1950}, in mixed and in pure strategies, see Subsection~\ref{subs:randomisation} for a possible meaning of this).
        \item The preceding point generalises to the dynamic setting. Then, we obtain a non-trivial generalisation of perfect Bayesian equilibrium for a generalised model of dynamic Bayesian games (see, e.g.\ \cite{Fudenberg1991,MasColell1995}). It is a non-trivial generalisation of Bayesian games and perfect Bayesian equilibrium, because general exogenous dynamic uncertainty can be handled that can not be seen as the outcome of a nature agent's decision making (for instance, Brownian noise in continuous time). 
        \item Under the hypothesis of perfect information, we obtain the subgame-perfect equilibrium (\cite{Selten1965}). If only perfect endogenous information and perfect recall are supposed instead of perfect information, we obtain a non-trivial generalisation of subgame-perfect equilibrium to the case of imperfect exogenous information.    
    \end{enumerate}
\end{remark}

Next, we comment on the multiple-selves approach to dynamic decision making.

\begin{remark}[Multiple selves]\label{rmk:multiple_selves_sef}
    Dynamic rationality rules out empty threats, eventually corresponding to certain plans that are optimal today, but suboptimal tomorrow.\footnote{More precisely, suboptimal in some possible future which the threat -- non-credibly -- attempts to deter others from bringing about.}  
    In the case of a dynamically inconsistent taste system, it may also rule out plans that are suboptimal today, but optimal tomorrow.
    
    A classical attempt for resolving both points is the so-called \emph{multiple selves} approach (going back to \cite{Strotz1955Myopia}, at least; see \cite{Piccione1997Interpretation,Gilboa1997Comment} in the context of the absent-minded driver story). One idea behind this says that strategically it does not make much sense to suppose an agent making contingent plans for the future, for its future self has the ability to revise it. In the inconsistent case, he might really do this. From that perspective, by definition, an agent should only act once. According to this principle, the following analysis can be performed in our setting.

    Let $\F$ be a well-posed stochastic extensive form on an exogenous scenario space $(\Omega,\ms E)$ and $\Pr = (\Pi,U,\ms W)$ be an \textsc{eu} preference structure. Let $s\in S$.
    
    Replace $\F$ with the stochastic extensive form 
    \[ \F = (F,\pi,\X,\hat I,\hat{\ms F},\hat{\ms C},\hat C),\]
    and $\Pr = (\Pi,U,\ms W)$ with $\hat{\Pr} = (\hat \Pi,\hat U,\ms W)$,
    where 
    \begin{itemize}[label=--]
        \item $\hat I =\{(i,\mf p) \mid i\in I,\, \mf p\in\mf P^i\}$;
        \item $\hat{\ms F}^{(i,\mf p)}_\x = \ms F^i_{\mf p}$ for all $i\in I$, $\mf p\in\mf P^i$, $\x\in\mf p$;\footnote{One can show that with the choices $\hat C$ defined below, the set of random moves for $(i,\mf p)$ is given by $\hat \X^{(i,\mf p)} = \mf p$. We recall that it does not matter how $\ms F$ and $\ms C$ are chosen outside of this set. This footnote hence also applies to the following point.}
        \item $\hat{\ms C}^{(i,\mf p)}_\x = \ms C^i_{\mf p}$, for all $i\in I$, $\mf p\in\mf P^i$, $\x\in\mf p$;
        \item $\hat C^{(i,\mf p)} = A^i(\mf p)$, for all $i\in I$, $\mf p\in\mf P^i$;
        \item $\hat\Pi$ is the family assigning to any $(i,\mf p)\in \hat I$ and the only endogenous information set of that agent, $\mf p$, the triple $(p_{i,\mf p},\ms P_{i,\mf p},\P_{i,\mf p})$;
        \item $\hat U$ is the family assigning to any $(i,\mf p)\in \hat I$ and the only endogenous information set of that agent, $\mf p$, the function $u_{i,\mf p}$.
    \end{itemize}
    That is, each agent is split into separate ``incarnations'' of itself at all of its endogenous information set. We do not bother the reader with the verification of the claim that this yields again a well-posed stochastic extensive form on $(\Omega,\ms E)$ and an \textsc{eu} preference structure on it. Note that the new structure trivially satisfies perfect recall and dynamic consistency of the taste system. 

    Now, dynamic rationality and equilibrium as defined above can be analysed in $(\hat\F,\hat\Pr)$.
\end{remark}

The expected utility framework, as implemented above, is attractive for several reasons, including its equivalence to a class of well-chosen rationality axioms on decision making. Nevertheless, it has been challenged on diverse grounds: in view of explicit experimental evidence, for its theoretical lack of aversion of Knightian uncertainty, for the assumption of complete preferences, for its tendency to overlook gain-loss-asymmetry, and other things (\cite{Ellsberg1961Uncertainties, Aumann1962Utility, Kahneman1979Prospect, Gilboa1989Maxmin, Schmeidler1989Subjective, Bewley2002Knightian, Hara2023Multiple}, for overviews see \cite{Gilboa2009Theory, Etner2012Decision}). Stochastic extensive form theory is not made for expected utility, and alternative preference structures or ways of modelling decision making may be defined on it as well. The author has chosen the classical expected utility framework because of its historic importance, its current use and the fact that the proposed alternatives often depart from it.

At this point, the author wants to mention that the terms ``decision problem'' and ``game'' are not introduced formally because their use in the literature is not unambiguous, and probably for good reasons. Sometimes, these terms just describe a phenomenon of strategic interaction or a non-formalised description of it (e.g.\ one can know the absent-minded driver's decision problem without knowing its extensive form description using graphs as proposed in \cite{Piccione1997Interpretation}). Sometimes, these terms are meant as a formal description of states, consequences, agents and choices (as in \cite{AlosFerrer2005} where the term ``extensive decision problem'' is used for the classical pseudo extensive forms of the present text). Sometimes, these terms are used only if, in addition to the formal description of states, consequences, agents and choices, payoffs, utilities, beliefs, preferences or the like are provided (as typically the case in stochastic control theory). In other contexts, one might be even more radical and argue that in order to be a game it must be played, that the play must be in equilibrium, and that there cannot be any other equilibrium (because the model could not explain why the latter is not played). Implicitly, such an argument also requires well-posedness. From this perspective, a game (with stochastic extensive form characteristics) would be a well-posed stochastic extensive form equipped with an equilibrium concept such there exists a unique equilibrium. The present text does not suggest any of these views being superior compared to others, and its concepts are compatible with all of them.

\subsection{Randomisation}\label{subs:randomisation}

Randomisation, that is, extending the exogenous scenario space and the information structure on it, is a common procedure in many fields. It underlies Nash's idea of equilibrium existence by introducing lotteries over strategies (the latter called ``pure'' because deterministic, that is, essentially defined on a singleton exogenous scenario space); it is a fundamental way of representing choice under uncertainty through the theory of expected utility following \cite{Neumann1944, Savage1972Foundations, Anscombe1963Definition, Aumann1974Subjectivity}; it underlies the implementation of decision rules as Bayesian equilibria and the revelation principle in mechanism design; it underlies the weak solution concept for stochastic differential equations, stochastic control problems and stochastic differential games. Mixed (and then also behaviour and pure strategies) can be defined in the spirit of \cite{Aumann1974Subjectivity}.
Certainly, this is not the right place to develop a general theory of randomisation of stochastic extensive forms. Let us restrict to some comments.

In the non-stochastic setting of finite games à la Nash, randomness enters only for ``mixing'' individual strategies and not as a correlation device. This typically happens under the common prior assumption, and so all strategies are objective. When one assumes that under the common prior the exogenous information structures of different agents are mutually independent, all strategy profiles are necessarily mixed. If, moreover, one assumes that the exogenous information available at different endogenous information sets is mutually independent, then any strategy profile is behavioural. 

By contrast, the point of most stochastic games lies in correlation given, for instance, by what one calls a state process that agents may have partial information about. Then pure, behavioural and mixed strategies are rather the exception than the rule. This is even more so if one does not assume common priors (see again the discussion in \cite{Aumann1974Subjectivity,Aumann1987Correlated}). Yet, similarly to the content of \cite[Assumption II]{Aumann1974Subjectivity}, one might wish $\ms E$ and $\Pi$ such as to assure reasonable options of behavioural strategy profiles $s$ (for instance, such that $\P_{i,\mf p}$ is non-atomic on $\ms F_{\mf p}^{i,s}$ for all $i\in I$, $\mf p\in\mf P^i$). 

Yet, the ``more'' randomisation is possible, the finer must be $\ms E$, hence, the less probability measures on it. 
Actually, this problem is not only related to randomisation, but also to tastes: the less trivial taste shall be, the finer must be $\ms W$, thus the finer must be $\ms P_{i,\mf p}$ and $\ms E$, and again, the smaller must be the set of admissible beliefs $\P_{i,\mf p}$. 
It is well-known from the case of classical extensive forms that without finiteness (or countability) assumptions, this creates a non-trivial trade-off between the richness of both randomisation procedures and tastes and this trade-off is well understood (Aumann in \cite{Aumann1961Borel,Aumann1963Choosing,Aumann1964Mixed}, also see the discussion by Alós-Ferrer and Ritzberger in \cite[Subsection~6.4.3]{AlosFerrer2016}, in particular Example 6.7 therein). In Subsection~\ref{subs:EU_preferences}, we have re-interpreted this trade-off in terms of a tension between non-triviality of endogenous information and non-triviality of \textsc{eu} preference structures.

\subsection{Simple examples}

Let us first consider some equilibria for a randomised version of the stochastic extensive forms $\F$ from Subsection~\ref{subs:SEF_simple-examples}, Lemma~\ref{lemma:simple_sef1}. Obviously, already in this simple case, there is a large range of possibilities, which moreover are well-known and analysed in other formalisations. Thus, the following selection can only be illustrative of the generality of our approach, and attempts to clarify its mechanics.

\begin{example}
    We consider a randomised version in that we let $(\Omega,\ms E)$ be a general exogenous scenario space and $\rho\colon \Omega \to \{1,2\}$ an $\ms E$-$\mc P\{1,2\}$-measurable surjection. We suppose it rich enough to admit real-valued random variables $\xi^k$, $k=0,1,2$ and a probability measure under which these four random variables are independent, $\rho$ is $\frac 12$-Bernoulli-distributed, and each $\xi^k$ is uniformly distributed on $[0,1]$. The $\xi^k$ represent randomisation devices available at the random moves $\x_k$, respectively.
    
    We let $I$ be a singleton, $\T = \{0,1\}$, $\A = \{1,2\}$, and $W = \Omega \times \A^\T$, and consider the associated action path \textsc{sdf} $(F,\pi,\X)$ which gives us the \textsc{sdf} from the first simple example (see \cite[Lemma~2.17]{Rapsch2024DecisionA}) if $\rho$ is also injective and $\omega_k$ denotes the preimage of $k\in\{1,2\}$ under $\rho$. Note that an element of $W$ is essentially a triple $(\omega,k,m) \in \Omega \times \{1,2\}\times\{1,2\}$. In reminiscence of the notation used for the simple example, denote the root of $(\X,\ge_\X)$ by $\x_0$ and let, for $k\in\{1,2\}$, $\x_k$ be a shorthand for the random move mapping any $\omega\in\Omega$ to $\x_k(\omega) = \Omega \times \{k\} \times\{1,2\}$ (which is nothing else than $\x_1(f)$ for any $f\colon \T\to \A$ with $f(0) = k$ when written in action path notation). 
    
    In view of Remark \ref{rmk:multiple_selves_sef} on the multiple selves approach, we let $\hat I = \{i,j\}$ be a set with two distinct elements, representing two agents, $i$ acting at time $0$, and $j$ at time $1$ -- so to speak, two independent copies of the unique action index $\in I$ above. Consider the map $u\colon W\to \R$ given by $u(\omega,f) = (-1)^{\rho(\omega) + f(0) + f(1)}$, which is a generalised form of the payoff known from ``matching pennies'', with two pennies whose two sides show $1$ and $2$, and where the fact whether matching means win or lose depends on the realisation of the random variable $\rho$. 
    
    Suppose that the taste system $U$ satisfies with $u_i = -u$ and $u_j = u$. We fix a prior belief, alias a probability measure on $(\Omega,\ms E)$, $\P$ such that $\rho$ and $\xi^k$, $k=0,1,2$, are $\P$-independent, and $\P(\rho = 1) = p$ for given $p\in [0,1]$. Given a strategy profile $s\in S$, let us call an \textsc{eu} preference structure $\Pr$ \emph{suitable} iff its taste system is $U$, $\P_{i,\{\x_0\}} = \P$ and $(\Pi,s)$ is dynamically consistent. Note that such a $\Pi$ can be constructed out of the data $\P$ and $s$ in any case, though not necessarily in a unique way. Also note that for any $s\in S$ and $\x\in\X$, $s^i(\x)$ can be identified with its random action alias an $\ms F^i(\x)$-measurable, $\A$-valued random variable. The independence of the signals means that we only consider ``mixed'' strategies in the traditional sense (no correlation) which is not a much of a restriction here anyway because of the zero-sum structure.

    \begin{enumerate}
        \item First suppose that the agents' information and choices are essentially given by the first line of the table defining $C$ in Subsection~\ref{subs:SEF_simple-examples}, so that, in particular, $\ms F$ essentially corresponds to case 1, i.e.\ no information about $\rho$, and, at time $1$, agent $j$ recalls $i$'s decision made at time $0$. That is, $\ms F^i_{\x_0} = \sigma(\xi^0)$, $\ms F^j_{\x_k} = \sigma(\xi^k)$, for $k=1,2$, and $\mc H^j_1 = \{\{(0,1)\},\{(0,2)\}\}$. 
    
        Let $s\in S$ be a strategy profile. 
        If $p=\frac 12$, then for any $s\in S$ and any suitable \textsc{eu} preference structure $\Pr$, $(s,\Pr)$ is in equilibrium, with expected utility zero. If $p>\frac 12$, then for $s\in S$ and any suitable \textsc{eu} preference structure $\Pr$, $(s,\Pr)$ is in equilibrium iff
        \[ \P_{j,\{\x_1\}}[s^j(\x_1) = 2] = \P_{j,\{\x_2\}}[s^j(\x_2) = 1] = 1. \]
        Its expected utility is $1-2p$ for $i$ and $2p-1$ for $j$, since $j$ can fully react to $i$'s action whatever the latter is (even if $i$ did randomise and even though $j$ cannot observe $i$'s personal randomisation signal $\xi^0$). 
        If $p<\frac 12$, conversely, then for $s\in S$ and any suitable \textsc{eu} preference structure $\Pr$, $(s,\Pr)$ is in equilibrium iff
        \[ \P_{j,\{\x_1\}}[s^j(\x_1) = 1] = \P_{j,\{\x_2\}}[s^j(\x_2) = 2] = 1. \]
        Its expected utility is $2p-1$ for $i$ and $1-2p$ for $j$.
        
        \item If we follow the second line of the table, the exogenous information structure remains similar, but agent $j$ cannot observe the choice made by agent $i$ -- which is more similar to the original ``matching pennies'' game. Then, again, $\ms F^i_{\x_0} = \sigma(\xi^0)$, but $\ms F^j_{\x_1} = \ms F^j_{\x_2} = \sigma(\xi^1,\xi^2)$ and $\mc H^j_1 = \{\{0\} \times\A\}$, as the agent receives the same signal at $\x_1$ and $\x_2$. Let $\mf p = \{\x_1,\x_2\}$.
        
        Let $s\in S$ and $\Pr$ be a suitable \textsc{eu} preference structure. If $p=\frac 12$, then $(s,\Pr)$ is in equilibrium without further restriction. If $p\neq\frac 12$, then $(s,\Pr)$ is in equilibrium iff
        \[ \P[s^i(\x_0) = 1] = \P_{j,\mf p}[s^j(\mf p) = 1] = \frac 12. \]
        In any case, equilibrium expected utility is equal to $0$ for both agents.

        \item Let us next consider the situation where agent $j$ has a further advantage by having full exogenous information, which corresponds to line 4, and in particular \textsc{eis} 2.(a). That is, $\ms F^i_{\x_0} = \sigma(\xi^0)$, but $\ms F^j_{\x_1} = \ms F^j_{\x_2} = \sigma(\rho,\xi^1,\xi^2)$ and $\mc H^j_1 = \{\{0\} \times\A\}$. Let again $\mf p = \{\x_1,\x_2\}$.

        Let $s\in S$ and $\Pr$ be a suitable \textsc{eu} preference structure. Let $\alpha = \P[s^i(\x_0) = 1]$ and for $k=1,2$ let $\beta_k\in [0,1]$ a number such that $\P[s^j(\mf p) = 1 \mid \rho = k] = \beta_k$, $\P$-almost surely. Then the expected utility of $i$, given $s^i(\x_0)$ is equal to
        \[ \pi_{i,\{\x_0\}}(s) = -(-1)^{s^i(\x_0)}\,\big(p(2\beta_1-1)+(1-p)(1-2\beta_2)\big), \]
        and the expected utility of $j$ given $\rho$ and $s^j(\mf p)$ is equal to
        \[ \pi_{j,\mf p}(s) = (1-2\alpha) (-1)^{\rho + s^j(\mf p)}. \]
        
        We conclude that $(s,\Pr)$ is in equilibrium iff 
        \[ \alpha = \frac 12, \qquad p(2\beta_1-1)+(1-p)(1-2\beta_2) = 0. \]
        For if $\alpha = \frac 12$, then any strategy for $j$ is a best response, but mixing is a best response for $i$ only if the right-hand side equality is satisfied. If $\alpha >\frac 12$, then $s^j$ is a best response iff $s^j(\mf p) = 3-\rho$, $\P$-almost surely. In this case, $\beta_1 = 0$ and $\beta_2 = 1$, so that $p(2\beta_1-1)+(1-p)(1-2\beta_2) = -1$. But then any best response $s^i$ by $i$ satisfies $s^i(\x_0) = 2$ $\P$-almost surely, whence $\alpha = 0$, a contradiction. A similar contradiction arises when assuming equilibrium with $\alpha < \frac 12$.

        Expected utility in equilibrium is almost surely zero for both agents. Hence, the informational advantage for $j$ (as compared to the preceding point) does not pay out.

        \item Let us now vary the preceding examples and consider a case where agent $i$ can determine what exogenous information agent $j$ obtains, which is a variation on the theme of exploration and exploitation. For this, we transcribe the fifth line with \textsc{eis} 2.(b). That is, $\ms F^i_{\x_0} = \sigma(\xi^0)$, but $\ms F^j_{\x_1} = \sigma(\rho,\xi^1)$, $\ms F^j_{\x_2} = \sigma(\xi^2)$ and $\mc H^j_1 = \{\{(0,1)\},\{(0,2)\}\}$.

        Let $s\in S$ and $\Pr$ be a suitable \textsc{eu} preference structure. Let $\alpha = \P[s^i(\x_0) = 1]$, and for $k=1,2$ let $\beta_k\in [0,1]$ a number such that $\P_{j,\{\x_1\}}[s^j(\x_1) = 1 \mid \rho = k] = \beta_k$, $\P_{j,\{\x_1\}}$-almost surely, and let $\gamma = \P_{j,\{\x_2\}}[s^j(\x_2) = 1]$. Expected utilities are given by
        \begin{align*}
            \pi_{i,\{\x_0\}}(s) =  - (-1)^{s^i(\x_0)}\, \Big( &1\{s^i(\x_0)=1\}\,\big(p(2\beta_1-1)+(1-p)(1-2\beta_2)\big) \\
            &+ 1\{s^i(\x_0)=2\}\,(1-2\gamma)(1-2p)\Big)
        \end{align*}
        for $i$, and
        \[ \pi_{j,\{\x_1\}}(s) = -(-1)^{\rho+s^j(\x_1)}, \qquad \pi_{j,\{\x_2\}}(s) = (1-2p)(-1)^{s^j(\x_2)} \]
        for $j$.

        We conclude that $(s,\Pr)$ is in equilibrium iff
        \[  \alpha \cdot p(1-p) = 0, \qquad \beta_1=0,~\beta_2=1, \qquad \begin{cases}
            \gamma = 0, &\text{ if } p < \frac 12, \\
            \gamma = 1, &\text{ if } p > \frac 12.
        \end{cases}. \]
        Note that $\alpha \cdot p(1-p) = 0$ means nothing else than: if $\rho$ is believed to be truly random, then $\alpha = 0$. In the deterministic case ($p=0$ or $p=1$), no restriction on $\alpha$ needs to be imposed; agent $i$ is completely indifferent in equilibrium.
        
        The proof is straightforward by backwards induction. The best response property of $s$ for agent $j$ implies that $\beta_1=0$ and $\beta_2=1$, and
        \[ \begin{cases}
            \gamma = 0, &\text{ if } p < \frac 12, \\
            \gamma = 1, &\text{ if } p > \frac 12.
        \end{cases}.\]
        Then, given this strategy of agent $j$, the expected utility of agent $i$ writes as
        \[  \pi_{i,\{\x_0\}}(s) =  - (-1)^{s^i(\x_0)}\, \big( -1\{s^i(\x_0)=1\}\,+ 1\{s^i(\x_0)=2\}\,|1-2p|\big) = -|1-2p|^{s^i(\x_0)-1},  \]
        and the best responses of $i$ are obviously the ones claimed above.

        Expected utility in equilibrium is almost surely $-|1-2p|$ for $i$ and $|1-2p|$ for $j$. If $j$ had access to full information about $\rho$ at $\x_1$ as well, it would be easy to see, $j$ could realise the expected utility of $1$ in equilibrium. In the present asymmetric case, however, provided $0<p<1$, agent $i$ can force $j$ to the random move with less exogenous information, reducing $j$'s expected equilibrium utility. Hence, $j$ cannot realise more expected utility in equilibrium than in the situation with no relevant information on $\rho$ at all (see the first case above). 
        \end{enumerate}
\end{example}

To close this subsection, let us consider the absent-minded driver stochastic extensive form discussed earlier, see Theorem~\ref{thm:absent_minded_driver_Gilboa_sef}. Following \cite{Piccione1997Interpretation}, we let taste $u$ of both agents by identically given by
\[ (\ast) \qquad u(\omega,D) = 0, \quad u(\omega,H) = 4, \quad u(\omega,M) = 1, \]
for $\omega\in\Omega$, and in equilibrium we suppose some form of dynamic consistency of strategy profile and belief. Here, we use the equilibrium notion introduced in the present text, restricting to those belief systems $\Pi$ whose common prior for $\{(1,\{\x_1\}),(2,\{\x_2\})\}$ is given by a fixed probability measure $\P$ with respect to that $(\rho,\xi^1,\xi^2)$ are independent and $\xi^1,\xi^2$ are both uniformly distributed on $[0,1]$.

Note that for any such $\P$ and any strategy profile $s\in S$, there is a belief system $\Pi$ such that $(\Pi,s)$ is dynamically consistent and $\P$ is a common prior for $\{(1,\{\x_1\}),(2,\{\x_2\})\}$. Moreover, $\Pi$ then necessarily satisfies $(\xi^i)_\ast\P_{i,\{\x_i\}} = (\xi^i)_\ast \P$, for both $i\in I$; or in other words:
\[ \P_{i,\{\x_i\}}(E) = \P(E), \qquad \text{ for all } E\in \ms F^i_{\x_i}. \]
$\Pi$ is essentially uniquely determined (i.e.\ ``along the equilibrium path'').

Also note that a strategy profile $s\in S$ corresponds to a pair $(E_1,E_2)$ of events $E_i\in\ms F_{\x_i}^i$ via $s^i(\x_i) = c_i(E_i)$, $i\in I$, and this correspondence is a bijection $S \to \ms F_{\x_1}^1 \times \ms F_{\x_2}^2$. In the sequel, we thus identify strategy profiles with their respective images under this bijection. $E_i$ is the event of agent $i$ exiting, $i\in I$.

Hence, fix $u$ and $\P$ as above. Given $s\in S$, represented by $(E_1,E_2)$ as above, and $p\in[0,1]$, call an \textsc{eu} preference structure \emph{suitable} iff $u_{i,\{\x_i\}} = u$ for both $i\in I$, $(\Pi,s)$ is dynamically consistent, and
\[ \P(\rho = 1) = \frac 12, \qquad \P(E_1) = \P(E_2) = 1-p. \]
Thus, $\frac 12$ is the prior probability of $\{\rho = 1\}$, and $p$ is the prior probability of continuing at either intersection, given the intersection. In particular, we deliberately focus on equilibria with 1) prior $\P$ under that $\rho$ is uniformly distributed, and 2) $\P(E_1) = \P(E_2)$. This selection of equilibria is natural in view of the symmetry of the problem, all the more in the case of a physical person split into two abstract agents. The indifference principle (a.k.a.\ principle of insufficient reason) further underpins that choice.

\begin{thm}\label{thm:absent_minded_driver_Gilboa_sef_well-posed}
    Consider the absent-minded driver \textsc{sef} introduced previously, defined on exogenous scenario space $(\Omega,\ms E)$. Let $p\in[0,1]$, let $s\in S$ be a strategy profile and let $\Pr$ be a suitable \textsc{eu} preference structure. Then, $(s,\Pr)$ is in equilibrium iff $p = \frac 23$.
\end{thm}

Hence, we find that the absent-minded driver paradox from \cite{Piccione1997Interpretation} evaporates, confirming the \emph{ex ante} optimal strategy of continuing independently with probability $\frac 23$ at each intersection, which is in a similar way the conclusion in \cite{Aumann1997Absent,Gilboa1997Comment} as well.

\subsection{Well-posedness of action path stochastic extensive forms}

Action path \textsc{sef} data provide a large class of well-posed stochastic extensive forms. In this framework, the three order-theoretic properties characterising well-posedness of a stochastic extensive form are implied by well-orderedness of time, as the following result shows in high generality.

\begin{thm}\label{thm:AP_sef_well-posed}
    Let $\F$ be the stochastic extensive form induced by action path \textsc{sef} data with well-ordered time $\T$. Then, $\F$ is well-posed.
\end{thm}

This theorem has important implications. Stochastic decision problems and games in discrete time or with time structure like in the long cheap talk game can be equipped with general stochastic noise, within an extensive form theory. On the other hand, nothing is said about continuous time; and indeed, in \cite{AlosFerrer2008}, it has been shown that certain classical action path pseudo-extensive forms in continuous time are not well-posed. This however does not preclude any action path classical or stochastic extensive form in continuous time from being well-posed, as shown in \cite{AlosFerrer2015}. Thus well-ordered time is a sufficient, but not a necessary condition for well-posedness. This is tightly related to the question of the decision-theoretic meaning of ``stochastic differential games''. These two points are discussed in the third part of the present series.

\subsection{Discussion on nature representations}

Traditionally, stochastic games in discrete time are modelled in what the author suggests calling \emph{nature representation}. That is, in a traditional extensive form model, using an additional agent called ``nature'' with perfect recall using a fixed mixed strategy. In the terminology of this article, this would entail an action path stochastic extensive form, in order to allow for that randomisation, and an agent with perfect recall. Provided the action space is regular enough (a Borel space), in that setting, any probability measure on discrete-time path space can be represented as the outcome of ``nature's'' strategy. This follows from classical disintegration results in probability theory (see \cite[Chapter~3]{Kallenberg2021}).

In continuous time, however, this representation can fail because the nature representation does not always yield a well-posed stochastic extensive form. There are exceptions: for instance, locally right-constant jump processes can be implemented like this (see the third paper for a remark on how to do so using stochastic action-reaction \textsc{sef}). An example of failure, however, is Brownian motion. Brownian motion is supported on paths of low regularity, paths that are far from being locally right-constant. Hence, a straightforward nature representation would suggest to take the action path pseudo-extensive form data $W=\Omega\times\A^\T$, for a suitable exogenous scenario space $(\Omega,\ms E)$, $\A=\R$, $\T=\R_+$, and singleton $I = \{\text{``nature''}\}$, such that ``nature'' admits perfect endogenous recall. By Theorem \ref{thm:link_endogenous_information_H}, perfect (endogenous) recall requires $\mc H^{\text{``nature''}}_t$ to contain only singletons, $t\in\T$. The corresponding decision forest is not weakly up-discrete, and, indeed, using well-known counterexamples from \cite{Simon1989, Stinchcombe1992, AlosFerrer2016}, one can easily see that the induced action path pseudo-\textsc{sef} is not well-posed.

This issue might be circumvented by accepting the ``nature'' agent to admit imperfect recall. For instance, in the just-mentioned example one may instead consider the other extreme, i.e.\ assume $\mc H^{\text{``nature''}}_t$ to be a singleton for all $t\in\T$. Then, we easily obtain well-posedness (compare \cite[Example~5.5]{AlosFerrer2016}, which can be directly adapted for a single agent). However, while the modeller might fix ``nature's'' strategy, he or she cannot fix the personal agents' ones. Strategies would be functions of information sets, and therefore, adaptedness of ``personal'' agents' strategies would lack an explanation in the case of the non-discrete ``nature'' action space $\A$. While it seems already odd that the ``nature'' agent is forced to be forgetful on the grounds of extensive form well-posedness, we argue that the category of perfect recall is irrelevant for generating exogenous randomness or information in the first place. Moreover, the nature representation does not help to explain adaptedness of the actual, ``personal'' agents' strategic behaviour. We conclude that stochastic extensive forms provide a strict generalisation of existing extensive form theory: it contains classical extensive forms, and although many stochastic extensive forms can be nature-represented in classical extensive form, this is not the case for an important class of extensive forms with interesting noise, such as Brownian motion.

\section*{Conclusion}
\addcontentsline{toc}{section}{Conclusion}

It is possible to implement general stochastic processes as background noise on refined partitions-based extensive forms without encountering outcome generation problems for a ``nature'' agent, while allowing for a rigorous decision-theoretic interpretation of the relationship between endogenous and exogenous information and choices. This is an improvement over the existing state of the art in both refined partitions-based decision and game theory (e.g.\ \cite{AlosFerrer2016,AlosFerrer2015}) and stochastic control and differential games theory (e.g.\ \cite{Pham2009Continuous,Karatzas1998Methods,Carmona2018,Cohen2015Stochastic}). This has been achieved by abandoning the assumption of a ``nature'' agent and instead constructing a theory of stochastic extensive forms based on stochastic decision forests. This provides a consistent theory of refined partitions-based dynamic choice under uncertainty, adapted to a mechanism of exogenous information revelation that generalises filtrations from probability theory. Moreover, the measurability assumptions on choices typically made in the literature can be understood as a conceptual necessity rather than a technical one. Randomisation -- whether behavioural, mixed and correlated -- arises naturally and can be explained within this context. 

The stochastic extensive form defines the rules. Then, given arbitrary ``subgames'' -- identified as random moves and closely linked to closed (random) histories -- agents can freely choose, using elementary Savage acts as strategies, thereby generating outcomes. Well-posed stochastic extensive forms are ultimately those for which this structure becomes fully meaningful. Well-posedness can be classified transparently, analogous to the results found in \cite{AlosFerrer2008}. Furthermore, dynamic rationality and a generalised form of perfect Bayesian equilibrium can be defined on this basis, providing an abstract, general, and interpretable solution concept. We note that stochastic extensive forms resemble a generalisation of dynamic Bayesian games; however, a key aspect of the former is their independence from any nature representation. As a consequence, they can accomodate complex exogenous randomness, such as Brownian motion and beyond. 

The present article demonstrates that general probability, the extensive form and choice under uncertainty can be unified into a simple, abstract concept of broad generality, requiring minimal special structure to formulate familiar concepts of rationality and equilibrium in a consistent manner. Moreover, we observe that the boundaries to this abstract concept are delimited in terms of a trade-off: between the flexibility of choice and the richness of information on the one hand, and well-posedness and the existence of appropriate (e.g.\ \textsc{eu}) preference structures and solution concepts on the other.

The generality of this idea is emphasised by the many applications covered. In the text many simple pedagogical examples are provided. Moreover, we note that the theory allows to model the absent-minded driver's story without decision-theoretic paradox, very much as in Gilboa's approach to the problem. However, the universality of the theory is perhaps best highlighted by the general model of action path stochastic extensive forms presented in this article, defined by a small number of easily verifiable conditions. This model can serve as a unified decision-theoretic foundation for a broad class of stochastic decision problems, particularly in well-ordered time. Moreover, as will be demonstrated in the third paper, it constitutes one step toward explaining stochastic decision problems in continuous time approximately.

In the literature on continuous-time and differential games, whether stochastic or not, there has been discussion on how the concept of subgame-perfect Nash equilibrium can be faithfully implemented. A particular example of this is the timing game, and a well-known approach to addressing this issue has been developed in \cite{Fudenberg1985,Riedel2017,Steg2018}. Since this strand of the literature avoids using extensive forms, ad hoc definitions of this equilibrium concept are employed, which, from a decision-theoretic perspective, do not constitute an equilibrium concept of an extensive form game in continuous time whatsoever, but rather one based on stacked strategic form games, justified by an approach in ``discrete time with an infinitesimally fine grid'' (\cite{Simon1987,Simon1989}). Moreover, in the stochastic case, the question arises as to whether one should be able to condition on histories stopped at stopping times with respect to the filtration-like exogenous information available, a question which has been identified and addressed in \cite{Riedel2017}. Another open question is why this is decision-theoretically appropriate, and whether an answer can be given in terms of extensive form theory.

In the third paper of the present series, these issues will be investigated from the perspective of stochastic extensive form theory, using the action path formulation. This approach promises to provide a foundation for subgame-perfect -- and, from a more general stochastic perspective, perfect Bayesian equilibrium -- in continuous time, within a general stochastic setting, in a specific approximate sense, but strictly stronger than what has been done in the cited works. Instead of interpreting ``continuous-time games'' as ``games'' in ``discrete time with an infinitesimally fine grid'', they will be viewed as limits of well-posed games in continuous time. This approximation will apply at the level of trees and choices, not merely of payoffs. While payoffs are derived objects, trees and choices are primitives and constitute the crucial decision-theoretic components. Only such a procedure can demonstrate why a concept like that of stochastic differential games is well-placed within the framework of stochastic extensive forms -- be it by a limiting process. Therefore, such an approximation will be instrumental in understanding the decision making process described by the model.

\section*{Acknowledgements}
\addcontentsline{toc}{section}{Acknowledgements}
First and foremost, the author is particularly indebted to his doctoral supervisor, Christoph Knochenhauer, who read and commented on several earlier versions of this paper, encouraged him to pursue this project, and actively supported him in presenting its contents at workshops. Special thanks are due to Frank Riedel for the helpful discussion we had in Berlin in 2022, especially for providing useful insights regarding the game-theoretic literature. The author also wants to thank Elizabeth Baldwin, Miguel Ballester, and Samuel N.\ Cohen for inspiring discussions on decision and control theory during his stay in Oxford in 2023. The author is also grateful to the attendants of the workshops and seminars where he had the opportunity to present and discuss earlier versions of this particular project, in Oxford, Kiel, Berlin, Palaiseau, including Daniel Andrei, Karolina Bassa, Peter E.\ Caines, Fanny Cartellier, Sebastian Ertel, Ali Lazrak, Kristoffer Lindensjö, Christopher Lorenz, Berenice Anne Neumann, Manos Perdikakis, Jan-Henrik Steg, Peter Tankov, and Jacco Thijssen, for their questions and comments. 
Partial funding by Deutsche Forschungsgemeinschaft (DFG) through, first, the \href{https://gepris.dfg.de/gepris/projekt/410208580}{IRTG 2544} Stochastic Analysis in Interaction, Project ID: 410208580, and, second, under Germany's Excellence Strategy, the \href{https://gepris.dfg.de/gepris/projekt/390685689}{EXC 2046/1} The Berlin Mathematics Research Center MATH+, Project ID: 390685689, is gratefully acknowledged. 

%% file: contents/appendix.tex
\section{Figures}

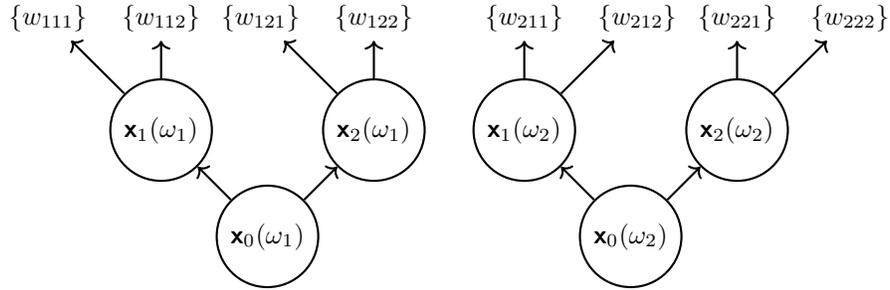
\begin{figure}[h!]
    \centering
    \begin{tikzpicture}[node distance={20mm}, thick, main/.style = {draw, circle}] 
    \node[main] (1) {$\x_0(\omega_1)$}; 
    \node[main] (2) [above left of=1] {$\x_1(\omega_1)$}; 
    \node[main] (3) [above right of=1] {$\x_2(\omega_1)$}; 
    \node[main] (4) [right of=3] {$\x_1(\omega_2)$}; 
    \node[main] (5) [below right of=4] {$\x_0(\omega_2)$};
    \node[main] (6) [above right of=5] {$\x_2(\omega_2)$};
    \node[] (8) [above=0.5cm of 2] {$\{w_{112}\}$};
    \node[] (7) [left=0.2cm of 8] {$\{w_{111}\}$};
    \node[] (10) [above=0.5cm of 3] {$\{w_{122}\}$};
    \node[] (9) [left=0.2cm of 10] {$\{w_{121}\}$};
    \node[] (11) [above=0.5cm of 4] {$\{w_{211}\}$};
    \node[] (12) [right=0.2cm of 11] {$\{w_{212}\}$};
    \node[] (13) [above=0.5cm of 6] {$\{w_{221}\}$};
    \node[] (14) [right=0.2cm of 13] {$\{w_{222}\}$};
    \draw[->] (1) -- (2); 
    \draw[->] (1) -- (3); 
    \draw[->] (5) -- (4); 
    \draw[->] (5) -- (6); 
    \draw[->] (2) -- (7); 
    \draw[->] (2) -- (8); 
    \draw[->] (3) -- (9); 
    \draw[->] (3) -- (10); 
    \draw[->] (4) -- (11); 
    \draw[->] (4) -- (12); 
    \draw[->] (6) -- (13); 
    \draw[->] (6) -- (14); 
    \end{tikzpicture} 
    \caption{A simple stochastic decision forest represented as a directed graph, with $w_{\ell km} = (\omega_\ell,k,m)$, for $(\ell,k,m)\in \{1,2\}^3$. Moves are indicated by circles.}
    \label{fig:simple_sdf}
\end{figure}

\begin{figure}[h!]
    \centering
    \begin{tikzpicture}[node distance={20mm}, thick, main/.style = {draw, circle}] 
    \node[main] (1) {$\x_0$}; 
    \node[main] (2) [above left=0.5cm and 2.3cm of 1] {$\x_1$}; 
    \node[main] (3) [above right=0.5cm and 2.3cm of 1] {$\x_2$};
    \node[] (9) [above=0.5cm of 2] {$\{w_{212}\}_{\{\omega_2\}}$};
    \node[] (8) [left=0cm of 9] {$\{w_{112}\}_{\{\omega_1\}}$};
    \node[] (7) [left=0cm of 8] {$\{w_{111}\}_{\{\omega_1\}}$};
    \node[] (10) [right=0cm of 9] {$\{w_{211}\}_{\{\omega_2\}}$};
    \node[] (12) [above=0.5cm of 3] {$\{w_{122}\}_{\{\omega_1\}}$};
    \node[] (13) [right=0cm of 12] {$\{w_{221}\}_{\{\omega_2\}}$};
    \node[] (11) [left=0cm of 12] {$\{w_{121}\}_{\{\omega_1\}}$};
    \node[] (14) [right=0cm of 13] {$\{w_{222}\}_{\{\omega_2\}}$};
    \draw[->] (1) -- (2); 
    \draw[->] (1) -- (3); 
    \draw[->] (2) -- (7); 
    \draw[->] (2) -- (8); 
    \draw[->] (2) -- (9); 
    \draw[->] (2) -- (10); 
    \draw[->] (3) -- (11); 
    \draw[->] (3) -- (12); 
    \draw[->] (3) -- (13); 
    \draw[->] (3) -- (14); 
    \end{tikzpicture} 
    \caption{The decision tree $(\Tr,\ge_\Tr)$ for the simple stochastic decision forest, with $w_{\ell km} = (\omega_\ell,k,m)$, for $(\ell,k,m)\in \{1,2\}^3$. (Random) moves are indicated by circles. Elements of $\Tr \setminus \X$, of the form $\{(\omega,\{w\})\}$ and seen as maps $\omega\mapsto \{w\}$, are denoted by $\{w\}_{\{\omega\}}$.}
    \label{fig:simple_sdf_Tr}
\end{figure}
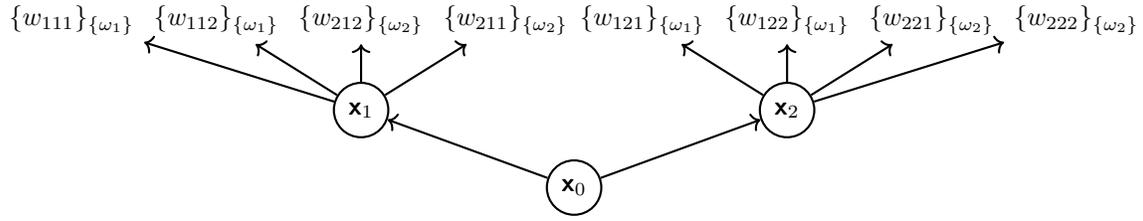
\newpage

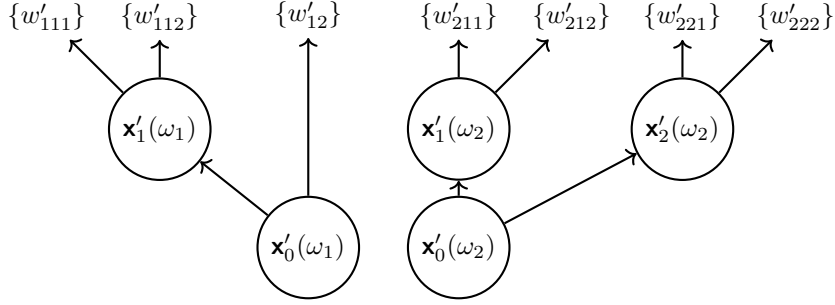
\begin{figure}[h!]
    \centering
    \begin{tikzpicture}[node distance={20mm}, thick, main/.style = {draw, circle}] 
    \node[main] (1) {$\x'_0(\omega_1)$}; 
    \node[main] (2) [above left=0.6cm and 1cm of 1] {$\x'_1(\omega_1)$}; 
    \node[] (3) [above=2.1cm of 1] {$\{w'_{12}\}$}; 
    \node[main] (5) [right of=1] {$\x'_0(\omega_2)$};
    \node[main] (4) [above=0.2cm of 5] {$\x'_1(\omega_2)$}; 
    \node[main] (6) [above right=0.6cm and 2cm of 5] {$\x'_2(\omega_2)$};
    \node[] (8) [above=0.5cm of 2] {$\{w'_{112}\}$};
    \node[] (7) [left=0.2cm of 8] {$\{w'_{111}\}$};
    \node[] (11) [above=0.5cm of 4] {$\{w'_{211}\}$};
    \node[] (12) [right=0.2cm of 11] {$\{w'_{212}\}$};
    \node[] (13) [above=0.5cm of 6] {$\{w'_{221}\}$};
    \node[] (14) [right=0.2cm of 13] {$\{w'_{222}\}$};
    \draw[->] (1) -- (2); 
    \draw[->] (1) -- (3); 
    \draw[->] (5) -- (4); 
    \draw[->] (5) -- (6);  
    \draw[->] (2) -- (7); 
    \draw[->] (2) -- (8); 
    \draw[->] (4) -- (11); 
    \draw[->] (4) -- (12); 
    \draw[->] (6) -- (13); 
    \draw[->] (6) -- (14); 
    \end{tikzpicture} 
    \caption{A variant of the simple stochastic decision forest in Figure~\ref{fig:simple_sdf} represented as a directed graph, with $w'_{\ell km} = (\omega_\ell,k,m)$, for all triples $(\ell,k,m)\in\{1,2\}^3$ with $(\omega_\ell,k,m)\in W'$, and $w'_{12} = (\omega_1,2)$. Moves are indicated by circles.}
    \label{fig:simple_sdf_variant}
\end{figure}

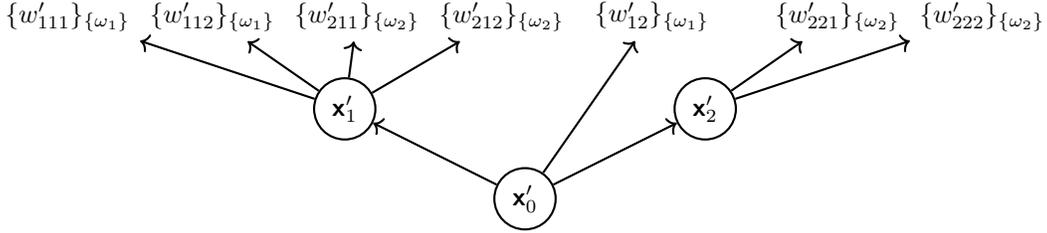
\begin{figure}[h!]
    \centering
    \begin{tikzpicture}[node distance={20mm}, thick, main/.style = {draw, circle}] 
    \node[main] (1) {$\x'_0$}; 
    \node[main] (2) [above left=0.6cm and 1.8cm of 1] {$\x'_1$}; 
    \node[] (3) [above right=1.8cm and 0.5cm of 1] {$\{w'_{12}\}_{\{\omega_1\}}$};  
    \node[main] (6) [above right=0.6cm and 1.8cm of 1] {$\x'_2$};
    \node[] (8) [above left=0.6cm and 0.5cm of 2] {$\{w'_{112}\}_{\{\omega_1\}}$};
    \node[] (7) [left=0cm of 8] {$\{w'_{111}\}_{\{\omega_1\}}$};
    \node[] (11) [right=0cm of 8] {$\{w'_{211}\}_{\{\omega_2\}}$};
    \node[] (12) [right=0cm of 11] {$\{w'_{212}\}_{\{\omega_2\}}$};
    \node[] (13) [above right=0.6cm and 0.5cm of 6] {$\{w'_{221}\}_{\{\omega_2\}}$};
    \node[] (14) [right=0cm of 13] {$\{w'_{222}\}_{\{\omega_2\}}$};
    \draw[->] (1) -- (2); 
    \draw[->] (1) -- (3); 
    \draw[->] (1) -- (6);  
    \draw[->] (2) -- (7); 
    \draw[->] (2) -- (8); 
    \draw[->] (2) -- (11); 
    \draw[->] (2) -- (12); 
    \draw[->] (6) -- (13); 
    \draw[->] (6) -- (14); 
    \end{tikzpicture} 
    \caption{The decision tree $(\Tr',\ge_{\Tr'})$ for the variant of the simple stochastic decision forest, with $w'_{\ell km} = (\omega_\ell,k,m)$, for all triples $(\ell,k,m)\in \{1,2\}^3$ with $(\omega_\ell,k,m) \in W'$, and $w'_{12} = (\omega_1,2)$. (Random) moves are indicated by circles. Elements of $\Tr' \setminus \X'$, of the form $\{(\omega,\{w'\})\}$ and seen as maps $\omega\mapsto \{w'\}$, are denoted by $\{w'\}_{\{\omega\}}$.}
    \label{fig:simple_sdf_variant_Tr}
\end{figure}

\begin{figure}[h!]
    \centering
    \begin{tikzpicture}[node distance={20mm}, thick, main/.style = {draw, circle}] 
    \node[main] (11) {$\x_1(\omega_1)$}; 
    \node[] (1D) [left=0.5 of 11] {$\{w_{1D}\}$};
    \node[main] (12) [above=0.5cm of 11] {$\x_2(\omega_1)$}; 
    \node[] (1H) [left=0.5 of 12] {$\{w_{1H}\}$};
    \node[] (1M) [above=0.5cm of 12] {$\{w_{1M}\}$}; 
    \node[main] (22) [right=0.5cm of 11] {$\x_2(\omega_2)$}; 
    \node[] (2D) [right=0.5 of 22] {$\{w_{1D}\}$};
    \node[main] (21) [above=0.5cm of 22] {$\x_1(\omega_2)$}; 
    \node[] (2H) [right=0.5 of 21] {$\{w_{1H}\}$};
    \node[] (2M) [above=0.5cm of 21] {$\{w_{1M}\}$}; 
    \draw[->] (11) -- (12); 
    \draw[->] (11) -- (1D); 
    \draw[->] (12) -- (1H); 
    \draw[->] (12) -- (1M);  
    \draw[->] (22) -- (21); 
    \draw[->] (22) -- (2D); 
    \draw[->] (21) -- (2H); 
    \draw[->] (21) -- (2M);  
    \end{tikzpicture} 
    \qquad
    \begin{tikzpicture}[node distance={20mm}, thick, main/.style = {draw, circle}] 
        \node[] (0) {};
        \node[main] (1) [above=0.5 of 0] {$\x_1$};
        \node[] (1D) [left=0.5cm of 1] {$\{w_{1D}\}_{\{\omega_1\}}$};
        \node[] (1H) [above left=0.5cm and 0.5cm of 1] {$\{w_{1H}\}_{\{\omega_1\}}$};
        \node[] (1M) [above=0.5cm of 1] {$\{w_{1M}\}_{\{\omega_1\}}$};  
        \node[main] (2) [right=1 of 1] {$\x_2$};
        \node[] (2D) [right=0.5cm of 2] {$\{w_{2D}\}_{\{\omega_2\}}$};
        \node[] (2H) [above right=0.5cm and 0.5cm of 2] {$\{w_{2H}\}_{\{\omega_2\}}$};
        \node[] (2M) [above=0.5cm of 2] {$\{w_{2M}\}_{\{\omega_2\}}$}; 
        \draw[->] (1) -- (1D);
        \draw[->] (1) -- (1H);
        \draw[->] (1) -- (1M);
        \draw[->] (1) -- (2M);
        \draw[->] (1) -- (2H);
        \draw[->] (2) -- (2M);
        \draw[->] (2) -- (1H);
        \draw[->] (2) -- (1M);
        \draw[->] (2) -- (2D);
        \draw[->] (2) -- (2H);
    \end{tikzpicture}
    \caption{The absent-minded driver \textsc{sdf}, following Gilboa: $(F,\supseteq)$ represented as a directed graph, in case $\rho$ is injective, with $\omega_\ell = \rho^{-1}(\ell)$, $w_{\ell S} = (\omega_\ell,S)$, for all $\ell\in \{1,2\}$ and symbols $S$ such that $(\omega,S)\in W$ (left), $(\Tr,\ge_\Tr)$ where elements of $\Tr \setminus \X$, of the form $\{(\omega,\{w\})\}$ and seen as maps $\omega\mapsto \{w\}$, are denoted by $\{w\}_{\{\omega\}}$ (right). Non-minimal elements are indicated by circles, respectively.}
    \label{fig:absent_minded_driver_Gilboa_sdf}
\end{figure}
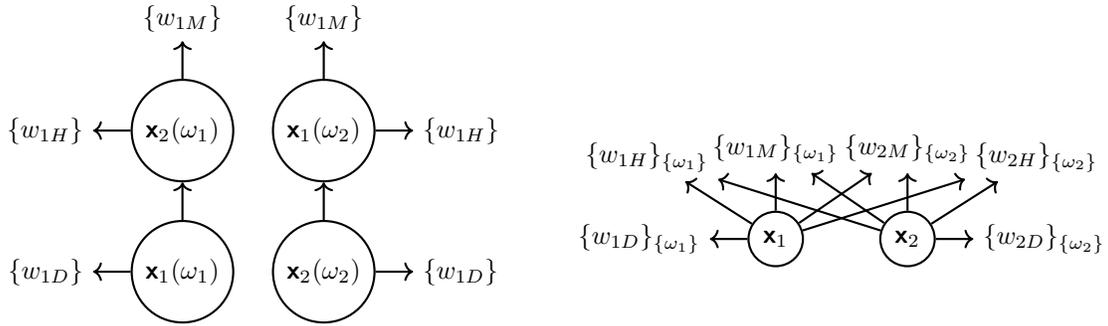

\newpage
\section{Proofs}\label{sec:proofs}

\subsection{Section~\ref{sec:Stochastic extensive forms}}\label{subsec:appendix.proofs.1}

\begin{lemma}\label{lemma:Xi_random_moves}
    Let $\F$ be a stochastic pseudo-extensive form on some exogenous scenario space $(\Omega,\ms E)$, and $i\in I$. Then, 
    \[ X^i = \{ \x(\omega) \mid \x\in \X^i,~\omega\in D_\x\}. \]
\end{lemma}

\begin{proof}
    Let $x\in X^i$. Then, by Axiom~\ref{def:sdf}.\ref{def:sdf.X.cov}, there is $\x\in\X$ and $\omega\in D_\x$ with $x=\x(\omega)$. Hence, $A^i(x)\subseteq A^i(\x)$, by definition of $A^i(.)$. As $A^i(x)\neq\emptyset$ by assumption, $A^i(\x)\neq\emptyset$, whence $\x\in\X^i$.

    Conversely, let $\x\in\X^i$ and $\omega\in D_\x$. Then, $A^i(\x(\omega)) = A^i(\x) \neq\emptyset$, by $\X^i$-completeness of all elements of $C^i$. Hence, $\x(\omega)\in X^i$.
\end{proof}

Letting $\X^i \bullet \Omega = \{(\x,\omega) \in \X\bullet\Omega \mid \x\in\X^i\}$, the lemma states that $X^i = \mc P\ev(\X^i \bullet \Omega)$, just as $X = \mc P\ev(\X \bullet \Omega)$ and $F = \mc P\ev(\Tr \bullet \Omega)$.

\begin{proof}[Proof of Proposition~\ref{prop:information_sets}]
    Let $\F = (F,\pi,\X,I,\ms F,\ms F,C)$ be a tuple as in Definition~\ref{def:SEF} satisfying Axioms~\ref{def:SEF}.$k$, $k=1,\dots,5$, but not necessarily Axiom~\ref{def:SEF}.\ref{def:SEF.choice_completeness}.\smallskip

    (Ad \ref{prop:information_sets.P(c)_partition}):~ First, any $c\in C^i$ is non-empty. Hence, there is $\omega\in\Omega$ with $c\cap W_\omega\neq\emptyset$. By non-redundancy, $P(c)\cap T_\omega\neq\emptyset$, whence $P(c)\neq\emptyset$. 
    
    Next, let $x\in X^i$. Hence, $A^i(x) \neq \emptyset$ so that there is $c\in C^i$ such that $x\in P(c)$. 

    Then, let $c,c'\in C^i$ such that $P(c) \cap P(c') \neq \emptyset$. Then, by Axiom~\ref{def:SEF}.\ref{def:SEF.P(c)}, $P(c) = P(c')$.\smallskip

    (Ad \ref{prop:information_sets.A(x)_partition}):~ First, note that $A^i(x)\neq\emptyset$ for all $x\in X^i$, by definition.
    
    As $A^i(\x(\omega)) = A^i(\x)$ for all $\x\in\X^i$ and $\omega\in D_\x$, Lemma~\ref{lemma:Xi_random_moves} easily implies that 
    \[ \{A^i(x) \mid x\in X^i\} = \{A^i(\x) \mid \x\in\X^i\}. \]

    Let $c\in C^i$. As $c$ is non-empty, there is $\omega\in\Omega$ with $c\cap W_\omega \neq \emptyset$. By non-redundancy, $P(c) \cap T_\omega \neq \emptyset$. Hence, there is $x\in T_\omega \subseteq F$ with $x\in P(c)$. It follows easily from the definition of $P(c)$ that $x\in X$, whence $c\in A^i(x)$, $i\in J(x)$, and $x\in X^i$.

    Let $x,x'\in X^i$ such that $A^i(x) \cap A^i(x') \neq \emptyset$. Let $c\in A^i(x) \cap A^i(x')$. Hence $x,x'\in P(c)$. Let $c'\in A^i(x)$. Then $x\in P(c) \cap P(c')$. Hence, using Axiom~\ref{def:SEF}.\ref{def:SEF.P(c)} we obtain $x'\in P(c) = P(c')$. Thus, $c'\in A^i(x')$. We conclude that $A^i(x) \subseteq A^i(x')$. The same argument can be repeated with the roles of $x$ and $x'$ reversed, whence in total $A^i(x) = A^i(x')$.\smallskip

    (Ad \ref{prop:information_sets.P(c)=P(c')}):~ Let $c,c'\in C^i$. By Axiom~\ref{def:SEF}.\ref{def:SEF.P(c)}, $P(c) = P(c')$ is true iff there is $x\in P(c) \cap P(c')$. Such an $x$ must be a move by definition of immediate predecessors. Hence, this statement is equivalent to saying that $c,c'\in A^i(x)$.\smallskip

    (Ad \ref{prop:information_sets.A(x)=A(x')}):~ Let $x,x'\in X^i$. Using Part~\ref{prop:information_sets.A(x)_partition} just proved before, $A^i(x) = A^i(x')$ is equivalent to the existence of $c \in A^i(x) \cap A^i(x')$, or put equivalently, $x,x'\in P(c)$ for some $c\in C^i$.\smallskip

    (Ad \ref{prop:information_sets.exists_mfP}):~ It suffices to show that there is a unique equivalence relation $\sim$ on $\X^i$ satisfying $\x\sim\x'$ iff $A^i(\x) = A^i(\x')$, for all $\x,\x'\in \X^i$. Uniqueness is trivial. Concerning existence, there clearly is a binary relation $\sim$ with the preceding property. Moreover, reflexivity, symmetry, and transitivity are trivial. Hence, $\sim$ is an equivalence relation.\smallskip

    (Ad \ref{prop:information_sets.Bij_mfP_P(c)}):~ Let $\Phi$ be the map with domain $\mf P^i$ given by
    \[ \forall \mf p \in \mf P^i\colon\qquad \Phi(\mf p) = \bigcup_{\x\in\mf p} \im \x.\]
    
    Concerning the claim about both the codomain and the image of $\Phi$, we prove the following helpful statement: 
    \[ (\ast) \qquad \forall \mf p\in\mf P^i\forall \x_0\in\mf p\forall c_0\in A^i(\x_0)\colon ~\Phi(\mf p) = P(c_0) . \]

    For the proof, let $\mf p \in \mf P^i$, $\x_0\in\mf p$, and $c_0\in A^i(\x_0)$. 
    By definition of $\mf P^i$, $A^i(\x) = A^i(\x') \neq \emptyset$ for all $\x,\x'\in\mf p$. By Parts~\ref{prop:information_sets.P(c)=P(c')} and~\ref{prop:information_sets.A(x)_partition}, proven above, we obtain that 
    \[ (\dagger) \qquad \forall \x,\x'\in \mf p \forall c\in A^i(\x) \forall c'\in A^i(\x')\colon ~ P(c) = P(c') . \]

    First, let $x\in \Phi(\mf p)$. Take $\x\in\mf p$ and $\omega\in D_\x$ such that $x = \x(\omega)$. As $\x\in\mf p$, there is $c\in A^i(\x)$. By $(\dagger)$, $P(c) = P(c_0)$, whence $x\in P(c) = P(c_0)$. This shows that $\Phi(\mf p)\subseteq P(c_0)$.

    Second, let $x\in P(c_0)$. By Axiom~\ref{def:sdf}.\ref{def:sdf.X.cov}, there is $\x\in\X$ and $\omega\in D_\x$ such that $x = \x(\omega)$. Thus $c_0\in A^i(\x)$ and $\x\in\X^i$. By Part~\ref{prop:information_sets.A(x)=A(x')}, proven before, $A^i(\x) = A^i(\x_0)$ must hold true. As $\x_0\in\mf p$, we infer that $\x\in\mf p$ as well. This shows that $P(c_0) \subseteq \Phi(\mf p)$. We conclude that $\Phi(\mf p)=P(c_0)$, and the proof of $(\ast)$ is complete.

    Now, regarding the codomain of $\Phi$, let $\mf p\in\mf P^i$. As $\mf P^i$ is a partition, $\mf p$ is non-empty. Hence, we can choose $\x_0\in\mf p$. As $\x_0\in\X^i$, there is $c_0\in A^i(\x_0)$. By $(\ast)$, $\Phi(\mf p) = P(c_0)$. Hence, the codomain is of the claimed form.

    Next, we determine the image of $\Phi$. Let $c_0\in C^i$. By Part~\ref{prop:information_sets.P(c)_partition}, proven before, there is $x_0\in X^i$ such that $x_0\in P(c_0)$. By Lemma~\ref{lemma:Xi_random_moves}, there are $\x_0\in\X^i$ and $\omega\in D_\x$ with $x_0 = \x_0(\omega)$. In particular, $c_0\in A^i(\x_0)$. By construction of $\mf P^i$, there is $\mf p\in\mf P^i$ such that $\x_0\in\mf p$. Then, statement $(\ast)$ implies that $\Phi(\mf p) = P(c_0)$. Hence, the image of $\Phi$ is given by the set of all $P(c)$, $c\in C^i$.

    It remains to prove injectivity. Let $\mf p,\mf p'\in \mf P^i$ such that $\Phi(\mf p) = \Phi(\mf p')$. As shown just before, there is $c\in C^i$ such that $\Phi(\mf p) = P(c) = \Phi(\mf p')$. By Part~\ref{prop:information_sets.P(c)_partition}, there is $x\in P(c)$. Hence, there are representatives $\x\in\mf p$ and $\x'\in\mf p'$ of both endogenous information sets and $\omega \in D_\x \cap D_{\x'}$ such that $\x(\omega) = x = \x'(\omega)$. By Parts~\ref{prop:information_sets.A(x)_partition} and~\ref{prop:information_sets.A(x)=A(x')}, we infer that $A^i(\x) = A^i(\x')$, whence $\mf p = \mf p'$. We conclude that $\Phi$ is injective.\smallskip

    (Ad \ref{prop:information_sets.msC_msF_const_on_mfP}):~ Let $\mf p\in\mf P^i$ and $\x,\x'\in \mf p$. Then, by definition of $\mf P^i$, $A^i(\x) = A^i(\x')\neq \emptyset$. Hence, by Axiom~\ref{def:SEF}.\ref{def:SEF.endo_exo_compatible}, $\ms F_\x^i = \ms F_{\x'}^i$ and $\ms C_\x^i=\ms C_{\x'}^i$. As a consequence,
    $ D_\x = \bigcup \ms F_\x^i = \bigcup \ms F_{\x'}^i = D_{\x'}$.
\end{proof}

\begin{lemma}\label{lemma:perfect_endo_information_implies_perfect_endo_recall}
    Let $\F$ be a stochastic pseudo-extensive form on some exogenous scenario space $(\Omega,\ms E)$ and $i\in I$ an agent. If $i$ has perfect endogenous (exogenous) information, then $i$ admits perfect endogenous (exogenous, respectively) recall.
\end{lemma}

\begin{proof}[Proof of Lemma~\ref{lemma:perfect_endo_information_implies_perfect_endo_recall}]
    Let $\F$ be a $\psi$-\textsc{sef}, and $i\in I$ an agent.\smallskip
    
    (Ad endogenous case):~ Suppose that $i$ has perfect endogenous information. Let $c,c'\in C^i$ and $\omega\in\Omega$ such that $c\cap c'\cap W_\omega\neq \emptyset$. By non-redundancy and \cite[Lemma~4.1]{Rapsch2024DecisionA}, we get $P(c\cap W_\omega) = P(c) \cap T_\omega \neq \emptyset$, and similarly $P(c'\cap W_\omega) = P(c') \cap T_\omega \neq \emptyset$. Proposition~\ref{prop:information_sets}, Part~\ref{prop:information_sets.Bij_mfP_P(c)}, and the fact that $i$ has perfect information, we infer the existence of $x,x'\in X^i$ such that
    \[ P(c\cap W_\omega) =  P(c) \cap T_\omega = \{x\}, \qquad P(c'\cap W_\omega) = P(c') \cap T_\omega = \{x'\}. \]
    Directly from the definition of immediate predecessors, we infer that $x\supseteq c\cap W_\omega$ and $x'\supseteq c'\cap W_\omega$. Hence, $x\cap x' \neq \emptyset$. By the representation by decision paths inherent in the definition of decision forests (compare Definition~\ref{def:decision_forest}, and with more details \cite[Definition 1.3]{Rapsch2024DecisionA}), we have $x\supseteq x'$ or $x'\supseteq x$. 

    Without loss of generality, we assume that $x\supseteq x'$. If $x=x'$, then $P(c) \cap P(c') \neq \emptyset$, whence by Axiom~\ref{def:SEF}.\ref{def:SEF.P(c)} $P(c) = P(c')$ and $c \cap W_\omega = c' \cap W_\omega$. 
    
    If $x\supsetneq x'$, then Axiom~\ref{def:SEF}.\ref{def:SEF.enough_choices} implies the existence of $\tilde c\in A^i(x)$ such that $\tilde c \supseteq x'$. As $c\in A^i(x)$, Proposition~\ref{prop:information_sets}, Part~\ref{prop:information_sets.P(c)=P(c')} implies $P(c) = P(\tilde c)$. Moreover, $\tilde c \supseteq x' \supseteq c'\cap W_\omega$, whence
    \[ c\cap \tilde c \cap W_\omega \supseteq c \cap c' \cap W_\omega \neq \emptyset. \]
    Axiom~\ref{def:SEF}.\ref{def:SEF.P(c)} yields $c\cap W_\omega = \tilde c \cap W_\omega$. We conclude that $c\cap W_\omega \supseteq c'\cap W_\omega$.\smallskip

    (Ad exogenous case):~ Suppose that $i$ has perfect exogenous information. Then, for any $\x,\x'\in \X^i$ with $\x\ge_\X \x'$ and any $E\in{\ms F}_\x^i$, we have $E\in\ms E$ and, hence, $E\cap D_{\x'} \in \ms E|_{D_{\x'}} = {\ms F}^i_{\x'}$.    
\end{proof}

\begin{proof}
    [Proof of Lemma~\ref{lemma:Heraclitus_property}]
    The proof is completely analogous to the proofs of \cite[Proposition 13]{AlosFerrer2005} and \cite[Proposition 4.1]{AlosFerrer2016}. Nevertheless, we give a proof here, both because of the different formal setting and for the reader's convenience.\smallskip

    (Ad \ref{lemma:Heraclitus_property.X}):~ Let $x,x'\in X$ such that $A^i(x) \cap A^i(x') \neq \emptyset$ and $x\supseteq x'$. In particular, $i\in J(x) \cap J(x')$.
    If $x\neq x'$, then, by Axiom~\ref{def:SEF.enough_choices} there would be $c'\in A^i(x)$ such that $c'\supseteq x'$.    
    By Proposition~\ref{prop:information_sets}, Part~\ref{prop:information_sets.A(x)_partition}, $A^i(x) = A^i(x')$, and thus $c'\in A^i(x')$ as well. In other words, $x'\in P(c')$. Hence, by definition of the immediate predecessor operator, there would be $y'\in\downarrow c'$ such that
     \[ \uparrow x' = \uparrow y' \setminus \downarrow c'. \]
    As $x'\in\downarrow c'$, this is a contradiction. Hence, the assumption was false and we conclude that $x=x'$.\smallskip

    (Ad \ref{lemma:Heraclitus_property.rmX}):~ Let $\x,\x'\in \X$ such that $A^i(\x)\cap A^i(\x') \neq \emptyset$ and $\x\ge_\X \x'$. In particular, $\x,\x'\in \X^i$. 
    Then, there is $\omega\in D_{\x'}$, and we have $\x(\omega) \supseteq \x'(\omega)$ and $A^i(\x(\omega)) \cap A^i(\x'(\omega))\neq \emptyset$. By the statement of the first part, just proven before, we get $\x(\omega) = \x'(\omega)$. Since $\x,\x'\in\X^i$ and evaluation $\X^i\bullet \Omega\to X$ is injective, we obtain $\x = \x'$.
\end{proof}

\begin{proof}
    [Proof of Lemma~\ref{lemma:completeness}]
    Let $\F = (F,\pi,\X,I,\ms F,\ms C,C)$ be a tuple satisfying the conditions defining a stochastic pseudo-extensive form on some exogenous scenario space $(\Omega,\ms E)$ possibly except Axiom~\ref{def:SEF.choice_completeness}, according to Definition~\ref{def:SEF}. For any $i\in I$, let $\hat C^i$ be as defined in the lemma's statement.\smallskip
    
    (Helpful statement \ref{lemma:completeness.def:hatC.iii}):~ To start, let us prove that Properties~\ref{lemma:completeness}.\ref{lemma:completeness.def:hatC.i} and~\ref{lemma:completeness}.\ref{lemma:completeness.def:hatC.ii} can be strengthened. Namely:
    \begin{enumerate}[label=(iii)]
        \item\label{lemma:completeness.def:hatC.iii} For all $i\in I$, all $\hat c\in \hat C^i$ and $\omega\in\Omega$ with $\hat c \cap W_\omega \neq \emptyset$, there is $c\in C^i$ such that $\hat c\cap W_\omega = c \cap W_\omega$ and $P(\hat c) = P(c)$.
    \end{enumerate}
    For the proof, let $i\in I$, $\hat c\in \hat C^i$ and $\omega\in\Omega$ be such that $\hat c \cap W_\omega \neq \emptyset$. By definition of $\hat C^i$, Property~\ref{lemma:completeness.def:hatC.i}, there is $c \in C^i$ such that $\hat c \cap W_\omega = c \cap W_\omega$. By \cite[Lemma~4.1]{Rapsch2024DecisionA}, we have
    \[ P(c) \cap T_\omega = P(c\cap W_\omega) = P(\hat c \cap W_\omega) = P(\hat c) \cap T_\omega, \]
    and by non-redundancy of $\hat c$ these sets are non-empty. Let $\mf P^i$ the partition of $\X^i$ according to Remark~\ref{rmk:prop_information_sets} and Proposition~\ref{prop:information_sets}, Part~\ref{prop:information_sets.exists_mfP}. Then, there is $\mf p\in\mf P^i$ such that
    \[ P(c) \cap T_\omega = \{\x(\omega) \mid \x\in\mf p\colon~ \omega \in D_\x\}. \]
    As $\hat c$ is an $\X^i$-complete choice and $\omega\in D_\x = D_\mf p$ for all $\x\in\mf p$ by Remark~\ref{rmk:prop_information_sets} and Proposition~\ref{prop:information_sets}, Part~\ref{prop:information_sets.msC_msF_const_on_mfP}, we obtain that
    \[ P(\hat c) \supseteq \{\x(\omega') \mid \x\in\mf p, ~\omega'\in D_\x\} = P(c). \]
    By definition of $\hat C^i$, Property~\ref{lemma:completeness.def:hatC.ii}, there is $c_1\in C^i$ such that $P(\hat c) \subseteq P(c_1)$. Hence, $\emptyset \neq P(c) \subseteq P(c_1)$ and thus, by Axiom~\ref{def:SEF}.\ref{def:SEF.P(c)}, $P(c) = P(c_1)$, whence $P(c) = P(\hat c)$ which proves \ref{lemma:completeness.def:hatC.iii}.\smallskip
    
    (Ad Properties~\ref{lemma:completeness}.\ref{lemma:completeness.property_1} and~\ref{lemma:completeness}.\ref{lemma:completeness.property_2}):~
    Both properties follow from the trivial inclusion $C^i \subseteq \hat C^i$ for all $i\in I$ and Property~\ref{lemma:completeness.def:hatC.iii}. \smallskip

    (Ad basic properties in Definition~\ref{def:SEF})~ For each $i\in I$, $\hat C^i$ is a set of choices, by construction. 
    
    For the remainder of the proof, the operators and sets associated to $\hat C$ and $(F,\pi,\X)$ will be denoted with a hat on top. That is, for $\hat C$ and $(F,\pi,\X)$, denote the set of choices in $\hat C^i$ available at a move $x\in X$ or random move $\x\in\X$ by $\hat A^i(x)$ or $\hat A^i(\x)$, respectively. Furthermore, for $x\in X$ and $\x\in\X$, let
    \[ \hat J(x) = \{i\in I \mid \hat A^i(x) \neq\emptyset\}, \qquad \hat J(\x) = \{i\in I \mid \hat A^i(\x) \neq\emptyset\}, \]
    and, for $i\in I$, let
    \[ \hat X^i = \{x\in X \mid i\in \hat J(x)\}, \qquad \hat\X^i = \{\x\in\X \mid i\in \hat J(\x)\}, \qquad \hat\X^i\bullet \Omega = \{(\x,\omega)\in\hat\X^i\times\Omega \mid \omega\in D_\x\}. \]
    Then, clearly, all $x\in X$, $\x\in\X$, and $i\in I$ satisfy $A^i(x) \subseteq \hat A^i(x)$ and $A^i(\x) \subseteq \hat A^i(\x)$, $J(x) = \hat J(x)$ and $J(\x) = \hat J(\x)$. In particular, $X^i = \hat X^i$ and $\X^i = \hat\X^i$.

    As a consequence, $\ms F$ is a family of exogenous information structures on $\hat\X^i$, $i\in I$, and $\ms C$ is a family of reference choice structures on $\hat\X^i$, $i\in I$. By hypothesis, for any $i\in I$, any element of $\hat C^i$ is $\ms F^i$-$\ms C^i$-adapted, and the evaluation $\hat\X^i \bullet \Omega = \X^i \bullet \Omega \to X$ is injective.\smallskip   
    
    (Ad Axiom~\ref{def:SEF}.\ref{def:SEF.P(c)}):~ Let $i\in I$ and $\hat c,\hat c'\in \hat C^i$ such that $P(\hat c) \cap P(\hat c') \neq \emptyset$. Then, by Property~\ref{lemma:completeness.def:hatC.iii}, there are $c,c'\in C^i$ with $P(c) = P(\hat c)$ and $P(c') = P(\hat c')$. Hence, by Axiom~\ref{def:SEF}.\ref{def:SEF.P(c)} applied to $\F$, 
    \[ P(\hat c) = P(c) = P(c') = P(\hat c'). \]
    Furthermore, let $\omega\in\Omega$ such that $\hat c \cap \hat c' \cap W_\omega \neq \emptyset$. According to Property~\ref{lemma:completeness.def:hatC.iii}, $c$ and $c'$ can be chosen such that $\hat c \cap W_\omega = c \cap W_\omega$ and $\hat c' \cap W_\omega = c' \cap W_\omega$. In particular we get $c\cap c' \cap W_\omega \neq \emptyset$. As $P(c) = P(c')$, Axiom~\ref{def:SEF}.\ref{def:SEF.P(c)} applied to $\F$ yields 
    \[ \hat c \cap W_\omega = c\cap W_\omega = c' \cap W_\omega = \hat c'\cap W_\omega. \]

    (Ad Axiom~\ref{def:SEF}.\ref{def:SEF.outcomes_faithful}):~ Let $x\in X$ and $(\hat c^i)_{i\in J(x)} \in \bigtimes_{i\in J(x)} \hat C^i$. Let $\omega = \pi(x)$. Then, by definition of $\hat C$, there is $(c^i)_{i\in J(x)} \in \bigtimes_{i\in J(x)} C^i$ such that for all $i\in J(x)$ we have $\hat c^i \cap W_\omega = c^i \cap W_\omega$. As $x \in T_\omega$, we have $W_\omega \supseteq x$, whence
    \[ x \cap \bigcap_{i\in J(x)} \hat c^i = x \cap \bigcap_{i\in J(x)} (\hat c^i \cap W_\omega) = x \cap \bigcap_{i\in J(x)} (c^i \cap W_\omega) = x \cap \bigcap_{i\in J(x)} c^i \neq \emptyset, \]
    because $\F$ satisfies \ref{def:SEF}.\ref{def:SEF.outcomes_faithful}.

    (Ad Axiom~\ref{def:SEF}.\ref{def:SEF.weak_separation}):~ Let $y,y'\in F$ with $\pi(y) = \pi(y')$ and $y\cap y' = \emptyset$. Let $\omega = \pi(y)$. Then, by Axiom~\ref{def:SEF}.\ref{def:SEF.weak_separation} applied to $\F$, there are $i\in I$ and $c,c'\in C^i$ with $c \supseteq y$, $c' \supseteq y'$ and $c\cap c' \cap W_\omega = \emptyset$. By Property~\ref{lemma:completeness.def:hatC.iii} of $\hat C^i$, there are $\hat c, \hat c'\in \hat C^i$ such that $\hat c\cap W_\omega = c\cap W_\omega$ and $\hat c' \cap W_\omega = c'\cap W_\omega$ and $P(\hat c) = P(c)$ as well as $P(\hat c') = P(c')$. Hence,
    \[ \hat c \supseteq \hat c \cap W_\omega = c \cap W_\omega \supseteq y \cap W_\omega = y, \]
    and similarly, $\hat c' \supseteq y'$. Moreover, 
    \[ \hat c \cap \hat c' \cap W_\omega = c \cap c' \cap W_\omega = \emptyset. \]

    (Ad Axiom~\ref{def:SEF}.\ref{def:SEF.separation}):~ For the proof of this axiom's validity, suppose that $\F$ satisfies Axiom~\ref{def:SEF}.\ref{def:SEF.separation}. We ought to show that $\hat\F$ does as well. 
    
    Let $y,y'\in F$ with $\pi(y) = \pi(y')$ and $y\cap y' = \emptyset$. Let $\omega = \pi(y)$. Then, by Axiom~\ref{def:SEF}.\ref{def:SEF.separation} applied to $\F$, there are $x\in X$, $i\in I$ and $c,c'\in C^i$ with $x\cap c \supseteq y$, $x\cap c' \supseteq y'$, $c\cap c' \cap W_\omega = \emptyset$ and $x\in P(c) \cap P(c) \cap T_\omega$. By Property~\ref{lemma:completeness.def:hatC.iii} of $\hat C^i$, there are $\hat c, \hat c'\in \hat C^i$ such that $\hat c\cap W_\omega = c\cap W_\omega$ and $\hat c' \cap W_\omega = c'\cap W_\omega$, and $P(\hat c) = P(c)$ as well as $P(\hat c') = P(c')$. Hence,
    \[ x\cap \hat c = x \cap \hat c \cap W_\omega = x \cap c \cap W_\omega = x \cap c \supseteq y, \]
    and similarly, $x \cap \hat c' \supseteq y'$. Moreover, 
    \[ \hat c \cap \hat c' \cap W_\omega = c \cap c' \cap W_\omega = \emptyset, \]
    and $x \in P(c) \cap P(c') \cap T_\omega = P(\hat c) \cap P(\hat c') \cap T_\omega$.\smallskip

    (Ad Axiom~\ref{def:SEF}.\ref{def:SEF.enough_choices}):~ Let $x\in X$, $i\in J(x)$ and $y\in \downarrow x \setminus \{x\}$. By Axiom~\ref{def:SEF}.\ref{def:SEF.enough_choices} applied to $\F$, we obtain $c\in A^i(x)$ with $c\supseteq y$. Let $\omega = \pi(x)$. Property~\ref{lemma:completeness.def:hatC.iii} of $\hat C$ ensures the existence of $\hat c\in\hat C^i$ such that $\hat c \cap W_\omega = c\cap W_\omega$ and $P(\hat c) = P(c)$. Hence, $x\in P(\hat c)$ alias $\hat c\in \hat A^i(x)$. Moreover, 
    \[ \hat c \cap W = c \cap W_\omega \supseteq y \cap W_\omega = y. \]

    (Ad Axiom~\ref{def:SEF}.\ref{def:SEF.endo_exo_compatible}):~ Let $i\in I$ and $\x,\x'\in\X$ such that $\hat A^i(\x) \cap \hat A^i(\x')\neq\emptyset$. Hence, there is $\hat c\in \hat C^i$ such that $\x(\omega),\x'(\omega') \in P(\hat c)$ for some (and any) $\omega\in D_\x$ and $\omega'\in D_{\x'}$. Following Property~\ref{lemma:completeness.def:hatC.iii} of $\hat C$, there is $c\in C^i$ such that $P(\hat c) = P(c)$. Hence, $c\in A^i(\x) \cap A^i(\x')$, whence $\ms F^i_\x = \ms F^i_{\x'}$ and $\ms C^i_\x = \ms C^i_{\x'}$. \smallskip

    (Ad Axiom~\ref{def:SEF}.\ref{def:SEF.choice_completeness}):~ Let $i\in I$ and $\hat c'$ an $\ms F^i$-$\ms C^i$-adapted choice satisfying Properties~\ref{def:SEF}.\ref{def:SEF.choice_completeness.i} and~\ref{def:SEF}.\ref{def:SEF.choice_completeness.ii}, that is:
    \begin{enumerate}[label=(\roman*)]
        \item any $\omega\in\Omega$ with $\hat c'\cap W_\omega \neq \emptyset$ admits $\hat c\in \hat C^i$ with $\hat c'\cap W_\omega = \hat c \cap W_\omega$,
        \item and there is $\hat c\in \hat C^i$ with $P(\hat c') = P(\hat c)$,
    \end{enumerate}
    Let $\omega\in \Omega$ be such that $\hat c'\cap W_\omega$. Take $\hat c\in\hat C^i$ with $\hat c'\cap W_\omega = \hat c \cap W_\omega$. By definition of $\hat C$, there is $c\in C^i$ with $\hat c\cap W_\omega = c\cap W_\omega$, whence $\hat c' \cap W_\omega = c\cap W_\omega$. 

    Further, take $\hat c\in\hat C^i$ such that $P(\hat c') = P(\hat c)$. By Property~\ref{lemma:completeness.def:hatC.iii} of $\hat C$ there is $c\in C^i$ such that $P(\hat c) = P(c)$, implying $P(\hat c') \subseteq P(c)$. 

    We conclude that $\hat c'\in\hat C^i$ which shows that the axiom is satisfied by $\F'$.
\end{proof}

\begin{proof}
    [Proof of Proposition~\ref{prop:strategies}]
    Let $i\in I$. Denote the surjections $X^i \surj \X^i$ and $\X^i \to \mf P^i$ by $p_{X,\X;i}$ and $p_{\X,\mf P;i}$ respectively. By Proposition~\ref{prop:information_sets}, we have for all $x\in X^i$ and $\x\in\X^i$:
    \[ A^i(x) = A^i(p_{X,\X;i}(x)), \qquad A^i(\x) = A^i(p_{\X,\mf P;i}(\x)). \]
    Moreover, for all $x,x'\in X^i$ we have
    \[ (p_{\X,\mf P;i}\circ p_{X,\X;i})(x) = (p_{\X,\mf P;i}\circ p_{X,\X;i})(x') \qquad \Longleftrightarrow \qquad A^i(x) = A^i(x'), \]
    and similarly, for all $\x,\x'\in \X^i$ we have
    \[ p_{\X,\mf P;i}(\x) = p_{\X,\mf P;i}(\x') \qquad \Longleftrightarrow \qquad A^i(\x) = A^i(\x'). \]
    The claim follows from this using the universal property of the quotient in the category of sets.
\end{proof}

\begin{proof}
    [Proof of Lemma~\ref{lemma:simple_sef1}]
    Let $(F,\pi,\X,I,\ms F,\ms C,C)$ be a tuple satisfying the hypothesis of the lemma. Let $i$ be the unique element of $I$. Recall the table defining $C$.
    
    According to \cite[Lemma~2.6]{Rapsch2024DecisionA}, $(F,\pi,\X)$ defines a stochastic decision forest on $(\Omega,\ms E)$. By \cite[Lemma~3.2]{Rapsch2024DecisionA}, $\ms F^i$ defines an exogenous information structure for it. By \cite[Lemma~4.3]{Rapsch2024DecisionA}, $\ms C^i$ defines a reference choice structure for it, and by \cite[Lemma~4.4]{Rapsch2024DecisionA}, the elements of $C^i$ are $\ms F^i$-$\ms C^i$-adapted choices.

    (Ad Axiom~\ref{def:SEF}.\ref{def:SEF.P(c)}):~ The immediate predecessor sets of all choices in $C^i$ have been explicitly calculated in \cite[Lemma Appendix A.3]{Rapsch2024DecisionA}. It follows from the results of that lemma that for all $c,c'\in C^i$, $P(c)$ and $P(c')$ can only non-trivially intersect if $c$ and $c'$ are denoted in the same column of the table. Given this, closer inspection of the table reveals that if $P(c)$ and $P(c')$ non-trivially intersect, then $P(c) = P(c')$, and $c\cap W_\omega$ and $c'\cap W_\omega$ are either disjoint or equal, for all $\omega\in\Omega$.\smallskip

    (Ad Axiom~\ref{def:SEF}.\ref{def:SEF.outcomes_faithful}):~ There is only one agent $i\in I$ here, hence this axiom is trivially satisfied. Indeed, if $x\in X$ and $c\in C^i$ such that $i\in J(x)$, then $x\in P(c)$. Hence, there is $y\in \downarrow c$ such that $\uparrow x = \uparrow y \setminus \downarrow c$. Hence, $x\cap c \supseteq y \cap c = y \neq \emptyset$.\smallskip

    (Ad Axiom~\ref{def:SEF}.\ref{def:SEF.separation}):~ Let $y,y'\in T_\omega$ for some $\omega\in\Omega$ with $y\cap y' = \emptyset$. Then, $y$ and $y'$ are moves at time $1$ or terminal nodes. As $T_\omega$ is a finite tree, there is a $\supseteq$-minimal $x\in F$ with $x\supseteq y\cup y'$. In particular, $x$ must be a move with $\pi(x) = \omega$. Let $i\in I$. For any possible values of $x$ and $y$, there are $c,c'\in A^i(x)$ such that $x\cap c \supseteq y$ and $x\cap c' \supseteq y'$, as evident from the table and \cite[Lemma Appendix A.3]{Rapsch2024DecisionA}.\smallskip

    (Ad Axiom~\ref{def:SEF}.\ref{def:SEF.enough_choices}):~ Let $x\in X$, $i\in J(x)$, and $y\in\downarrow x \setminus \{x\}$. As evidenced by the table and \cite[Lemma Appendix A.3]{Rapsch2024DecisionA}, there is $c\in A^i(x)$ such that $c \supseteq y$.\smallskip

    (Ad Axiom~\ref{def:SEF}.\ref{def:SEF.endo_exo_compatible}):~ Upon consulting the table, we infer using \cite[Lemma Appendix A.3]{Rapsch2024DecisionA} that $\x_1$ and $\x_2$ are the only two random moves sharing an available choice $c$, and this only in the cases given by th second, fourth, and eighth line of the table. In these cases, we have indeed $\ms F^i_{\x_1} = \ms F^i_{\x_2}$ (\textsc{eis} = 1, 2(a), or 3). We also have $\ms C^i_{\x_1} = \ms C^i_{\x_2}$.\smallskip

    (Ad Axiom~\ref{def:SEF}.\ref{def:SEF.choice_completeness}):~ Let $c'$ be an $\ms F^i$-$\ms C^i$-adapted choice satisfying \ref{def:SEF}.\ref{def:SEF.choice_completeness.i} and~\ref{def:SEF}.\ref{def:SEF.choice_completeness.ii}. From the latter two properties and \cite[Lemma Appendix A.3]{Rapsch2024DecisionA}, we infer that 1) $P(c') = \im \x_0$ and $c' = c_{f\bullet}$ for some $f\in M$, or 2a) we are in line two, four, or eight and $P(c') = \im\x_1 \cup \im\x_2$ and $c' = c_{\bullet g}$ for some $g\in M$, or 2b) we are in line one, three, five, six, or seven and $P(c') \in \{\im\x_1,\im\x_2\}$ and $c' = c_{kg}$ for some $k=1,2$ and $g\in M$.
    Given this, the $\ms F^i$-$\ms C^i$-adaptedness implies that $c'$ has to be one of the entries in the given line of the table, thus an element of $C^i$.
\end{proof}
    
\begin{proof}
    [Proof of Lemma~\ref{lemma:simple_sef2}]
    The proof of this lemma is highly analogous to that of Lemma~\ref{lemma:simple_sef1} just above. Let $(F',\pi',\X',I',\ms F',\ms C',C')$ be a tuple satisfying the hypothesis of the lemma. Let $i$ be the unique element of $I'$. Recall the table defining $C'$.
    
    According to \cite[Lemma~2.7]{Rapsch2024DecisionA}, $(F',\pi',\X')$ defines a stochastic decision forest on $(\Omega,\ms E)$. By \cite[Lemma~3.3]{Rapsch2024DecisionA}, $\ms F^{\prime i}$ defines an exogenous information structure for it. By \cite[Lemma~4.5]{Rapsch2024DecisionA}, $\ms C^{\prime i}$ defines a reference choice structure for it, and by \cite[Lemma~4.6]{Rapsch2024DecisionA}, the elements of $C^{\prime i}$ are $\ms F^{\prime i}$-$\ms C^{\prime i}$-adapted choices.

    (Ad Axiom~\ref{def:SEF}.\ref{def:SEF.P(c)}):~ The immediate predecessor sets of all choices in $C^{\prime i}$ have been explicitly calculated in \cite[Lemma Appendix A.4]{Rapsch2024DecisionA}. It follows from the results of that lemma that for all $c,c'\in C^{\prime i}$, $P(c)$ and $P(c')$ can only non-trivially intersect if $c$ and $c'$ are denoted in the same column of the table. Given this, closer inspection of the table reveals that if $P(c)$ and $P(c')$ non-trivially intersect, then $P(c) = P(c')$, and $c\cap W_\omega$ and $c'\cap W_\omega$ are either disjoint or equal, for all $\omega\in\Omega$.\smallskip

    (Ad Axiom~\ref{def:SEF}.\ref{def:SEF.outcomes_faithful}):~ There is only one agent $i\in I'$ here, hence this axiom is trivially satisfied. Exactly the same abstract argument can be used as in the proof of Lemma~\ref{lemma:simple_sef1}, Part~\ref{def:SEF}.\ref{def:SEF.outcomes_faithful}, just above.\smallskip

    (Ad Axiom~\ref{def:SEF}.\ref{def:SEF.separation}):~ Let $y,y'\in T'_\omega$ for some $\omega\in\Omega$ with $y\cap y' = \emptyset$. Then, $y$ and $y'$ are moves at time $1$ or terminal nodes. As $T_\omega$ is a finite tree, there is a $\supseteq$-minimal $x\in F$ with $x\supseteq y\cup y'$. In particular, $x$ must be a move with $\pi(x) = \omega$. Let $i\in I$. For any possible values of $x$ and $y$, there are $c,c'\in A^i(x)$ such that $x\cap c \supseteq y$ and $x\cap c' \supseteq y'$, as evident from the table and \cite[Lemma~Appendix~A.4]{Rapsch2024DecisionA}.\smallskip

    (Ad Axiom~\ref{def:SEF}.\ref{def:SEF.enough_choices}):~ Let $x\in X$, $i\in J(x)$, and $y\in\downarrow x \setminus \{x\}$. As evidenced by the table and \cite[Lemma Appendix A.4]{Rapsch2024DecisionA}, there is $c\in A^i(x)$ such that $c \supseteq y$.\smallskip

    (Ad Axiom~\ref{def:SEF}.\ref{def:SEF.endo_exo_compatible}):~ Upon consulting the table, we infer using \cite[Lemma Appendix A.4]{Rapsch2024DecisionA} that there are no two random moves sharing an available choice $c$. Hence, this axiom is trivially satisfied.\smallskip

    (Ad Axiom~\ref{def:SEF}.\ref{def:SEF.choice_completeness}):~ Let $c'$ be an $\ms F^{\prime i}$-$\ms C^{\prime i}$-adapted choice satisfying \ref{def:SEF}.\ref{def:SEF.choice_completeness.i} and~\ref{def:SEF}.\ref{def:SEF.choice_completeness.ii}. From the latter two properties and \cite[Lemma Appendix A.4]{Rapsch2024DecisionA}, we infer that 1) $P(c') = \im \x'_0$ and $c' = c'_{f\bullet}$ for some $f\in M$, or 2) $P(c') \in \{\im\x'_1,\im\x'_2\}$ and $c' = c'_{kg}$ for some $k=1,2$ and $g\in M$.
    Given this, the $\ms F^{\prime i}$-$\ms C^{\prime i}$-adaptedness implies that $c'$ has to be one of the entries in the given line of the table, thus an element of $C^{\prime i}$.
\end{proof}

\begin{proof}[Proof of Theorem~\ref{thm:absent_minded_driver_Gilboa_sef}]
    It has been demonstrated in the ``Simple examples'' subsections in \cite{Rapsch2024DecisionA} that $(F,\pi,\X)$ is a stochastic decision forest, that $\ms F^i$ is an exogenous information structure, that $\ms C^i$ is a reference choice structure, and that $C^i$ is a set of $\ms F^i$-$\ms C^i$-adapted choices, for both $i\in I$.\smallskip

    (Ad Axiom~\ref{def:SEF}.\ref{def:SEF.P(c)}):~ Recall that, for all $i\in I=\{1,2\}$, $E\in\ms F^i_{\x_i}$, we have
    \[ P(c_i(E)) = \im \x_i, \]
    and $\{\im \x_1, \im\x_2\}$ defines a (though not ``order consistent'') partition of $X$.
    Hence, for $c,c'\in C^i$, $P(c)$ and $P(c')$ can only non-trivially intersect if $c$ and $c'$ are available at the same random move, and in that case, $P(c) = P(c')$. If this is the case, then for all $\omega\in\Omega$, $c\cap W_\omega$ and $c'\cap W_\omega$ can both only equal ``exit'' or ``continue'' in scenario $\omega$, i.e.\ $\op{Ex}_i \cap W_\omega$ or $\op{Ct}_i \cap W_\omega$, and thus are either equal or disjoint.\smallskip

    (Ad Axiom~\ref{def:SEF}.\ref{def:SEF.outcomes_faithful}):~ At any move, there is exactly one active agent $i\in I$, hence this axiom is trivially satisfied. Exactly the same abstract argument can be used as in the proof of Lemma~\ref{lemma:simple_sef2}, Part~\ref{def:SEF}.\ref{def:SEF.outcomes_faithful}, just above.\smallskip

    (Ad Axiom~\ref{def:SEF}.\ref{def:SEF.separation}):~ Let $y,y'\in T_\omega$ for some $\omega\in\Omega$ with $y\cap y' = \emptyset$. Then, $\{(\omega,D)\}\in\{y,y'\}$, or, $\{\{(\omega,H)\},\{(\omega,M)\}\} = \{y,y'\}$. In the first case, take $x = \x_{\rho(\omega)}(\omega)$, and in the second, take $x=\x_{3-\rho(\omega)}(\omega)$. In both cases, we can take the two choices $c$ and $c'$ to ``continue'' and to ``exit'', or conversely, that are available at $x$, i.e.\ $x\in P(c)\cap P(c') \cap W_\omega$, such that $c\cap x \supseteq y$ and $c'\cap x \supseteq y'$. By construction, $c\cap c' = \emptyset$.\smallskip

    (Ad Axiom~\ref{def:SEF}.\ref{def:SEF.enough_choices}):~ Let $x\in X$, $i\in J(x)$, and $y\in\downarrow x \setminus \{x\}$. Then $x=\x_i(\omega)$, for some $\omega\in\Omega$. If $\rho(\omega) = i$, $y$ can be any terminal node in $T_\omega$; else, $y\in\{\{(\omega,H)\},\{(\omega,M)\}\}$. By appropriately choosing $c$ to be either ``continue'' or ``exit'', we clearly get $c\supseteq y$ in any of these cases.\smallskip

    (Ad Axiom~\ref{def:SEF}.\ref{def:SEF.endo_exo_compatible}):~ As $\X^1$ and $\X^2$ are singletons, this condition is trivially satisfied.\smallskip

    (Ad Axiom~\ref{def:SEF}.\ref{def:SEF.choice_completeness}):~ Let $i\in I=\{1,2\}$ and $c'$ be an $\ms F^{i}$-$\ms C^{i}$-adapted choice satisfying \ref{def:SEF}.\ref{def:SEF.choice_completeness.i} and~\ref{def:SEF}.\ref{def:SEF.choice_completeness.ii}. We infer that $P(c') = \im\x_i$ from the latter one. By the first condition, $c'$ must consist in ```continue'' or ``exit'' in every scenario. By adaptedness, and the definition of $\ms C^i$, both ``continue'' and ``exit'' must be chosen on $\ms F^i_{\x_i}$-measurable events. Hence, $c' = c_i(E)$ for some $E\in\ms F_{\x_i}^i$. Thus, $c'\in C^i$.
\end{proof}

\subsection{Section~\ref{sec:AP_SEF}}\label{subsec:appendix.proofs.2}

As a preparation for proving Theorem~\ref{thm:AP_sef}, we consider the following lemmata. The first lemma relates the domain $D$ of the ``random action'' $g$ to the domain of those random moves that the corresponding choice is available at in view of the second lemma. 

\begin{lemma}\label{lemma:AP_sef_compute_D}
    Let $\D$ be action path $\psi$-\textsc{sef} data on an exogenous scenario space $(\Omega,\ms E)$, $i\in I$, and $c\in C^i$. Let $t\in\T$, $A_{<t} \in \mc H_t^i$, $D\in\ms E$, $g\colon D\to\A^i$ such that $c = c(A_{<t},i,g)$. Then, for any $f\in \A^\T$ with $f|_{[0,t)_\T}\in A_{<t}$, we have
    \[ D = D_{t,f}. \]
\end{lemma}

\begin{proof}
    Let $f\in\A^\T$ be such that $f|_{[0,t)_\T}\in A_{<t}$. 

    Let $\omega\in D$. By the definition of $C^i$ and the fact that $c\in C^i$, there is $\tilde f\in\A^\T$ such that $(\omega,\tilde f)\in c$ and $\tilde f|_{[0,t)_\T} = f|_{[0,t)_\T}$. As $c\in\ms C_t$, Assumption \hyperlink{Ass:AP.C1}{AP.C1} implies that $\omega\in D_{t,\tilde f} = D_{t,f}$. We conclude that $D\subseteq D_{t,f}$. 
    
    As $c\neq\emptyset$, according to Assumption \hyperlink{Ass:AP.C0}{AP.C0}, we infer that $D\neq\emptyset$. Let $\omega\in D$ as above, whence $\omega \in D_{t,f}$. Then, as above,
    \[ (\omega,\tilde f) \in x_t(\omega,f) \cap c. \]
    By Assumption \hyperlink{Ass:AP.C2}{AP.C2}, all $\omega'\in D_{t,f}$ satisfy $x_t(\omega',f) \cap c\neq \emptyset$ and must therefore be elements of $D$. Hence, $D_{t,f} \subseteq D$.
\end{proof}

\begin{lemma}\label{lemma:AP_sef_compute_P(c)}
    Let $\D$ be action path $\psi$-\textsc{sef} data on an exogenous scenario space $(\Omega,\ms E)$, $i\in I$, and $c\in C^i$. Let $t\in\T$, $A_{<t} \in \mc H_t^i$, $D\in\ms E$, $g\colon D\to \A^i$ be such that $c = c(A_{<t},i,g)$. Then,
    \begin{equation*}
    \begin{aligned}
        P(c) =&~\{ x_t(\omega,f) \mid (\omega,f)\in D \times \A^\T\colon~ f|_{[0,t)_\T} \in A_{<t}\}\\
        =&~\{ x_t(\omega,f) \mid (\omega,f)\in \Omega \times \A^\T\colon~ \omega\in D_{t,f},~ f|_{[0,t)_\T} \in A_{<t}\}.
    \end{aligned}
    \end{equation*} 
\end{lemma}

\begin{proof}
    It suffices to prove the first equality because the second one follows from it in view of Lemma~\ref{lemma:AP_sef_compute_D}.

    As $c\in\ms C_t$ by assumption, we have 
    \[ P(c) = \{x_t(w) \mid w \in c\}, \]
    by \cite[Lemma~4.8]{Rapsch2024DecisionA}. Now, let $w = (\omega,f)\in \Omega\times \A^\T$. Then, by the definition of $c$, $w\in c$ implies that $f|_{[0,t)_\T} \in A_{<t}$ and $\omega\in D$. Conversely, if $f|_{[0,t)_\T} \in A_{<t}$ and $\omega\in D$, then there is $\tilde f\in\A^\T$ such that $\tilde w = (\omega,\tilde f)\in c$ and $\tilde f|_{[0,t)_\T} = f|_{[0,t)_\T}$ since $c\in C^i$. Hence, $x_t(w) = x_t(\tilde w) \in P(c)$.
\end{proof}

\begin{proof}
    [Proof of Proposition~\ref{prop:Dfti}]
    (Ad \ref{prop:Dfti.hatDft=unionDfti}):~ This follows directly from the fact that for $a',a''\in\A$ we have $a'=a''$ iff for all $i\in I$, $p^i(a') = p^i(a'')$.\smallskip

    (Ad \ref{prop:Dfti.hatDft_subseteq_Dft}):~ Let $t\in\T$, $f\in\A^\T$ and $\omega\in \hat D_{t,f}$. Then there are  $f',f'' \in\A^\T$ such that $(\omega,f'), (\omega,f'')\in W$, $f'|_{[0,t)_\T} = f|_{[0,t)_\T} = f''|_{[0,t)_\T}$ and $f'(t) \neq f''(t)$. Hence, $(\omega,f')$ and $(\omega,f'')$ provide two distinct elements of $x_t(\omega,f)$, whence $\omega\in D_{t,f}$.\smallskip

    (Ad \ref{prop:Dfti.Dfti_nonempty}):~ Let $t\in\T$, $f\in\A^\T$ and $i\in I$. Then the assertion that $D_{t,f} \neq \emptyset$ and $\x_t(f) \in \X^i$ hold true is equivalent to the one that there are $\omega\in \Omega$ and $c\in C^i$ such that $x_t(\omega,f)\in P(c)$.

    To show the implication ``$\Rightarrow$'', suppose that there are such $\omega$ and $c$. Then, by Lemmata \ref{lemma:AP_sef_compute_P(c)} and~\ref{lemma:AP_sef_compute_D}, we can represent $c$ using $A_{<t}\in\mc H^i_t$ and $g\colon D_{t,f} \to \A^i$ as $c = c(A_{<t},i,g)$. As $c\in\ms C_t$ by hypothesis, by Assumption \hyperlink{Ass:AP.C1}{AP.C1}, there is $f'\in\A^\T$ such that $(\omega,f')\in W$, $f'|_{[0,t)_\T} = f|_{[0,t)_\T}$ and $p^i \circ f'(t) \neq g(\omega)$. On the other hand, by Lemma~\ref{lemma:AP_sef_compute_P(c)}, the fact that $x_t(\omega,f)\in P(c)$ implies that there is $f''\in\A^\T$ satisfying $(\omega,f'')\in W$, $f''|_{[0,t)_\T} = f|_{[0,t)_\T}$ and $p^i \circ f''(t) = g(\omega)$. Hence, $\omega\in D_{t,f}^i$.

    As the preceding argument can be made for any $\omega\in D_{t,f}$ it follows that under the assumption ``$D_{t,f}\neq\emptyset$ and $\x_t(f)\in \X^i$'', we have $D_{t,f} \subseteq D_{t,f}^i$. In view of Parts~\ref{prop:Dfti.hatDft=unionDfti} and~\ref{prop:Dfti.hatDft_subseteq_Dft}, we obtain $D_{t,f} = D_{t,f}^i$.

    To show the other implication ``$\Leftarrow$'', suppose that there is $\omega\in D_{t,f}^i$. Hence, there is $f'\in\A^\T$ with $f'|_{[0,t)_\T}=f|_{[0,t)_\T}$ and $(\omega,f')\in W$. Hence, $\omega\in D_{t,f}^i = D_{t,f'}^i$. Since $\omega\in D_{t,f'}^i \subseteq D_{t,f'}$, Assumption \hyperlink{Ass:AP.SEF2}{AP.SEF2} combined with Lemma~\ref{lemma:AP_sef_compute_P(c)} directly yields the existence of $c\in C^i$ with $x_t(\omega,f) = x_t(\omega,f')\in P(c)$.
\end{proof}

\begin{example}\label{ex:SEF_with_Dtf_neq_hatDtf}
    Consider singleton $\Omega$, $\T = \R_+$ with standard order, singleton $I$, let $\A = \R$ and $W$ be the set of pairs $(\omega,f)\in\Omega\times\A^\T$ such that $f(0) = 0$ and $f$ is right-constant, that is, for all $t\in\T$, there is $\e>0$ such that $f|_{[t,t+\e)}$ is single-valued. Then, $(I,\A,\T,W)$ clearly satisfies the Assumptions \hyperlink{Ass:AP.SDF0}{AP.SDF$k$}, $k=0,\dots,3$.
    
    Further, let $i\in I$, $\ms F^i_\x = \ms E = \mc P\Omega$ for all $\x\in\tilde\X^i$, $\ms F^i = (\ms F^i_\x)_{\x\in\tilde\X^i}$, $\ms F=(\ms F^i)_{i\in I}$, and $\mc H^i_t = \big\{\{f\in \A^{[0,t)} \mid D_{t,f}^i\neq \emptyset\}\big\}$, for the unique action index $i\in I$ and all $t\in\T$, $\mc H^i = (\mc H^i_t)_{t\in\T}$, $\mc H = (\mc H^i)_{i\in I}$. Then, $(I,\A,\T,W,\ms F,\mc H)$ clearly satisfies the Assumptions \hyperlink{Ass:AP.SEF0}{AP.SEF$k$}, $k=0,\dots,3$.
    
    However, we clearly have $D_{0,f} = \Omega \neq \emptyset = \hat D_{0,f}$, for all $f\in\A^\T$.
\end{example}

\begin{proof}
    [Proof of Theorem~\ref{thm:AP_sef}]
    Let $\D$ action path $\psi$-\textsc{sef} data on an exogenous scenario space $(\Omega,\ms E)$, and let $\F$ be the induced \textsc{sef} candidate.\smallskip
    
    (Ad basic properties in Definition~\ref{def:SEF}):~ According to \cite[Theorem~2.15]{Rapsch2024DecisionA}, $(F,\pi,\X)$ defines an order consistent stochastic decision forest on $(\Omega,\ms E)$. Moreover, by Corollary \ref{cor:tildeXi=Xi}, for all $i\in I$, $\tilde\X^i = \X^i$. Hence, $\ms F$ defines a family of exogenous information structures on $\X^i$, $i\in I$, and, by \cite[Proposition 4.10]{Rapsch2024DecisionA}, $\ms C$ defines a family of reference choice structures on $\X^i$, $i\in I$. By construction, for each $i\in I$, the elements of $C^i$ are $\ms F^i$-$\ms C^i$-adapted choices. Moreover, by order consistency and \cite[Proposition~2.4]{Rapsch2024DecisionA}, evaluation $\X^i\bullet\Omega\to X$ is injective.\smallskip

    (Ad Axiom~\ref{def:SEF}.\ref{def:SEF.P(c)}):~ Let $i\in I$ and $c,c'\in C^i$ such that $P(c) \cap P(c') \neq \emptyset$. Represent $c$ and $c'$ using $t,t'\in\T$, $A_{<t}\in \mc H^i_t$, $A'_{<t'}\in\mc H^i_{t'}$, $D,D'\in\ms E$, $g\colon D\to \A^i$ and $g'\colon D'\to \A^i$ such that
    \[ c = c(A_{<t},i,g), \qquad c' = c(A'_{<t'},i,g').\]
    In view of Lemma~\ref{lemma:AP_sef_compute_P(c)}, there are $(\omega,f)\in D\times\A^\T$ and $(\omega',f')\in D'\times\A^\T$ such that $f|_{[0,t)_\T} \in A_{<t}$, $f'|_{[0,t')_\T} \in A'_{<t'}$, $p^i\circ f(t) = g(\omega)$, $p^i\circ f'(t) = g'(\omega')$, and
    \[ x_t(\omega,f) = x_{t'}(\omega',f'). \]
    By applying $\pi$, we obtain $\omega = \omega'$. By applying the ``time'' map $\mf t$, we obtain $t = t'$. By the definition of nodes in action path \textsc{sdf}, we get $f|_{[0,t)_\T} = f'|_{[0,t)_\T}$. Hence, $A_{<t} \cap A'_{<t} \neq\emptyset$. By Assumption \hyperlink{Ass:H1}{AP.H1}, $A_{<t} = A'_{<t}$. By Lemma~\ref{lemma:AP_sef_compute_D} we get $D=D'$, and by Lemma~\ref{lemma:AP_sef_compute_P(c)}, we obtain $P(c) = P(c')$. 

    We remain in the situation where $P(c) = P(c')$ and consider $\omega\in\Omega$ such that $c\cap c'\cap W_\omega \neq \emptyset$. Then, $\omega\in D$ and $g(\omega) = g'(\omega)$, which implies $c\cap W_\omega = c'\cap W_\omega$.\smallskip

    (Ad Axiom~\ref{def:SEF}.\ref{def:SEF.outcomes_faithful}):~ Let $x\in X$ and $(c^i)_{i\in J(x)} \in \bigtimes_{i\in J(x)} C^i$. Represent $x$ using $t = \mf t(x) \in \T$ and $w = (\omega,f)\in W$, as $x = x_t(w)$. For any $i\in J(x)$, represent $c^i$ using $A^i_{<t}\in\mc H^i_t$, $D^i\in\ms E$, and $g^i\colon D^i\to \A^i$, as $c = c(A^i_{<t},i,g^i)$. 

    For all $i\in J(x)$, we have $x\in P(c^i)$. By Lemma~\ref{lemma:AP_sef_compute_P(c)}, this implies that, for all $i\in J(x)$, $(\omega,f|_{[0,t)_\T})\in D^i\times A^i_{<t}$. Hence, for any $i\in J(x)$ there is $f_i\in\A^\T$ with $(\omega,f_i)\in c^i$ and $f_i|_{[0,t)_\T} = f|_{[0,t)_\T}\in A^i_{<t}$.
    
    Then, by Assumption \hyperlink{Ass:AP.SEF0}{AP.SEF0} there is $\tilde f\in\A^\T$ such that for all $i\in J(x)$, $p^i\circ \tilde f(t) = p^i \circ f_i(t) = g^i(\omega)$ and $(\omega,\tilde f)\in x_t(\omega,f_i) = x_t(\omega,f)$. We conclude that
    \[ (\omega,\tilde f) \in x \cap \bigcap_{i\in J(x)} c^i. \]

    (Ad Axiom~\ref{def:SEF}.\ref{def:SEF.weak_separation}):~ Let $y,y'\in F$ with $y\cap y' = \emptyset$ and $\pi(y) = \pi(y')$. Denote $\omega = \pi(y)$. There are $f,f'\in\A^\T$ and $t_0\in\T$ such that a) $(\omega,f)\in y$, $(\omega,f')\in y'$, and $f(t_0)\neq f'(t_0)$ and b) $t_0 < \mf t(y) \wedge \mf t(y')$.
    
    By Assumption \hyperlink{Ass:AP.psi-SEF3}{AP.$\psi$-SEF3}, there are $t\in\T$ and $i\in I$ such that $p^i \circ f(t) \neq p^i \circ f'(t)$, $\omega\in D_{t,f}^i \cap D_{t,f'}^i$, and $t\le t_0$. There is a unique pair $A_{<t},A'_{<t}\in\mc H^i_t$ such that $f|_{[0,t)_\T} \in A_{<t}$ and $f'|_{[0,t)_\T} \in A'_{<t}$. By Assumption \hyperlink{Ass:AP.SEF2}{AP.SEF2} and definition of $C$, there are $g,g'\colon D_{t,f} \to \A^i$ such that, with $c=c(A_{<t},i,g)$ and $c'=c(A'_{<t},i,g')$, we have $(\omega,f)\in c\in C^i$ and $(\omega,f')\in c'\in C^i$. In particular, $c\cap c'\cap W_\omega = \emptyset$. Moreover, $\omega\in D_{t,f}\cap D_{t,f'}$, and thus Lemma~\ref{lemma:AP_sef_compute_P(c)} implies that $c$ ($c'$) is available to agent $i$ at the random moves $\x_t(f)$ ($\x_t(f')$, respectively).

    It remains to prove that $c \supseteq y$ and $c' \supseteq y'$. In view of the problem's symmetry, it suffices to give a proof of the first inclusion.
    The inclusion $c\supseteq y$ holds true by construction if $y$ is a singleton. Else, there is $u\in\T$ such that $y = x_u(\omega,f)$. By construction, $t\le t_0 <\mf t(y) = u$, whence $t<u$. Hence, we can argue as follows. If $w\in y$, then there is $\tilde f\in\A^\T$ with $\tilde f|_{[0,u)_\T} = f|_{[0,u)_\T}$ such that $w = (\omega,\tilde f)$. In particular, $\tilde f|_{[0,t)_\T} = f|_{[0,t)_\T}$ and $\tilde f(t) = f(t)$, so that $w = (\omega,\tilde f)\in c$ as well. We conclude that $c \supseteq y$.\smallskip
    
    (Ad Axiom~\ref{def:SEF}.\ref{def:SEF.separation}):~ To show this property, we assume that $\D$ is \textsc{sef} data, so that even Assumption \hyperlink{Ass:AP.SEF3}{AP.SEF3} is satisfied. 
    
    Let $y,y'\in F$ with $y\cap y' = \emptyset$ and $\pi(y) = \pi(y')$. Denote $\omega = \pi(y)$. There are $f,f'\in\A^\T$ such that $(\omega,f)\in y$ and $(\omega,f')\in y'$, and $f\neq f'$. Hence, by Assumption \hyperlink{Ass:AP.SEF3}{AP.SEF3} there is $t\in\T$ such that $f|_{[0,t)_\T} = f'|_{[0,t)_\T}$ and $f(t) \neq f'(t)$. There is $i\in I$ such that $p^i \circ f(t) \neq p^i \circ f'(t)$. Hence, $\omega\in D_{t,f}^i$. There is unique $A_{<t}\in\mc H^i_t$ such that $f|_{[0,t)_\T} \in A_{<t}$. By Assumption \hyperlink{Ass:AP.SEF2}{AP.SEF2} and definition of $C$, there are $g,g'\colon D_{t,f} \to \A^i$ such that, with $c=c(A_{<t},i,g)$ and $c'=c(A_{<t},i,g')$, we have $(\omega,f)\in c\in C^i$ and $(\omega,f')\in c'\in C^i$. In particular, $c\cap c'\cap W_\omega = \emptyset$. Moreover, $\omega\in D_{t,f}$, and thus Lemma~\ref{lemma:AP_sef_compute_P(c)} implies that $c$ and $c'$ are available to agent $i$ at the random move $\x_t(f)$.

    It remains to prove that $x_t(\omega,f) \cap c \supseteq y$ and $x_t(\omega,f) \cap c' \supseteq y'$. By symmetry of the problem, it suffices to give a proof for the first inclusion. As $x_t(\omega,f) \cap y\neq \emptyset$, both nodes are contained in some decision path alias maximal chain in $(F,\supseteq)$ and are thus comparable. The same holds true for $x_t(\omega,f)$ and $y'$. If we had $y\supseteq x_t(\omega,f)$, then $y$ and $y'$ would be comparable since in $(F,\supseteq)$ principal up-sets are chains. But this would imply $y\cap y'\neq\emptyset$, a contradiction. Hence, $x_t(\omega,f) \supsetneq y$. The second part of the inclusion, namely $c\supseteq y$, holds true by construction if $y$ is a singleton. Else, there is $u\in\T$ such that $y = x_u(\omega,f)$. As $x_t(\omega,f) \supsetneq y$, the strict monotonicity of $\mf t$ implies that $t<u$. Hence, we can argue as follows. If $w\in y$, then there is $\tilde f\in\A^\T$ with $\tilde f|_{[0,u)_\T} = f|_{[0,u)_\T}$ such that $w = (\omega,\tilde f)$. In particular, $\tilde f|_{[0,t)_\T} = f|_{[0,t)_\T}$ and $\tilde f(t) = f(t)$, so that $w = (\omega,\tilde f)\in c$ as well. We conclude that $x_t(\omega,f) \cap c \supseteq y$.\smallskip

    (Ad Axiom~\ref{def:SEF}.\ref{def:SEF.enough_choices}):~ Let $x\in X$, $i\in J(x)$ and $y\in \downarrow x \setminus \{x\}$. Let $\omega = \pi(x)$. There is $f\in\A^\T$ such that $(\omega,f)\in y$. Then, there is $t\in\T$ such that $x = x_t(\omega,f)$. As $i\in J(x)$, we have $\x_t(f)\in\X^i$ and $\omega\in D_{t,f}$. By Proposition~\ref{prop:Dfti}, we infer that $\omega\in D_{t,f}^i$. Hence, by Assumption \hyperlink{Ass:AP.SEF2}{AP.SEF2}, there are $g\colon D_{t,f}\to\A^i$ and $A_{<t}\in\mc H_t^i$ with $f|_{[0,t)_\T}\in A_{<t}$ such that $c=c(A_{<t},i,g)$ satisfies $(\omega,f)\in c \in C^i$. By Lemma~\ref{lemma:AP_sef_compute_P(c)}, $x = x_t(\omega,f)\in P(c)$, whence $c\in A^i(x)$.

    Let $w\in y$. Then there is $\tilde f\in\A^\T$ such that $w=(\omega,\tilde f)$. If $y$ is a singleton, then $\tilde f = f$, whence $w\in c$. If $y$ is not a singleton, then there is $u\in\T$ such that $y = x_u(\omega,f)$. By strict monotonicity of $\mf t$, we have $t<u$. In the very same way as we did in the proof of the previous axiom, we infer that $\tilde f|_{[0,t)_\T} = f|_{[0,t)_\T}$ and $\tilde f(t) = f(t)$, whence $(\omega,\tilde f)\in c$. To conclude, we have shown that $c\supseteq y$.\smallskip

    (Ad Axiom~\ref{def:SEF}.\ref{def:SEF.endo_exo_compatible}):~ Let $i\in I$ and $\x,\x'\in\X$ such that $A^i(\x) \cap A^i(\x') \neq \emptyset$. Hence, there are $\omega\in D_\x$, $\omega'\in D_{\x'}$ and $c\in C^i$ such that $\x(\omega),\x'(\omega)\in P(c)$. There are $A_{<t}\in\mc H^i_t$, $D\in\ms E$, $g\colon D\to \A^i$ such that $c = c(A_{<t},i,g)$. By Lemma~\ref{lemma:AP_sef_compute_P(c)}, there are $f,f'\in\A^\T$ with $f|_{[0,t)_\T},f'|_{[0,t)_\T}\in A_{<t}$ such that $\x(\omega) = x_t(\omega,f)$ and $\x'(\omega') = x_t(\omega',f')$. As the left-hand sides are moves, $D_{t,f}$ and $D_{t,f'}$ are non-empty, and by Definition~\ref{def:sdf}.\ref{def:sdf.X.OC}, we get $\x = \x_t(f)$ and $\x' = \x_t(f')$. 
    By Definition~\ref{def:mcH}.\ref{def:mcH.msF_compatible}, we obtain 
    \[ \ms F^i_\x = \ms F^i_{\x_t(f)} = \ms F^i_{\x_t(f')} = \ms F^i_{\x'}. \]

    Note that this implies
    \[ (\ast) \qquad D_\x = D_{\x'}. \]

    It remains to show that $\ms C^i_{\x} = \ms C^i_{\x'}$. In view of the problem's symmetry, it suffices to show the inclusion ``$\subseteq$''. Let $c_1\in\ms C^i_{\x_t(f)}$. Then there are $\tilde A_{<t}\in \mc H_t^i$ and $A_t^i\in\ms B(\A^i)$ such that with $A_t$ as in ($\ms C_{\x}^i$.\ref{def:msC.2}), we have $c_1 = c(\tilde A_{<t},A_t)$, $c_1\in\ms C_t$, and for all $\tilde\omega\in D_\x$, $\x(\tilde\omega) \cap c_1 \neq\emptyset$. Hence, $f|_{[0,t)_\T}\in \tilde A_{<t}\cap A_{<t}$, and thus $\tilde A_{<t} = A_{<t}$ by Definition~\ref{def:mcH}.\ref{def:mcH.partition}.

    In view of $(\ast)$, it remains to show Property ($\ms C_\x^i$.\ref{def:msC.4}) for $\x'$, that is, for all $\tilde\omega\in D_{\x'}$, we have $\x'(\tilde\omega) \cap c_1 \neq \emptyset$. Let $\tilde\omega\in D_{\x'}$. Then $\tilde\omega\in D_\x$. Hence, there is $\tilde f\in\A^\T$ such that $(\tilde\omega,\tilde f)\in\x(\tilde\omega) \cap c_1$. By Assumption \hyperlink{Ass:AP.C1}{AP.C1}, there $\tilde f_1\in\A^\T$ such that $(\tilde\omega,\tilde f_1) \in \x(\tilde\omega) \setminus c_1$. Hence, $\tilde\omega\in D_{t,\tilde f}^i$, and thus, by Assumption \hyperlink{Ass:AP.SEF2}{AP.SEF2}, there is $c_2\in C^i$ such that a) $(\tilde\omega,\tilde f)\in c_2$, and b) there is $\tilde f'\in\A^\T$ satisfying $(\tilde\omega,\tilde f') \in c_2$ and $\tilde f'|_{[0,t)_\T} = f'|_{[0,t)_\T}$. As a consequence, $(\tilde\omega,\tilde f')\in \x'(\tilde\omega)$, $\tilde f'|_{[0,t)_\T}\in A_{<t}$, and $p^i \circ \tilde f'(t) = p^i \circ \tilde f(t)$, whence $(\tilde\omega,\tilde f')\in c_1$. We have thus shown the existence of an element of $\x'(\tilde\omega) \cap c_1$.\smallskip

    (Ad Axiom~\ref{def:SEF}.\ref{def:SEF.choice_completeness}):~ Let $i\in I$ and $c'$ an $\ms F^i$-$\ms C^i$-adapted choice satisfying Properties~\ref{def:SEF.choice_completeness.i} and~\ref{def:SEF.choice_completeness.ii} in Definition~\ref{def:SEF}. 

    By the latter one, there is $c\in C^i$ with $P(c) = P(c')$. There is $t\in\T$, $A_{<t}\in\mc H^i_t$, $D\in\ms E$, and $g\colon D \to \A^i$ such that $c = c(A_{<t},i,g)$. 
    
    Note that the non-redundancy of both $c$ and $c'$ and \cite[Lemma~4.1]{Rapsch2024DecisionA} imply that
    \begin{equation*}
    (\star)\qquad
    \begin{aligned} 
        D =&~ \{\omega\in \Omega \mid c\cap W_\omega \neq \emptyset\} = \{\omega\in\Omega \mid P(c) \cap T_\omega \neq \emptyset\} \\
        =&~ \{\omega\in\Omega \mid P(c') \cap T_\omega \neq \emptyset\} = \{\omega\in\Omega \mid c'\cap W_\omega \neq \emptyset\}.
    \end{aligned}
    \end{equation*}
    
    Let us define a function $g'\colon D \to \A^i$ as follows. Let $\omega\in D$. Then, in view of $(\star)$ and Property~\ref{def:SEF.choice_completeness.i}, there is $c_\omega\in C^i$ with $c'\cap W_\omega = c_\omega\cap W_\omega$. Such $c_\omega$ satisfies
    \[ P(c) \cap T_\omega = P(c') \cap T_\omega = P(c_\omega) \cap T_\omega \neq \emptyset,\] 
    by \cite[Lemma~4.1]{Rapsch2024DecisionA} and non-redundancy of the choices involved. We obtain, as shown shortly afterwards, that any $c_\omega\in C^i$ with $c'\cap W_\omega = c_\omega \cap W_\omega$ must admit a representation
    \[ (\ast) \qquad c_\omega = c(A_{<t},i,g_\omega) \] 
    for some $g_\omega\colon D \to \A^i$ such that for all $(\omega',f'_{<t})\in D\times A_{<t}$ there is $f'\in\A^\T$ with $f'|_{[0,t)_\T} = f'_{<t}$ and $(\omega',f')\in c_\omega$. As $c_\omega\cap W_\omega\neq\emptyset$,  $g_\omega(\omega)$ is independent of the choice of $c_\omega\in C^i$ such that $c'\cap W_\omega = c_\omega \cap W_\omega$. Hence, we can and do define $g'(\omega)$ by the equation
    \[ g'(\omega) = g_\omega(\omega). \]
    
    To see the existence of a representation as in $(\ast)$, represent $c_\omega$ -- which is an element of $C^i$ by assumption -- as $c(A_{<t;\omega},i,g_\omega)$ for some $A_{<t;\omega}\in\mc H^i_t$, $D_\omega\in\ms E$ with $\omega\in D_\omega$, and $g_\omega\colon D_\omega \to \A^i$, such that for all $(\omega',f'_{<t})\in D_\omega\times A_{<t;\omega}$ there is $f'\in\A^\T$ with $f'|_{[0,t)_\T} = f'_{<t}$ and $(\omega',f')\in c_\omega$. Then, there is an element $x\in P(c) \cap T_\omega = P(c_\omega)\cap T_\omega$, which by Lemma~\ref{lemma:AP_sef_compute_P(c)} applied to both $c$ and $c_\omega$, can be represented as $x = x_t(\omega,f)$ with $f\in \A^\T$ such that $f|_{[0,t)_\T}\in A_{<t}\cap A_{<t;\omega}$. Hence, by Definition~\ref{def:mcH}.\ref{def:mcH.partition}, $A_{<t} = A_{<t;\omega}$. Moreover, by Lemma~\ref{lemma:AP_sef_compute_D} applied to $c$ and $c_\omega$, we get $D = D_{t,f} = D_\omega$. This proves the existence of the representation $(\ast)$. In particular, $g'$ is well-defined.

    We now claim that 
    \[ (\dagger)\qquad c' = c(A_{<t},i,g'). \]
    It suffices to show that for all $\omega\in\Omega$, we have
    \[ c'\cap W_\omega = c(A_{<t},i,g') \cap W_\omega. \]
    This is clear for $\omega\in D^\complement$, by $(\star)$ and by definition of $c(A_{<t},i,g)$. If on the other hand $\omega\in D$, then there is $c_\omega\in C^i$ with $c_\omega\cap W_\omega = c' \cap W_\omega$ that can be represented as in $(\ast)$, and we have
    \[ c'\cap W_\omega = c_\omega\cap W_\omega = c(A_{<t},i,g_\omega) \cap W_\omega = c(A_{<t},i,g') \cap W_\omega,\]
    by $(\ast)$ and the definition of $g'$. This proves $(\dagger)$.

    It remains to show that 
    \begin{enumerate}[label=(\alph*),ref=(\alph*)]
        \item\label{thm:AP_SEF.proof.choice_compl.claim.i} $c'\in \ms C_t$, and
        \item\label{thm:AP_SEF.proof.choice_compl.claim.ii} for all $(\omega,f_{<t})\in D\times A_{<t}$ there is $f\in\A^\T$ with $(\omega,f)\in c'$ and $f|_{[0,t)_\T} = f_{<t}$.
    \end{enumerate}
    Regarding \ref{thm:AP_SEF.proof.choice_compl.claim.i}, Assumption \hyperlink{Ass:AP.C0}{AP.C0} is satisfied for $c'$ because it is satisfied for $c$. Indeed, as $c\neq \emptyset$ by assumption, there is $\omega\in\Omega$ such that $c\cap W_\omega \neq\emptyset$. Hence, there is $\omega\in \Omega$ satisfying $P(c) \cap T_\omega \neq \emptyset$ by non-redundancy. Thus, 
    \[ P(c') \cap T_\omega = P(c) \cap T_\omega \neq \emptyset,\]
    whence $c' \neq \emptyset$.

    For Assumption \hyperlink{Ass:AP.C1}{AP.C1}, let $w\in c'$ and let $\omega\in\Omega$ be such that $w\in W_\omega$. Let $c_\omega\in C^i$ such that $c_\omega \cap W_\omega = c' \cap W_\omega$. Then, $w\in c'\cap W_\omega = c_\omega\cap W_\omega$. As $c_\omega\in \ms C_t$, there is 
    \[ w'\in x_t(w) \setminus c_\omega = (x_t(w) \cap W_\omega) \setminus c_\omega = (x_t(w) \cap W_\omega) \setminus c' \subseteq x_t(w) \setminus c'. \]

    For Assumption \hyperlink{Ass:AP.C2}{AP.C2}, let $f\in\A^\T$ with $f|_{[0,t)_\T}\in A_{<t}$ such that $x_t(\omega_0,f) \cap c' \neq \emptyset$ for some $\omega_0 \in D_{t,f}$. Hence, $c' \cap W_{\omega_0} \neq \emptyset$, which implies $\omega_0 \in D$. Then, there is $c_{\omega_0}\in C^i$ such that $c_{\omega_0} = c(A_{<t},i,g_{\omega_0})$, for some $g_{\omega_0}\colon D\to \A^i$ as in $(\ast)$, $c_{\omega_0} \cap W_{\omega_0} = c' \cap W_{\omega_0}$ and $x_t(\omega_0,f) \cap c_{\omega_0} \neq\emptyset$. As $c_{\omega_0}\in C^i$, Lemma~\ref{lemma:AP_sef_compute_D} implies that $D = D_{t,f}$, and Lemma~\ref{lemma:AP_sef_compute_P(c)} implies that $x_t(\omega_0,f)\in P(c_{\omega_0})$. As a consequence,
    \[ x_t(\omega_0,f) \in P(c_{\omega_0}) \cap T_{\omega_0} = P(c') \cap T_{\omega_0} = P(c) \cap T_{\omega_0}, \]
    where the first equality follows from \cite[Lemma~4.1]{Rapsch2024DecisionA} and the fact that $c_{\omega_0} \cap W_{\omega_0} = c' \cap W_{\omega_0}$. Hence, $x_t(\omega_0,f)\in X^i$ and $\x_t(f)\in\X^i$. By completeness of $c$, we get that $c$ is available at the random move $\x_t(f)$.
    
    Let $\omega\in D_{t,f}$. Then, $\omega\in D$. Thus, there is $c_\omega\in C^i$ as in $(\ast)$ such that $c_\omega \cap W_\omega = c' \cap W_\omega$ and 
    \[ x_t(\omega,f)\in P(c) \cap T_\omega = P(c') \cap T_\omega = P(c_\omega) \cap T_\omega, \]
    which we can show using the same argument as above. By Lemma \cite[Lemma~4.8]{Rapsch2024DecisionA}, we have $P(c_\omega) = \{x_t(w) \mid w\in c_\omega\}$. Hence, there is $\tilde f\in\A^\T$ with $\tilde f|_{[0,t)_\T} = f|_{[0,t)_\T}$ such that $(\omega,\tilde f) \in c_\omega$. Fix such an $\tilde f$. We infer that 
    \[ (\omega,\tilde f) \in c_\omega \cap W_\omega = c' \cap W_\omega. \]
    Hence, $(\omega,\tilde f) \in x_t(\omega,f) \cap c'$. We have thus completely proven that $c'\in\ms C_t$.

    Regarding \ref{thm:AP_SEF.proof.choice_compl.claim.ii}, let $(\omega,f_{<t})\in D\times A_{<t}$. Then, by $(\ast)$, there is $c_\omega\in C^i$ such that there is $f\in\A^\T$ with $f|_{[0,t)_\T} = f_{<t}$ and
    \[ (\omega,f) \in c_\omega \cap W_\omega = c' \cap W_\omega \subseteq c'. \]

    We conclude that $c'\in C^i$, and the proof is complete.
\end{proof}

\begin{proof}
    [Proof of Proposition~\ref{prop:compute_Ai(x)_mfp}]
    (Ad \ref{prop:compute_Ai(x)_mfp.Ai(x)}):~ Let $t\in\T$ and $f\in\A^\T$ such that $D_{t,f}\neq\emptyset$. Let $\x = \x_t(f)$. Furthermore, let $c\in C^i$, and represent it by $u\in\T$, $A_{<u}\in \mc H^i_u$, $D\in\ms E$, and $g\colon D\to\A^i$ via $c = c(A_{<u},i,g)$.

    Then, $c\in A^i(\x)$ iff $\x(\omega) \in P(c)$ for some and all $\omega\in D_\x$. By Lemmata \ref{lemma:AP_sef_compute_D} and~\ref{lemma:AP_sef_compute_P(c)}, this is equivalent to $t=u$, $f|_{[0,t)_\T}\in A_{<u}=A_{<t}$ and $D=D_{t,f}$. \smallskip

    (Ad \ref{prop:compute_Ai(x)_mfp.mfp}):~ Call the presumed map introduced in the claim $\p$.
    
    To see that $\p$ defines a map indeed, it suffices to show that for any pair $(t,A_{<t})$ as in the proposition and any $f\in\A^\T$ with $f|_{[0,t)_\T}\in A_{<t}$ we have $D_{t,f} \neq \emptyset$. Indeed, by \ref{def:mcH}.\ref{def:mcH.partition}, $D_{t,f}^i \neq \emptyset$, and thus, by Proposition~\ref{prop:Dfti}, $D_{t,f}\neq\emptyset$.
    
    Regarding injectivity of $\p$, we show the following stronger statement: For any two pairs $(t,A_{<t})$ and $(u,A'_{<u})$, as in the statement, such that the respective values $\mf p = \p(t,A_{<t})$ and $\mf p' = \p(u,A'_{<u})$ non-trivially intersect, it must hold true that $t=u$ and $A_{<t}=A'_{<u}$. Indeed, if $\mf p\cap\mf p' \neq\emptyset$, then both $A_{<t}$ and $A'_{<u}$ are non-empty, there are $f,f'\in \A^\T$ with $f|_{[0,t)_\T} \in A_{<t}$, $f'|_{[0,u)_\T}\in A'_{<u}$ with $\x_t(f) = \x_u(f')$. Applying the ``time'' map $\mf t$ to both sides yields $t=u$, whence $A_{<t} \cap A'_{<u}\neq\emptyset$. Thus, by Definition~\ref{def:mcH}.\ref{def:mcH.partition}, $A_{<t} = A'_{<u}$.

    Regarding the claim about the image, we recall that according to the preceding step, the values of $\p$ are pairwise disjoint. As elements of $\mc H^i_t$ are non-empty for all $t\in\T$, the values of $\p$ are non-empty. Moreover, using Proposition~\ref{prop:Dfti}.\ref{prop:Dfti.Dfti_nonempty} and Definition~\ref{def:mcH}.\ref{def:mcH.partition}, we infer that the image $\im\p$ of $\p$ defines a partition of $\X^i$. 
    
    It remains to show that this partition equals $\mf P^i$. For showing this, let $\x,\x'\in\X^i$. $\x$ and $\x'$ belong to the same element of $\im\p$ iff there are $t\in\T$, $A_{<t}\in\mc H^i_t$, $f,f'\in\A^\T$ such that $f|_{[0,t)_\T},f'|_{[0,t)_\T}\in A_{<t}$, $\x=\x_t(f)$ and $\x'=\x_t(f')$. Note that this implies that $D_{t,f}^i,D_{t,f'}^i \ne \emptyset$. By definition of the latter events, for all $\omega\in D_{t,f}^i$ and $\omega'\in D_{t,f}^i$, $f,f'$ can even be chosen such that $(\omega,f),(\omega',f')\in W$.
    
    Hence, in view of Assumption \hyperlink{Ass:AP.SEF2}{AP.SEF2} and the definition of $C^i$, this is equivalent to the existence of $c\in C^i$ that can be represented as $c=c(A_{<t},i,g)$ by $A_{<t}\in \mc H^i_t$ and $g\colon D\to\A^i$, for some $t\in\T$ and $D\in\ms E$, such that there are $f,f'\in\A^\T$ satisfying $f|_{[0,t)_\T},f'|_{[0,t)_\T}\in A_{<t}$, $\x=\x_t(f)$, $\x'=\x_t(f')$, as well as $D_{t,f} = D = D_{t,f'}$. The latter is implied by Lemma~\ref{lemma:AP_sef_compute_D}.
    
    By Part~\ref{prop:compute_Ai(x)_mfp.Ai(x)}, this statement is equivalent to $A^i(\x) \cap A^i(\x') \neq \emptyset$. By Proposition~\ref{prop:information_sets}.\ref{prop:information_sets.A(x)_partition}, this is equivalent to $A^i(\x) = A^i(\x')$ which, by definition, is equivalent to $\x,\x'\in\mf p$ for some $\mf p\in\mf P^i$. We conclude that $\mf P^i = \im \p$.
\end{proof}

\begin{proof}
    [Proof of Theorem~\ref{thm:adapted_choices_mb_functions}]
    Let the data be given as in the proposition's statement, that is, let $(\Omega,\ms E)$ be an exogenous scenario space, $\D$ be action path $\psi$-\textsc{sef} data on it, $\F$ be the induced action path $\psi$-\textsc{sef}, $i\in I$, $t\in\T$, $A_{<t}\in\mc H^i_t$, $D\in\ms E$ and $g\colon D\to \A^i$ be a map such that $c = c(A_{<t},i,g)\in\ms C_t$.
    
    Then Parts~\ref{thm:adapted_choices_mb_functions.non_red_compl}, \ref{thm:adapted_choices_mb_functions.Dx_subseteq_D}, and~\ref{thm:adapted_choices_mb_functions.mb}.($\Leftarrow$) follow directly from parts 1, 2, and 3 in \cite[Theorem~4.13]{Rapsch2024DecisionA}. It thus remains to prove Part~\ref{thm:adapted_choices_mb_functions.mb}.($\Rightarrow$). To prove this, it suffices to show that \cite[Assumption AP.C3]{Rapsch2024DecisionA} is satisfied for $(A_{<t},i,g)$ and $\ms C^i$, in view of part 4 in \cite[Theorem~4.13]{Rapsch2024DecisionA}.

    To show this, let $\x\in\X^i$ be such that $c$ is available at $\x$. There is $(t,f)\in\T\times\A^\T$ with $\x = \x_t(f)$ and $D_{t,f} \neq\emptyset$. As $\x\in\X^i$, we have, by Proposition \ref{prop:Dfti}, $D_{t,f}^i \neq\emptyset$.

    Thus, by Assumption \hyperlink{Ass:AP.SEF1}{AP.SEF1} there is a generator $\ms G(\A^i)$ of $\ms B(\A^i)$, stable under non-trivial intersections, such that for all $G\in\ms G(\A^i)$, we have $c(A_{<t},A_t^{i,G}) \in\ms C_\x^i$ (with the notation from that assumption). This proves that \cite[Assumption AP.C3]{Rapsch2024DecisionA} is satisfied, thus completing the proof.  
\end{proof}

\begin{proof}
    [Proof of Theorem~\ref{thm:link_endogenous_information_H}]
    Let $(\Omega,\ms E)$ be an exogenous scenario space, $\D$ be action path $\psi$-\textsc{sef} data on it, $\F$ be the induced action path $\psi$-\textsc{sef}, and $i\in I$.\smallskip
    
    (Ad \ref{thm:link_endogenous_information_H.perfect_recall}.$\Rightarrow$):~ Suppose that $i$ admits perfect endogenous recall, and let $t,u\in\T$ with $t<u$, $A_{<t}\in\mc H^i_t$ and $A_{<u}\in\mc H_u^i$. Note first that $A_{<u}\neq\emptyset$, by definition (\ref{def:mcH}.\ref{def:mcH.partition}). Hence, $(\mc P p_{u,t})(A_{<u})\neq \emptyset$. Thus, we cannot have $(\mc P p_{u,t})(A_{<u}) \cap A_{<t} = \emptyset$ and $(\mc P p_{u,t})(A_{<u}) \subseteq A_{<t}$ at the same time. It therefore remains to show the following assertion:
    \[ (\ast) \qquad \Bigg[  (\mc P p_{u,t})(A_{<u}) \cap A_{<t} \neq \emptyset \quad \Longrightarrow \quad \Big[ (\mc P p_{u,t})(A_{<u}) \subseteq A_{<t}, ~ \forall f,f'\in A_{<u}\colon p^i\circ f(t) = p^i\circ f'(t) \Big]\Bigg]. \]
    Let us therefore consider $A_{<u}\in\mc H_u^i$ and $A_{<t}\in\mc H^i_t$ such that  $(\mc P p_{u,t})(A_{<u}) \cap A_{<t} \neq \emptyset$.
    
    From this, we infer that there is $f\in\A^\T$ such that $f|_{[0,u)_\T} \in A_{<u}$ and $f|_{[0,t)_\T} \in A_{<t}$. As a consequence, by Definition~\ref{def:mcH}.\ref{def:mcH.partition}, both $D_{t,f}^i\neq\emptyset$ and $D_{f,u}^i \neq\emptyset$. By Proposition~\ref{prop:Dfti}.\ref{prop:Dfti.Dfti_nonempty}, we obtain $D_{t,f}^i = D_{t,f}$ and $D_{f,u}^i = D_{f,u}$. In particular, we obtain
    \[ D_{f,u}^i = D_{f,u} \subseteq D_{t,f} = D_{t,f}^i. \]
    
    Let $\omega\in D_{f,u}^i$. Upon modifying $f$ at time $u$ and later times, we can assume that $(\omega,f)\in W$. By Assumption \hyperlink{Ass:AP.SEF2}{AP.SEF2} and the definition of $C$, there are $g_t^\omega\colon D_{t,f} \to \A^i$ and $g_u^\omega\colon D_{f,u} \to \A^i$ such that $p^i \circ f(t) = g_t^\omega(\omega)$, $p^i\circ f(u) = g_u^\omega(\omega)$, and, with $c_t^\omega = c(A_{<t},i,g_t^\omega)$ and $c_u^\omega = c(A_{<u},i,g_u^\omega)$, $c_t^\omega,c_u^\omega \in C^i$. In particular, 
    $ (\omega,f) \in c_t^\omega\cap c_u^\omega\cap W_\omega$. 
    Hence, as $i$ admits perfect endogenous recall, we have 
    $c_t^\omega \cap W_\omega \supseteq c_u^\omega \cap W_\omega$ or $c_t^\omega \cap W_\omega \subseteq c_u^\omega \cap W_\omega$. As $\omega\in D_{f,u}^i$ and $(\omega,f)\in W$, there is $f'\in\A^\T$ with $(\omega,f')\in x_u(\omega,f)$ such that $p^i \circ f'(u) \neq p^i \circ f(u) = g_u^\omega(\omega)$. Hence, $(\omega,f') \in (c_t^\omega \cap W_\omega)\setminus (c_u^\omega \cap W_\omega)$. 
    We conclude that all $\omega\in D_{f,u}^i$ satisfy
    \[ (\dagger)\qquad c_t^\omega \cap W_\omega \supseteq c_u^\omega \cap W_\omega. \]

    We now show that $(\mc P p_{u,t})(A_{<u}) \subseteq A_{<t}$. For this, let $f_{<u}\in A_{<u}$. There is an $\omega\in D_{f,u}$, and thus, as $c_u^\omega\in C^i$, there is $f'\in\A^\T$ with $(\omega,f')\in c_u^\omega$ and $f'|_{[0,u)_\T} = f_{<u}$. Then, by $(\dagger)$, $(\omega,f')\in c_t^\omega$. In particular, $p_{u,t}(f_{<u}) = f'|_{[0,t)_\T} \in A_{<t}$. This shows the claimed inclusion.
    
    Next, let $f_{<u},f'_{<u}\in A_{<u}$. Again, there is $\omega\in D_{f,u}$, and as $c_u^\omega\in C^i$, there are $f,f'\in \A^\T$ with $(\omega,f),(\omega,f')\in c_u^\omega$, $f|_{[0,u)_\T} = f_{<u}$, and $f'_{[0,u)_\T} = f'_{<u}$. By $(\dagger)$, we have
    \[ p^i \circ f_{<u}(t) = p^i \circ f(t) = g_t^\omega(\omega) = p^i\circ f'(t) = p^i \circ f'_{<u}(t). \]

    (Ad \ref{thm:link_endogenous_information_H.perfect_recall}.$\Leftarrow$):~ Suppose the right-hand condition on $\mc H^i$ to be satisfied. 
    Let $c,c'\in C^i$ and $\omega\in\Omega$ such that $c\cap c' \cap W_\omega \neq \emptyset$. We have to show that $c\cap W_\omega$ and $c'\cap W_\omega$ can be compared by set inclusion.

    Represent $c,c'\in C^i$ using $t,u\in \T$, $A_{<t}\in \mc H^i_t$, $A'_{<u}\in\mc H^u_t$, $D_t,D_u\in \ms E$, and $g_t \colon D_t \to \A^i$, $g_u\colon D_u \to \A^i$ such that $c = c(A_{<t},i,g_t)$ and $c' = c(A'_{<u},i,g_u)$. By hypothesis, there is $f\in \A^\T$ with $(\omega,f) \in c \cap c'$. 
    
    Without loss of generality, we can assume that $t\le u$. First, consider the case ``$t=u$''. 
    Then, $f|_{[0,t)_\T} \in A_{<t} \cap A'_{<u}$. Hence, by Definition~\ref{def:mcH}.\ref{def:mcH.partition}, $A_{<t} = A'_{<u}$. Moreover, $g_t(\omega) = p^i \circ f(t) = p^i \circ f(u) = g_u(\omega)$. Hence, $c\cap W_\omega = c'\cap W_\omega$.

    In the other case ``$t<u$'', $ f|_{[0,t)_\T} \in (\mc P p_{u,t})(A'_{<u}) \cap A_{<t}$. Thus, by hypothesis, $(\mc P p_{u,t})(A'_{<u}) \subseteq A_{<t}$. We claim that $c\cap W_\omega \supseteq c' \cap W_\omega$. To show this, let $w'\in c'\cap W_\omega$. Represent $w' = (\omega,f')$ for some $f'\in \A^\T$. In particular, $f'|_{[0,u)_\T} \in A'_{<u}$. We infer that $f'|_{[0,t)_\T} = p_{u,t}(f'|_{[0,u)_\T}) \in A_{<t}$, and, since $f|_{[0,u)_\T}\in A'_{<u}$, $p^i \circ f'(t) = p^i \circ f(t) = g_t(\omega)$. Hence, since $w'\in W$, we get $w'\in c \cap W_\omega$.\smallskip

    (Ad \ref{thm:link_endogenous_information_H.perfect_information}):~ By definition, $i$ has perfect endogenous information iff a) all $\mf p\in\mf P^i$ are singletons and b) for all $j\in I\setminus\{i\}$ we have $\X^i \cap \X^j = \emptyset$. By Proposition~\ref{prop:compute_Ai(x)_mfp}.\ref{prop:compute_Ai(x)_mfp.mfp}, a) is equivalent to the statement that for all $t\in\T$, all $A\in \mc H_t^i$ are singletons. By Proposition~\ref{prop:compute_Ai(x)_mfp}.\ref{prop:compute_Ai(x)_mfp.Ai(x)}, and in view of Assumption~\hyperlink{Ass:AP.SEF2}{AP.SEF2}, b) is equivalent to the statement that for all $t\in\T$, all $j\in I\setminus \{i\}$, all $A\in\mc H^i_t$ and all $A'\in\mc H^j_t$, we have $A\cap A' = \emptyset$.
\end{proof}

\subsection{Section~\ref{sec:well-posedness_equilibrium}}\label{subsec:appendix.proofs.3}

\begin{proof}[Proof of Lemma~\ref{lemma:closed_history}]
    Let $(F,\ge)$ be a decision forest and $h\in H$ be a history. Let $W(h)$ the set of all maximal chains $w$ with $h\subseteq w$, which is non-empty by the Hausdorff maximality principle. \smallskip
    
    (Ad \ref{lemma:closed_history.1}):~ Let $h_0 = \bigcap W(h)$, i.e.\ the intersection of all maximal chains $w$ with $h\subseteq w$. As $h\subseteq h_0$, $h_0 \neq\emptyset$. As an intersection of chains, $h_0$ is a chain. Concerning upward-closure, let $x\in h_0$ and let $w\in W(h)$. Then $x\in w$. As $(F,\ge)$ is a forest, $\uparrow x$ is a chain. As $w$ is a chain containing $x$, $\uparrow x \cup w$ is a chain as well. As $w$ is a maximal chain, $\uparrow x \subseteq w$. Hence, $\uparrow x \subseteq h_0$.
    
    Furthermore, by construction, $h_0$ clearly satisfies Property a). Regarding Property b), let $h_1$ be a history satisfying a), that is, $h_1 \subseteq w$ for all maximal chains $w$ with $h\subseteq w$. Then, $h_1\subseteq h_0$ by definition of $h_0$.
    
    Regarding uniqueness, let $\overline h$ an upward closed chain satisfying a) and b). Then, $h_0 \subseteq \overline h$ by b) applied to $\overline h$ and a) applied to $h_0$; and $\overline h\subseteq h_0$ by b) applied to $h_0$ and a) applied to $\overline h$. Hence, $h_0 = \overline h$.\smallskip

    (Ad \ref{lemma:closed_history.2}):~ For every $w\in W(h)$, there is $x\in w$ with $h\subseteq \uparrow x$, because $h$ is non-maximal and upward closed. As $\uparrow x \subseteq w$, we infer 
    \[ \bigcap\{\uparrow x \mid x\in F\colon h\subseteq \uparrow x\} \subseteq \bigcap W(h). \]
    On the other hand, for every $x\in F$ with $h\subseteq \uparrow x$ there is a maximal chain $w$ with $\uparrow x \subseteq w$, because $(F,\ge)$ is a forest and by the Hausdorff maximality principle. In particular, $w\in W(h)$. Moreover, for any $y\in w \setminus \uparrow x$, we have $W(x) \neq W(y)$, because $(F,\ge)$ is a decision forest. As $x\ge y$ and $h\subseteq \uparrow x$, we then have $W(y) \subsetneq W(x) \subseteq W(h)$. Hence, for such $y$, there is $w'\in W(h)$ with $y\notin w'$. Thus $\bigcap W(h) = w\cap \bigcap W(h) \subseteq \uparrow x$. We conclude that 
    \[ \bigcap W(h)\subseteq \bigcap\{\uparrow x \mid x\in F\colon h\subseteq \uparrow x\}. \]
    From Part~\ref{lemma:closed_history.1}, it follows directly that $\overline h = \bigcap W(h)$ which then implies the claim.\smallskip

    (Ad \ref{lemma:closed_history.3}):~ Note that the right-hand condition is equivalent to $W(h) = W(h')$. If this is true, then, in view of Part \ref{lemma:closed_history.1} just proven, $\overline h = \bigcap W(h) = \bigcap W(h') = \overline{h'}$. Conversely, if $\overline{h} = \overline{h'}$, and $w\in W(h)$, then $h' \subseteq\overline{h'}=\overline h\subseteq w$, hence $w\in W(h')$. Thus, $W(h)\subseteq W(h')$. Repeating the argument with $h$ and $h'$ swapped, yields $W(h) = W(h')$.\smallskip

    (Ad \ref{lemma:closed_history.4}):~ Let $B_h$ the set of lower bounds of $h$, i.e.\
    \[ B_h = \{ y\in F\mid \forall x\in h\colon x\ge y\}. \]
    The statement is equivalent to saying that $\overline h=h$ if $B_h$ has no maximum, and $\overline h = h \cup \{\max B_h\}$ otherwise. It is this latter statement that is proven in the following.

    First, suppose that $B_h$ has no maximum. It suffices to show that $F\setminus h \subseteq F\setminus \overline h$, because $h\subseteq \overline h$ by construction. Let $x\in F\setminus h$. If $x\notin B_h$, then $x\notin\overline h$, because $\overline h$ is a chain containing $h$. It remains to consider the case where $x\in B_h \setminus h$. By hypothesis, there is $x'\in B_h$ with $x\ngeq x'$. By Part~\ref{lemma:closed_history.2}, $x\notin \overline h$.

    Second, suppose that $B_h$ has a maximum. Let $x\in F$ such that $h\subseteq \uparrow x$. Then, $x\in B_h$, whence $\max B_h\in\uparrow x$. Hence, by Part~\ref{lemma:closed_history.2}, $\max B_h\in\overline h$. As $h\subseteq\overline h$ by construction, we obtain $\overline h \supseteq h\cup \{\max B_h\}$. Conversely, if $x\in\overline h \setminus h$, then $x\in w\setminus h$ for any $w\in W(h)$. As $W(h)\neq\emptyset$ and $h$ is an upward closed chain, $x\in B_h$. Furthermore, for any $y\in B_h$, $h\subseteq \uparrow y$, by definition. Thus, by Part~\ref{lemma:closed_history.2} just proven, $x\in\uparrow y$, i.e.\ $x\ge y$. Thus, $x=\max B_h$. We conclude that $\overline h = h\cup \{\max B_h\}$.
\end{proof}

\begin{proof}[Proof of Lemma~\ref{lemma:histories_canonical_inj}]
    (Ad range):~ Let $x\in X$. Then $\uparrow x$ is a non-empty, upward closed chain, and it is non-maximal because $x$ is not terminal. Clearly, $x = \inf \uparrow x$ and $x\in\uparrow x$, so that $\uparrow x$ is closed by \ref{lemma:closed_history}. Further, let $\x\in\X$ and $(F,\pi,\X)$ be order consistent. If there is $\x'\in\X$ admitting $\omega\in D_\x \cap D_{\x'}$ with $\x'(\omega) \supseteq \x(\omega)$, then by order consistency, $D_{\x'}\supseteq D_\x$ and $\x'(\omega') \supseteq \x(\omega')$ for all $\omega'\in D_\x$.\smallskip

    (Ad injectivity):~ First, let $x,x'\in X$ be such that $\uparrow x = \uparrow x'$. Then, $x\supseteq x' \supseteq x$, whence $x=x'$. Second, let $\x,\x'\in\X$ be such that $D_{\x} = D_{\x'}$ and $\uparrow \x(\omega) = \uparrow \x'(\omega)$ for all $\omega\in D_\x$. Then, by the first part of the proof, $\x(\omega) = \x'(\omega)$ for all $\omega\in D_\x$, whence $\x = \x'$.
\end{proof}

\begin{lemma}\label{lemma:random_histories_h_f(h)}
    Let $(F,\pi,\X)$ be an order consistent, surely non-trivial, and maximal stochastic decision forest on an exogenous scenario space $(\Omega,\ms E)$. Let $\h$ be a random history with domain $D_\h$, $\omega\in D_\h$, and $f(\h) =  \{\x\in\X \mid \exists \omega'\in D_\x\cap D_\h\colon \x(\omega')\in\h(\omega')\}$. Then, we have:
    \[ \h(\omega) = \{ \x(\omega) \mid \x\in f(\h)\}. \]
\end{lemma}

\begin{proof}
    Given the data $(F,\pi,\X)$, $\h$, $\omega$, $f$ from the lemma, let $\x\in\X$. Then $\x\in f(\h)$ iff there is $\omega'\in D_\x\cap D_\h$ such that $\x(\omega')\in\h(\omega')$. As $\h$ is a random history, this latter statement is equivalent to saying that $D_\x\supseteq D_\h$ and for all $\omega'\in D_\h$ we have $\x'(\omega') \in \h(\omega')$. By the same argument, though, this latter statement is equivalent to saying that $D_\h\subseteq D_\x$ and $\x(\omega)\in \h(\omega)$.
\end{proof}

\begin{proof}[Proof of Proposition~\ref{prop:random_histories}]
    (Ad \ref{prop:random_histories.image_H_Tr}):~ Let $\h\in\H$. Recall that, by definition, $f(\h) = \{\x\in\X \mid \exists \omega\in D_\x\cap D_\h\colon \x(\omega)\in\h(\omega)\}$. 

    There is $\omega\in D_\h$ and $\h(\omega)$ is non-empty. There is, thus, $\x\in\X$ with $\omega\in D_\x$ and $\x(\omega)\in\h(\omega)$ because $\X$ covers $X$. Hence, $\x\in f(\h)$. We therefore see that $f(\h)$ is non-empty.
    
    Let $\x,\x'\in f(\h)$. Then, $D_\x\cap D_{\x'} \supseteq D_\h$ which is non-empty. Hence, there is $\omega\in D_\h$ with $\x(\omega),\x'(\omega)\in\h(\omega)$. As $\h(\omega)$ is a chain, $\x(\omega)$ and $\x'(\omega)$ are related via $\supseteq$. Up to changing the roles of $\x$ and $\x'$, we obtain $\x(\omega)\supseteq \x'(\omega)$, and thus, by order consistency, $\x\ge_\X \x'$. Thus, $f(\h)$ is a chain.

    Let $\x\in f(\h)$ and $\x'\in\X$ such that $\x'\ge_\X \x$. Then, $D_{\x'}\supseteq D_\x\supseteq D_\h$, and for all $\omega\in D_\h$ we have $\x'(\omega) \supseteq \x(\omega)$. As $\h(\omega)$ is upward closed, $\x'(\omega)\in\h(\omega)$ for all $\omega\in D_\h$. Hence, $\x'\in f(\h)$, and $f(\h)$ is upward closed.

    Regarding non-maximality, we first note that for any $\omega\in D_\h$, there is $y\in F\setminus \h(\omega)$ with $x \supsetneq y$ for all $x\in\h(\omega)$, by non-maximality of $\h(\omega)$. In that case we infer that, for all $\x\in f(\h)$, $\omega\in D_\x$ and $\x(\omega)\in\h(\omega)$, whence $\x(\omega) \supsetneq y$. 
    
    If a) there is $\omega\in D_\h$ and $y\in X\setminus \h(\omega)$ with $x \supsetneq y$ for all $x\in\h(\omega)$, then there is $\x'\in\X$ with $y = \x'(\omega)$, whence $\x>_\X \x'$. Then, we conclude that $f(\h)$ is a non-maximal chain in $(\X,\ge_\X)$ and that $f(\h) \in H_\Tr$.
    
    If, however, b) there is no such pair $(\omega,y)$, then $f(\h)$ is a maximal chain in $(\X,\ge_\X)$ and for any $\omega\in D_\h$, there is a terminal node $y\in F \setminus \h(\omega)$ with $x \supsetneq y$ for all $x\in\h(\omega)$. We infer that there is a random terminal node $\y\colon \{\omega\} \to \{y\}$, and $\x >_\Tr\y = \{(\omega,y)\}$. Hence, $D_\h \subseteq \bigcap_{\x\in h_\Tr} D_\x$ and for any $\omega\in D_\h$, there is $w_\omega\in W_\omega$ with $\y(\omega) = \{w_\omega\}$ whence $w_\omega \in \bigcap_{\x\in h_\Tr} \x(\omega)$. Moreover, $D_\h\neq \emptyset$, which implies that $f(\h)$ is not a maximal chain in $(\Tr,\ge_\Tr)$. $D_\h \in\ms E$ by hypothesis. We have shown that $f(\h)$ is contained in the set described in the claim by the disjunction of conditions a) and b). \smallskip

    For the converse statement, let $h_\Tr$ be an element of the set described in the claim by the disjunction of conditions a) and b). First, suppose a) that $h_\Tr\in H_\Tr$ is a non-maximal chain in $(\X,\ge_\X)$. Then, there is $\x_1\in\X$ with $\x>_\X \x_1$ for all $\x\in h_\Tr$. Let $D = D_{\x_1}$. Second, suppose b) that $h_\Tr$ is a maximal chain in $(\X,\ge_\X)$ admitting non-empty $D\in\ms E\setminus\{\emptyset\}$ with $D\subseteq \bigcap_{\x\in h_\Tr} D_\x$ such that for any $\omega\in D$ there is $w\in W_\omega$ with $w\in \bigcap_{\x\in h_\Tr} \x(\omega)$.

    In both cases, we have a non-empty event $D$ and we define a map $\h$ with domain $D$ by letting 
    \[ (\ast)\qquad\h(\omega) = \{\x(\omega) \mid \x\in h_\Tr\}, \]
    for all $\omega\in D$. We claim that $\h\in\H$ and that $f(\h) = h_\Tr$.
    
    Regarding the first of these claims, the domain $D$ is by construction a non-empty event. Moreover, for any $\omega\in D$, $\h(\omega)\in H$. Indeed, let $\omega\in D$. Then, as $h_\Tr$ is a non-empty chain, $\h(\omega)$ is so, too, by definition of $\ge_\Tr$. Moreover, let $x\in\h(\omega)$ and $x'\in\uparrow x$. Then, there is $\x\in h_\Tr$ with $x=\x(\omega)$, and $x'\in X$. Accordingly, there is $\x'\in\X$ with $x' = \x'(\omega)$. Thus, by order consistency, $\x'\ge_\X \x$. As $h_\Tr$ is upward closed, $\x'\in h_\Tr$, whence $x'=\x'(\omega)\in h(\omega)$. Hence, $\h(\omega)$ is upward closed. Furthermore, $\h(\omega)$ is not maximal as a chain. To show this, we distinguish the two cases from before. In case a), $\x>_\Tr \x_1$ by construction, in particular, $\x(\omega) \supsetneq \x_1(\omega)$ for all $\x\in h_\Tr$. In case b), $\x(\omega) \supsetneq w_\omega$ for all $\x\in h_\Tr$. Hence, $\h(\omega)$ is not maximal as a chain in $(F,\supseteq)$. This shows that for any $\omega\in D$, $\h(\omega)\in H$. Further, let $\x\in\X$ admit $\omega\in D_\x\cap D$ with $\x(\omega)\in\h(\omega)$. Then, by order consistency and \cite[Proposition 2.4]{Rapsch2024DecisionA}, $\x\in h_\Tr$. Hence, $D_\x\supseteq D$ and $\x(\omega')\in\h(\omega')$ for all $\omega'\in D$. It is thus proven that $\h\in\H$.

    Regarding the second of the two claims above, namely that $f(\h) = h_\Tr$, let $\x\in\X$. By definition of $f$, $\x\in f(\h)$ holds true iff there is $\omega\in D_\x\cap D_\h$ with $\x(\omega)\in\h(\omega)$. By definition of $\h$, the latter is equivalent to saying that there is $\omega\in D_\x\cap D$ with $\x(\omega) = \x'(\omega)$ for some $\x'\in h_\Tr$. By order consistency and \cite[Proposition 2.4]{Rapsch2024DecisionA}, and by definition of $D$, this is equivalent to $\x\in h_\Tr$. Thus, $f(\h) = h_\Tr$.    
    \smallskip

    (Ad \ref{prop:random_histories.faithful}):~ Let $\h_1,\h_2\in\H$ such that $f(\h_1) = f(\h_2)$. Let $\omega\in D_{\h_1}\cap D_{\h_2}$. \\Then, by Lemma~\ref{lemma:random_histories_h_f(h)}, 
    \[ \h_1(\omega) = \{\x(\omega) \mid \x\in f(\h_1)\} = \{\x(\omega) \mid \x\in f(\h_2)\} = \h_2(\omega). \]
    Hence, there is a map $\h$ on $D = D_{\h_1} \cup D_{\h_2}$ satisfying $\h|_{D_{\h_k}} = \h_k$ for both $k=1,2$. By construction, $D$ is a non-empty event and $\h(\omega)$ is a history in $(F,\supseteq)$ for any $\omega\in D$. Moreover, we clearly have $f(\h_1) = f(\h) = f(\h_2)$. Thus, if there are $\x\in\X$ and $\omega\in D_\x\cap D_\h$ with $\x(\omega)\in\h(\omega)$, then $\x\in f(\h) = f(\h_1) = f(\h_2)$, and thus, $D_\x\supseteq D_{\h_1} \cup D_{\h_2} = D_\h$ and for all $\omega'\in D_\h$ we have $\x(\omega')\in \h(\omega')$, since $\h_1,\h_2\in\H$.\smallskip

    (Ad \ref{prop:random_histories.closed}):~ Let $\h\in\H$ be a closed random history. We have to show that $\overline{f(\h)} = f(\h)$. In view of Lemma~\ref{lemma:closed_history} it remains to show that, if $f(\h)$ admits an infimum in $(\Tr,\ge_\Tr)$, then $\inf f(\h)\in f(\h)$. Suppose that $f(\h)$ has an infimum. There is $\omega\in D_\h$. From order consistency and Lemma~\ref{lemma:random_histories_h_f(h)}, it follows directly that $(\inf f(\h))(\omega)$ is an (and the) infimum of $\h(\omega)$ in $(F,\supseteq)$. As $\h(\omega)$ is a closed history by hypothesis, we get $(\inf f(\h))(\omega)\in\h(\omega)$. Thus $\inf f(\h)\in\X$, $D_\h\subseteq D_{\inf f(\h)}$, and $\inf f(\h) \in f(\h)$.
\end{proof}

\begin{proof}[Proof of Lemma~\ref{lemma:R_continuous_in_h}]
    Let $\F$ be a stochastic pseudo-extensive form on an exogenous scenario space $(\Omega,\ms E)$, $s\in S$, $h\in H$ and $w\in W$. It suffices to show that for all $x\in X$, we have
    \[ x\subseteq \bigcap h \qquad \Longleftrightarrow \qquad x\subseteq \bigcap \overline h. \]
    This is trivial if $h=\overline h$. Therefore, suppose that the latter is not the case. By Lemma~\ref{lemma:closed_history}, $h$ has an infimum and $\overline h = h \cup \{\inf h\}$. As $h\subseteq \overline h$, the implication ``$\Leftarrow$'' is evident. For the converse one, suppose that $x\subseteq \bigcap h$. In other words, $x\subseteq y$ for all $y\in h$. By definition of the infimum, we infer $x\subseteq \inf h = \bigcap \overline h$.
\end{proof}

\begin{lemma}\label{lemma:extension_of_X_strategies}
    Let $\F$ be a stochastic pseudo-extensive form, $i\in I$, and $s^i_0 \colon X^i_0 \to C^i$ be a map on some set of $i$'s moves $X^i_0 \subseteq X^i$. Then $s^i_0$ is the restriction of an $X$-strategy iff 
    \begin{enumerate}
        \item\label{lemma:extension_of_X_strategies.1} for all $x\in X^i_0$, we have $s^i(x) \in A^i(x)$;
        \item\label{lemma:extension_of_X_strategies.2} for all $x,x'\in X^i_0$ with $A^i(x) = A^i(x')$, we have $s^i(x) = s^i(x')$.
    \end{enumerate}
\end{lemma}

\begin{proof}
    $S^i$ is non-empty because $A^i(\mf p) \neq \emptyset$ for all $\mf p\in \mf P^i$, by definition. Hence, in view of Proposition~\ref{prop:strategies} we can choose an $X$-strategy $s_1^i$. Then, define $s^i$ as follows. Let $x\in X^i$. If there is $x_0\in X^i_0$ with $A^i(x) = A^i(x_0)$, let $s^i(x) = s^i_0(x_0)$. Else, let $s^i(x) = s_1^i(x)$. By Property~\ref{lemma:extension_of_X_strategies.2}, $s^i$ is well-defined.

    By Property~\ref{lemma:extension_of_X_strategies.1} and the fact that $s^i_1$ is an $X$-strategy, we clearly have $s^i(x)\in A^i(x)$ for all $x\in X^i$. Moreover, if $x,x'\in X^i$ satisfy $A^i(x) = A^i(x')$, there are two cases. First, if there is $x_0\in X^i_0$ with $A^i(x) = A^i(x_0)$, then we also have $A^i(x') = A^i(x_0)$, hence $s^i(x) = s^i_0(x_0) = s^i(x')$. Else, there is no such $x_0$. Then, neither is there $x_0\in X_0^i$ with $A^i(x') = A^i(x_0)$. Hence, $s^i(x) = s_1^i(x) = s_1^i(x') = s^i(x')$, because $s_1^i$ is an $X$-strategy.
\end{proof}

\begin{proof}[Proof of Theorem~\ref{thm:well_posed.onto_outcomes}]
    Let $h\in H$ and $w\in\bigcap h$. Let $i\in I$ and $X_0^i = \{x\in X^i \mid w\in x\subseteq \bigcap h\}$. 
    
    As $F$ is a decision forest, $\uparrow \{w\} = \{ y\in F \mid w\in y\}$ is a maximal chain. Let $x\in X_0^i$. Then, $x\in X$ and, thus, $\uparrow x$ is not a maximal chain. As a consequence, there is $y\in\downarrow x\setminus\{x\}$ with $w\in y$, and by Axiom~\ref{def:SEF}.\ref{def:SEF.enough_choices}, there is $c\in A^i(x)$ with $c\supseteq y$. Let $s_0^i(x) = c$. 
    
    This defines a map $s^i_0\colon X^i_0 \to C^i$ with $s^i_0(x) \in A^i(x)$ for all $x\in X_0^i$. By the Heraclitus property from Lemma~\ref{lemma:Heraclitus_property}, $x,x'\in X^i_0$ with $x\neq x'$ necessarily satisfy $A^i(x) \cap A^i(x') = \emptyset$. Hence, by Lemma~\ref{lemma:extension_of_X_strategies}, $s_0^i$ can be extended to an $X$-strategy $s^i$, which uniquely corresponds to a strategy by Proposition~\ref{prop:strategies}. Letting $s = (s^i)_{i\in I}$, we have $s\in S$ and, by construction, $w\in R(s,w\mid h)$.
\end{proof}

\begin{proof}[Proof of Proposition~\ref{prop:scwise_SEF}]
    Let $\F$ be a stochastic pseudo-extensive form on an exogenous scenario space $(\Omega,\ms E)$, and let $\omega\in\Omega$. Let $C_\omega$ be defined as in the claim. Let $\Omega_\omega = \{\omega\}$ and $\ms E_\omega = \mc P\Omega_\omega$.\smallskip 

    (Ad ``$T_\omega$ induces a stochastic decision forest''):~ $T_\omega$ is a connected component of the decision forest $F$. Hence, by \cite[Theorem~1.7]{Rapsch2024DecisionA}, it is a decision tree over $W_\omega$. With $\pi_\omega = \pi|_{T_\omega}$ and ${\X_\omega}$ being the set of maps $\Omega_\omega \to X$, $(T_\omega,\pi_\omega,{\X_\omega})$ is a stochastic decision forest on $(\Omega_\omega,\ms E_\omega)$.\smallskip

    (Ad ``$C_\omega$ is a family of sets of $\X^i_\omega$-complete choices and evaluation maps are injective''):~ Let $c\in C^i$ such that $c\cap W_\omega \neq \emptyset$. There is a non-empty set $F_c\subseteq F$ of nodes such that $c = \bigcup F_c$. As $c\cap W_\omega = \bigcup (F_c \cap T_\omega)$, $F_c \cap T_\omega$ is non-empty. Moreover, $F_c \cap T_\omega$ is a set of nodes in $T_\omega$. Hence, $c\cap W_\omega$ is a choice in $(T_\omega,\pi_\omega,{\X_\omega})$.
    
    Let ${X_\omega}$, $P_\omega(.)$, $A^i_\omega(.)$, $J_\omega(.)$, ${\X^i_\omega}$, etc.\ be associated to $({T_\omega},\pi_\omega,{\X_\omega})$ and $C_\omega$ as $X$, $P(.)$, $A^i(.)$, $J(.)$, $\X^i$ etc.\ are associated to $(F,\pi,\X)$ and $C$. We infer that for any $c\subseteq W$ we have
    \begin{equation}\label{eq:P_omega(c_omega)}
        P_\omega(c\cap W_\omega) = P(c\cap W_\omega) = P(c) \cap T_\omega,
    \end{equation}
    using the definition of $P$ and $P_\omega$, as well as \cite[Lemma~4.1]{Rapsch2024DecisionA}. 
    Furthermore, note that for all $x\in {X_\omega}$ and $\x\in {\X_\omega}$ we have
    \begin{equation}\label{eq:Ai'(x)}
        A^i_\omega(x) = \{ c \cap W_\omega \mid c\in A^i(x)\}, \qquad A^i_\omega(\x) = A^i_\omega(\x(\omega)). 
    \end{equation}
    As a consequence, we have
    \begin{equation}\label{eq:Xi'}
        {X^i_\omega} = X^i\cap T_\omega, \qquad {\X^i_\omega} =  (X^i\cap T_\omega)^{\Omega_\omega}.
    \end{equation}
    Moreover, ${\X^i_\omega}\bullet\Omega_\omega = {\X^i_\omega} \times \Omega_\omega$, and the evaluation map from that set to ${X_\omega}$ is clearly injective.
    \smallskip

    For any $i\in I$ and $\x\in{\X^i_\omega}$, let $(\ms F_\omega)^{i}_{\x} = \ms E_\omega$ and $(\ms C_\omega)^{i}_{\x} = \emptyset$. Let $\ms F_\omega = ((\ms F_\omega)^i)_{i\in I}$ and $\ms C_\omega = ((\ms C_\omega)^{i})_{i\in I}$. Let
    \[ \F_\omega = ({T_\omega},\pi_\omega,{\X_\omega},I,\ms F_\omega,\ms C_\omega,C_\omega). \]

    (Ad ``$\ms F_\omega$ and $\ms C_\omega$ define adequate \textsc{eis} and reference choice structures'' and Axioms~\ref{def:SEF.endo_exo_compatible}, \ref{def:SEF.choice_completeness}):~ For any $i\in I$, $(\ms F_\omega)^{i}$ defines an exogenous information structure on ${\X_\omega^i}$ and $(\ms C_\omega)^{i}$ defines a reference choice structure on ${\X_\omega^i}$. Trivially, any $c\in C^i_\omega$ is $(\ms F_\omega)^{i}$-$(\ms C_\omega)^{i}$-adapted. Axioms  \ref{def:SEF.endo_exo_compatible} and~\ref{def:SEF.choice_completeness} are trivially satisfied.\smallskip

    (Ad Axiom~\ref{def:SEF.P(c)}):~ Let $i\in I$ and $c_\omega,c'_\omega\in C^i_\omega$ such that $P_\omega(c_\omega) \cap P_\omega(c'_\omega) \neq \emptyset$. There are $c,c'\in C^i$ with $c_\omega = c\cap W_\omega$ and $c'_\omega = c'\cap W_\omega$. Hence, using Equation~\ref{eq:P_omega(c_omega)}, we infer $P(c) \cap P(c') \neq \emptyset$. Thus, by Axiom~\ref{def:SEF.P(c)} applied to $\F$, we get $P(c) = P(c')$. Using Equation~\ref{eq:P_omega(c_omega)}, we obtain
    \[ P_\omega(c_\omega) = P(c) \cap T_\omega = P(c') \cap T_\omega = P_\omega(c'_\omega). \]
    Furthermore, Axiom~\ref{def:SEF.P(c)} applied to $c$, $c'$, and $\omega$ implies that $c_\omega = c'_\omega$ or $c_\omega \cap c'_\omega = \emptyset$. Hence, $\F_\omega$ satisfies Axiom~\ref{def:SEF.P(c)}.\smallskip

    (Ad Axiom~\ref{def:SEF.outcomes_faithful}):~ Let $x\in T_\omega$ and $(c^i_\omega)_{i\in J_\omega(x)} \in \bigtimes_{i\in J_\omega(x)} C^i_\omega$. Using Equation~\ref{eq:Ai'(x)}, we get $J_\omega(x) = J(x)$. Hence, there is $(c^i)_{i\in J(x)} \in \bigtimes_{i\in J(x)} C^i$ such that $c^i_\omega = c^i \cap W_\omega$ for all $i\in J(x)$. Then, Axiom~\ref{def:SEF.outcomes_faithful} applied to $x$ and $(c^i)_i$ yields:
    \[ x\cap \bigcap_{i\in J_\omega(x)} c^i_\omega = x\cap \bigcap_{i\in J(x)} c^i \neq \emptyset. \]

    (Ad Axiom~\ref{def:SEF.weak_separation}):~ Let $y,y'\in T_\omega$ disjoint. By Axiom~\ref{def:SEF.weak_separation} applied to $\F$, there are $i\in I$ and $c,c'\in C^i$ such that $y\subseteq c$, $y'\subseteq c'$, and $c\cap c'\cap W_\omega = \emptyset$. Then, $c_\omega = c\cap W_\omega\neq \emptyset$ and $c'_\omega = c'\cap W_\omega \neq \emptyset$, hence, $c_\omega,c'_\omega\in C^i_\omega$, and $y\subseteq c_\omega$, $y'\subseteq c'_\omega$, and $c_\omega\cap c'_\omega = \emptyset$.\smallskip

    (Ad Axiom~\ref{def:SEF.separation}):~ For this Axiom, suppose that $\F$ is even a stochastic extensive form. Let $y,y'\in T_\omega$ disjoint. By Axiom~\ref{def:SEF.separation} applied to $\F$, there are $x\in X$, $i\in I$ and $c,c'\in C^i$ such that $y\subseteq x\cap c$, $y'\subseteq x\cap c'$, $c\cap c'\cap W_\omega = \emptyset$, and $x\in P(c) \cap P(c') \cap W_\omega$. Then, $c_\omega = c\cap W_\omega\neq \emptyset$ and $c'_\omega = c'\cap W_\omega \neq \emptyset$, hence, $c_\omega,c'_\omega\in C^i_\omega$, and $y\subseteq x\cap c_\omega$, $y'\subseteq x\cap c'_\omega$, $c_\omega\cap c'_\omega = \emptyset$, and $x\in P_\omega(c) \cap P_\omega(c')$, in view of Equation~\ref{eq:P_omega(c_omega)}.\smallskip

    (Ad Axiom~\ref{def:SEF.enough_choices}):~ Let $x\in T_\omega$, $i\in J_\omega(x)$, and $y\in \downarrow x \setminus \{x\}$. Using Equation~\ref{eq:Ai'(x)}, we get $J_\omega(x) = J(x)$. Hence, by Axiom~\ref{def:SEF.enough_choices} applied to $\F$, there is $c\in A^i(x)$ with $c\supseteq y$. Hence, by Equation~\ref{eq:Ai'(x)}, $c_\omega = c\cap W_\omega\in A^i_\omega(x)$ and $c_\omega \supseteq y$.
\end{proof}

\begin{proof}[Proof of Theorem~\ref{thm:SEF_well-posed}]
    The claimed equivalences follow directly from the following statements:
    \begin{enumerate}
        \item\label{thm:SEF_well-posed.proof.H} $H = \bigcup_{\omega\in\Omega} H_\omega$, where $H_\omega$ is the set of histories in $(T_\omega,\supseteq)$;
        \item\label{thm:SEF_well-posed.proof.Homega_Womega_comp} for all $\omega\in\Omega$ and $h\in H_\omega$, we have $\bigcap h \subseteq W_\omega$;
        \item\label{thm:SEF_well-posed.proof.si_restriction} if $s^i\in S^i$ is an $X$-strategy for agent $i\in I$ and $\omega\in\Omega$, then the map $s^i_\omega$ with domain $X^i\cap T_\omega$ defined by the assignment $x \mapsto s^i(x) \cap W_\omega $ defines an $X$-strategy in $(T_\omega,I,C_\omega)$;
        \item\label{thm:SEF_well-posed.proof.siomega_extension} if conversely $s^i_\omega$ is an $X$-strategy for an agent $i\in I$ in $(T_\omega,I,C_\omega)$, for some $\omega\in\Omega$, then there is an $X$-strategy $s^i$ for $i$ in $\F$ such that $s^i_\omega(x) = s^i(x) \cap W_\omega$ for all $x\in X^i\cap T_\omega$;
        \item\label{thm:SEF_well-posed.proof.R=R_omega} for all $s\in S$, $\omega\in\Omega$, $h\in H_\omega$ and $w\in \bigcap h$ we have
        \[ R(w,s \mid h) \cap W_\omega = R_\omega(w,s_\omega \mid h), \]
        where $s_\omega = (s^i_\omega)_{i\in I}$, $s^i_\omega$ is the restriction of $s^i$ according to \ref{thm:SEF_well-posed.proof.si_restriction} for each $i\in I$, and $R_\omega$ is the map $R(\F_\omega)$ associated to the $\psi$-\textsc{sef} $\F_\omega$ according to Definition~\ref{def:sef_well-posed}.
    \end{enumerate}

   (Ad \ref{thm:SEF_well-posed.proof.H}):~ For $\omega\in\Omega$, let $H_\omega$ be the set of histories in $(T_\omega,\supseteq)$. Then, similarly, $H$ is the disjoint union of all $H_\omega$, because histories are chains, and the connected components of $(F,\supseteq)$ are given by the collection of all $T_\omega$, $\omega\in\Omega$.\smallskip

   (Ad \ref{thm:SEF_well-posed.proof.Homega_Womega_comp}):~ For all $\omega\in\Omega$, $W_\omega$ is the root of $(T_\omega,\supseteq)$, whence the claim using Part~\ref{thm:SEF_well-posed.proof.H}.\smallskip

    (Ad \ref{thm:SEF_well-posed.proof.si_restriction}):~ Let $s^i\in S^i$ is an $X$-strategy for agent $i\in I$ and $\omega\in\Omega$. By Equation~\ref{eq:Xi'} established in the proof of Proposition~\ref{prop:scwise_SEF} we have ${X^i_\omega} = X^i \cap T_\omega$. Moreover, if $x\in X^i\cap T_\omega$, then $s^i(x) \in A^i(x)$ by definition of $s^i$. Hence, by Equation~\ref{eq:Ai'(x)} established in the proof of Proposition~\ref{prop:scwise_SEF}, we get $s^i(x) \cap W_\omega\in A^i_\omega(x)$. 
    
    Furthermore, for all $x,x'\in {X^i_\omega}$ with $A^i_\omega(x) = A^i_\omega(x')$ we have $A^i(x) = A^i(x')$. Indeed, if $x,x'\in {X^i_\omega}$ satisfy $A^i_\omega(x) = A^i_\omega(x')$, then there is $c\in C^i$ with $x,x'\in P_\omega(c\cap W_\omega) = P(c\cap T_\omega)\subseteq P(c)$, by Proposition~\ref{prop:information_sets} and Equation~\ref{eq:P_omega(c_omega)}. Thus, by Proposition~\ref{prop:information_sets}, $A^i(x) = A^i(x')$. As $s^i$ is an $X$-strategy, we obtain $s^i(x) = s^i(x')$. As a consequence, $s^i(x) \cap W_\omega = s^i(x') \cap W_\omega$.\smallskip

    (Ad \ref{thm:SEF_well-posed.proof.siomega_extension}):~ Conversely, let $s^i_\omega$ be an $X$-strategy for an agent $i\in I$ in $(T_\omega,I,C_\omega)$. Let $X_0^i$ be a representative system of the partition $\{P_\omega(c_\omega) \mid c_\omega\in C^i_\omega\}$, see Proposition~\ref{prop:information_sets}, Part~\ref{prop:information_sets.P(c)_partition}. Then, by definition of $C^i_\omega$ and by Equation~\ref{eq:Ai'(x)}, there is a map $\tilde s^i_0 \colon X^i_0 \to C^i$ with $\tilde s^i_0(x) \cap W_\omega = s^i_\omega(x)$ and $\tilde s^i_0(x) \in A^i(x)$, for all $x\in X^i_0$. For all $x,x'\in X^i_0$ with $A^i(x) = A^i(x')$, we have $A^i_\omega(x) = A^i_\omega(x')$, by Equation~\ref{eq:Ai'(x)}. Thus, $x,x'\in P_\omega(c_\omega)$ for some $c_\omega\in C^i_\omega$, by Proposition~\ref{prop:information_sets}, Part~\ref{prop:information_sets.A(x)=A(x')}, applied to $\F_\omega$. As $X^i_0$ is a representative system, we infer $x=x'$, whence $\tilde s^i_0(x) = \tilde s^i_0(x')$.

    Hence, by Lemma~\ref{lemma:extension_of_X_strategies} $\tilde s^i_0$ can be extended to an $X$-strategy $s^i$ for $i$ on $\F$. Let $x\in X^i \cap T_\omega$. Then, by Equation~\ref{eq:Xi'}, $x\in {X^i_\omega}$ and consequently there is $c_\omega\in C^i_\omega$ with $x\in P_\omega(c_\omega)$. There are $x_0\in X^i_0$ with $x_0\in P_\omega(c_\omega)$ and $c\in C^i$ with $c_\omega = c\cap W_\omega$. Hence, by Equation~\ref{eq:P_omega(c_omega)}, $x,x_0\in P(c)$ and $x,x_0 \in P_\omega(c_\omega)$, which implies both $A^i(x) = A^i(x_0)$ and $A^i_\omega(x) = A^i_\omega(x_0)$, by Proposition~\ref{prop:information_sets}, Part~\ref{prop:information_sets.A(x)=A(x')}.
    As $s^i_\omega$ and $s^i$ are $X$-strategies, we infer that
    \[ s^i_\omega(x) = s^i_\omega(x_0) = \tilde s^i(x_0) \cap W_\omega = s^i(x_0) \cap W_\omega = s^i(x) \cap W_\omega. \]

    (Ad \ref{thm:SEF_well-posed.proof.R=R_omega}):~ Let $s\in S$, $\omega\in\Omega$, $h\in H_\omega$ and $w\in \bigcap h$. 
    Then, 
    \begin{align*}
        R(w,s\mid h) \cap W_\omega =&\, \bigcap \Big\{ s^i(x) \cap W_\omega \mid x \in X,\,i\in J(x)\colon~ w\in x \subseteq \bigcap h \Big\} \\
        =&\, \bigcap \Big\{ s^i_\omega(x) \mid x \in X_\omega,\,i\in J_\omega(x)\colon~ w\in x \subseteq \bigcap h \Big\} \\
        =&\, R_\omega(w,s \mid h).
    \end{align*}
    Indeed, $x\in X$ and $i\in J(x)$ satisfy $w\in x\subseteq \bigcap h$ iff $x\in X^i \cap T_\omega$ and $w\in x\subseteq \bigcap h$, because the collection of $W_{\omega'}$, $\omega'\in\Omega$, equals the set of roots of the decision forest $F$, and is in particular a partition of $W$ (\cite[Theorem~1.7]{Rapsch2024DecisionA}). But $X^i\cap T_\omega = X^i_\omega$ by Equation~\ref{eq:Xi'}. Hence, $x\in X^i\cap T_\omega$ is equivalent to $x\in X_\omega$ and $i\in J_\omega(x)$.
\end{proof}

\begin{proof}[Proof of Corollary \ref{cor:SEF_existence=~order-theoretic_properties}]
    This is a direct consequence of Theorem~\ref{thm:SEF_well-posed} and \cite[Theorem~2]{AlosFerrer2008} (alias \cite[Theorem~5.2]{AlosFerrer2016}). 
    
    Every $\psi$-\textsc{sef} with decision forest $F$ satisfies Property~\ref{def:sef_well-posed}.\ref{def:sef_well-posed.well_posed.existence} iff for every $\psi$-sef $\F$ with decision forest $F$ on an exogenous scenario space $(\Omega,\ms E)$, the induced classical $\psi$-\textsc{sef}s $(T_\omega,I,C_\omega)$ do so for all $\omega\in\Omega$, by Theorem~\ref{thm:SEF_well-posed}. This is the case iff for all connected components $T$ of $F$, all classical $\psi$-\textsc{sef}s with decision tree $T$ satisfy Property~\ref{def:sef_well-posed}.\ref{def:sef_well-posed.well_posed.existence}. Indeed, the ``if'' part is clear, and the ``only if'' part is shown using the exogenous scenario space $\Omega$ given by the set of connected components of $(F,\supseteq)$ with $\ms E = \mc P\Omega$. 
    
    By \cite[Theorem~2]{AlosFerrer2008}, for all connected components $T$ of $F$, all classical $\psi$-\textsc{sef}s with decision tree $T$ satisfy Property~\ref{def:sef_well-posed}.\ref{def:sef_well-posed.well_posed.existence} (are ``playable everywhere'' in the language of \cite{AlosFerrer2008}) iff all connected components $T$ of $F$ are weakly up-discrete and coherent with respect to $\supseteq$. But this is clearly equivalent to $(F,\supseteq)$ being weakly up-discrete and coherent.
\end{proof}

\begin{proof}[Proof of Corollary \ref{cor:SEF_well-posedness=~order-theoretic_properties}]
    This is a direct consequence of Theorem~\ref{thm:SEF_well-posed}, and \cite[Theorem~6]{AlosFerrer2011Comment} and \cite[Corollary 5]{AlosFerrer2011Comment} (alias \cite[Theorem~5.5 and Corollary 5.4]{AlosFerrer2016}). For the remainder of the proof, let $\F$ be a stochastic extensive form on an exogenous scenario space $(\Omega,\ms E)$.\smallskip

    (Ad equivalence of \ref{cor:SEF_well-posedness=~order-theoretic_properties.well-posed} and~\ref{cor:SEF_well-posedness=~order-theoretic_properties.order_properties}):~ $\F$ is well-posed iff for all $\omega\in\Omega$, $(T_\omega,I,C_\omega)$ is well-posed, by Theorem~\ref{thm:SEF_well-posed}. By \cite[Theorem~6]{AlosFerrer2011Comment}, the latter is equivalent to the statement that for all $\omega\in\Omega$, $(T_\omega,\supseteq)$ is regular, weakly up-discrete, and coherent. Clearly, this is equivalent to $(F,\supseteq)$ having these three properties.\smallskip

    (Ad equivalence of \ref{cor:SEF_well-posedness=~order-theoretic_properties.well-posed} and~\ref{cor:SEF_well-posedness=~order-theoretic_properties.order_properties2}):~ $\F$ is well-posed iff for all $\omega\in\Omega$, $(T_\omega,I,C_\omega)$ is well-posed, by Theorem~\ref{thm:SEF_well-posed}. By \cite[Corollary 5]{AlosFerrer2011Comment}, the latter is equivalent to the statement that for all $\omega\in\Omega$, $(T_\omega,\supseteq)$ is regular, and up-discrete. Again, this is clearly equivalent to $(F,\supseteq)$ having these two properties.
\end{proof}

\begin{lemma}\label{lemma:closed_histories_in_coherent_regular_forests}
    Let $(F,\ge)$ be a coherent and regular decision forest and $h\in H$ a closed history. Then, there is $x\in F$ with $h = \uparrow x$.
\end{lemma}

\begin{proof}
    Let $h\in H$ be a closed history. It suffices to show that $h$ has a minimum $x$ since then $h = \uparrow x$, because $h$ is upward closed. 
    
    If $h$ had no minimum, then by coherence there would be a continuation $c$ with a maximum $x$. Then $h = \uparrow x \setminus \{x\}$. Indeed, if $y\in h$, then $y\supsetneq x$ because $h\cup c$ is a chain, $h$ is upward closed, and $x\notin h$. Conversely, let $z\in \uparrow x \setminus \{x\}$. For any $y\in h$, we would have $y\supseteq z$ or $z\supseteq y$ because $\uparrow x$ would be a chain containing $h$; and for any $y\in c$ we would have $z\supsetneq y$ because $x = \max c$. Hence $h\cup c \cup \{z\}$ would be a chain. By maximality of $h\cup c$ and $z\notin c$, we would get $z\in h$. 
    
    Having established that $h = \uparrow x \setminus \{x\}$, we would infer from regularity, that $h = \uparrow x \setminus \{x\}$ would have an infimum. As $h$ is supposed to be closed, it would contain its infimum and thus $h$ would have a minimum -- a contradiction.
\end{proof}

\begin{proof}[Proof of Proposition~\ref{prop:random_histories=~random_moves}]
    Let $\F$ be a well-posed stochastic extensive form on an exogenous scenario space $(\Omega,\ms E)$. The first claim follows directly from Corollary~\ref{cor:SEF_well-posedness=~order-theoretic_properties} and Lemma~\ref{lemma:closed_histories_in_coherent_regular_forests}. Next, let $\h\in\H$ be a closed random history and suppose that $(F,\pi,\X)$ is order consistent. By the first part proven just above, for any $\omega\in D_\h$, there is $x_\omega\in T_\omega$ with $\h(\omega) = \uparrow x_\omega$. For any $\omega\in D_\h$, $x_\omega \in X$ because $\h(\omega)$ is not a maximal chain, but upward closed. For every $\omega\in D_\h$, there is $\x_\omega\in\X$ with $\omega\in D_{\x_\omega}$ and $x_\omega = \x_\omega(\omega)$. 

    Let $\omega,\omega'\in D_\h$. Then, as $\h$ is a random history, $D_\h \subseteq D_{\x_\omega}\cap D_{\x_{\omega'}}$ and $\x_\omega(\omega')\in\h(\omega')$. Thus, $\x_\omega(\omega') \supseteq \x_{\omega'}(\omega')$. By order consistency, then, $\x_\omega \ge_\X \x_{\omega'}$. Repeating the argument after having permuted $\omega$ and $\omega'$ yields the converse inequality whence $\x_\omega = \x_{\omega'}$. 
\end{proof}

\begin{proof}[Proof of Theorem~\ref{thm:absent_minded_driver_Gilboa_sef_well-posed}]
    Let $p\in[0,1]$, $s\in S$ and $\Pr$ be a suitable \textsc{eu} preference structure. Let $i,j\in\{1,2\}$ with $i\neq j$ and let $\tilde s\in S$ with $\tilde s^i = s^i$. As above, $\tilde s^j$ can be identified with an exit event $\tilde E_j\in\ms F^j_{\x_j}$.

    Note that $E_{\neg i,j} = \{\rho = j\}$, $E_{i,j} = \{\rho = i\}$, and $E^{s}_{i,j} = \{\rho=i\}\cap E_i^\complement$. Hence, using the $\P$-independence of $\rho$ and $\ms F^i_{\x_i}$,
    \[ \P(E_{\neg i,j} \cup E^{s}_{i,j}) = \P(\rho = j) + \P(\rho = i)(1-\P(E_i)) = \frac 12 (1+p). \]
    Thus, for every $E\in\ms E$, dynamic consistency of $(\Pi,s)$ implies that
    \[ \P_{j,\{\x_j\}}(E) = 2\,\frac{\P((\{\rho = j\} \cup [\{\rho=i\}\cap E_i^\complement])\cap E)}{1+p}. \]
    Easy computations using the $\P$-independence of $\rho$, $\ms F^i_{\x_i}$ and $\ms F^j_{\x_j}$ yield:
    \begin{align*}
        \P_{j,\{\x_j\}}(\rho=j,\tilde E_j^\complement,E_i\mid \ms F^j_{\x_j}) =&\, \frac{(1-p)1_{\tilde E_j^\complement}}{1+p}, \\
        \P_{j,\{\x_j\}}(\rho=i,\tilde E_j\mid \ms F^j_{\x_j}) =&\, \frac{p1_{\tilde E_j}}{1+p}, \\
        \P_{j,\{\x_j\}}(\rho=j,\tilde E_j^\complement,E_i^\complement\mid \ms F^j_{\x_j}) =&\, \frac{p1_{\tilde E_j^\complement} }{1+p}, \\
        \P_{j,\{\x_j\}}(\rho=i,\tilde E_j^\complement\mid \ms F^j_{\x_j}) =&\, \frac{p1_{\tilde E_j^\complement}}{1+p}. \\
    \end{align*}
    
    Then,
    \begin{align*}
        \pi_{j,\{\x_j\}} =&\, \E_{j,\{\x_j\}}\Big[4 \cdot \Big(1\{\rho=j,\tilde E_j^\complement,E_i\} + 1\{\rho=i,\tilde E_j\}\Big) \\
        & \qquad\quad + 1 \cdot \Big(1\{\rho=j,\tilde E_j^\complement,E_i^\complement\} + 1\{\rho=i,\tilde E_j^\complement\}\Big) \\
        & \qquad\quad + 0 \cdot 1\{\rho=j,\tilde E_j\} \Big] \\
        =&\, \frac{1}{1+p} \Big(1_{\tilde E_j^\complement} \big(4(1-p)+p+p\big) + 1_{\tilde E_j} 4p\Big)\\
        =&\, \frac{1}{1+p} \Big(1_{\tilde E_j^\complement} \big(4-2p\big) + 1_{\tilde E_j} 4p\Big).\\
    \end{align*}
    
    For $p<\frac 23$, this is maximised $\P_{j,\{\x_j\}}$-almost surely by all $\tilde E_j\in\ms F^j_{\x_j}$ with $\P_{j,\{\x_j\}}(\tilde E_j) = 0$. If $(s,\Pr)$ were in equilibrium, then, $p = 1 - \P(E_j) = 1$, a contradiction. 

    For $p>\frac 23$, conversely, the above expression is maximised $\P_{j,\{\x_j\}}$-almost surely by all $\tilde E_j\in\ms F^j_{\x_j}$ with $\P_{j,\{\x_j\}}(\tilde E_j) = 1$. If $(s,\Pr)$ were in equilibrium, then, $p = 1 - \P(E_j) = 0$, a contradiction. 

    Hence, if $(s,\Pr)$ is in equilibrium, then $p = \frac 23$. Conversely, suppose that $p = \frac 23$. Then, by the computation above, $\pi_{j,\{\x_j\}} = \frac 85$, which is independent on the chosen strategy $\tilde s^j$ alias $\tilde E_j$. In particular, $\tilde s^j$ is a best response to $s^i$. Switching the roles of $i$ and $j$ also shows that $\pi_{i,\{\x_i\}} = \frac 85$ so that $(s,\Pr)$ is in equilibrium.
\end{proof}

\begin{proof}[Proof of Theorem~\ref{thm:AP_sef_well-posed}]
    Let $\D$ be action path stochastic extensive form data on an exogenous scenario space $(\Omega,\ms E)$ with well-ordered time $\T$ and let $\F$ be the induced stochastic extensive form. In view of Theorem~\ref{thm:SEF_well-posed}, it suffices to show that the underlying decision forest $F$ is up-discrete and regular with respect to $\supseteq$.\smallskip

    (Ad up-discreteness):~ Let $c$ be a non-empty chain in $(F,\supseteq)$. Then, as $F$ is a decision forest, $\bigcap c\neq\emptyset$. Indeed, $c$ is contained in some maximal chain of the form $\uparrow \{w\}$ for some $w\in W$, see \cite[Definition 1.3, Proposition 1.4]{Rapsch2024DecisionA}. Therefore, there is $w=(\omega,f)\in \bigcap c$, which is an element of $W$. Hence, there is $\T'\subseteq \T$ with
    \[ c \cup \{\{w\}\} = \{ x_t(\omega,f) \mid t\in\T'\} \cup \{\{w\}\}. \]
    If $\T'$ is empty, then $w = \max c$. If $\T'$ is non-empty, then, by hypothesis, it has a minimum $t_0$ with respect to the well-order $\le$ on $\T$. Then, by the very definition of $x_t(\omega,f)$ for $t\in\T'$ and the fact that they all contain $w$, we directly obtain $x_{t_0}(\omega,f) = \max c$.\smallskip

    (Ad regularity):~ Let $x\in F$ be non-maximal. If $\uparrow x \setminus \{x\}$ has a minimum, then it also has an infimum. It therefore remains to consider the case where $\uparrow x \setminus \{x\}$ has no minimum. For that proof, let
    \[ B_x = \{ y\in F \mid \forall z\in\uparrow x \setminus \{x\} \colon z\supseteq y\} \]
    denote the set of lower bounds of $\uparrow x \setminus \{x\}$, a set that clearly contains $x$.
    
    If $x$ is terminal, then $x = \{w\}$ for some $w = (\omega,f)\in W$. Let $y\in B_x$ and $w'\in y$. Then $w'=(\omega,f')$ for some $f'\in\A^\T$ because $\uparrow x\setminus \{x\}$ is non-empty by non-maximality of $x$, implying that $x$ and $y$ are elements of the same connected component. If we had $f\neq f'$, then by hypothesis there would be a minimal $t\in\T$ with $f(t) \neq f'(t)$. Thus, $x_t(\omega,f)\in X$, and as $\uparrow x \setminus \{x\}$ is assumed to have no minimum, there would be $u\in\T$ with $t<u$ and $x_u(\omega,f)\in \uparrow x \setminus \{x\}$. But then $w'\notin x_u(\omega,f)$, in contradiction to $w'\in y \subseteq x_u(\omega,f)$, as $y\in B_x$. Hence, the assumption $f\neq f'$ was false and we must have $f=f'$. Thus, $w'=w$ and $y=\{w\}=x$. We conclude that $B_x=\{x\}$. Hence, $B_x$ has a maximum and $\uparrow x \setminus \{x\}$ an infimum.

    If $x$ is not terminal, then $x\in X$ and $x = x_{\mf t(x)}(w)$ for some $w = (\omega,f)\in W$. Let $y\in B_x$ and let $w'\in y$. Then $w' = (\omega,f')$ for some $f'\in \A^\T$, for similar reasons as above. Then $f'=f$ or $f'\neq f$. In the latter case, the hypothesis implies the existence of minimal $t\in\T$ with $f(t) \neq f'(t)$. If we had $t<\mf t(x)$, then there would be $u\in \T$ with $t<u<\mf t(x)$, since otherwise $x_t(\omega,f)$ would be a minimum of $\uparrow x\setminus \{x\}$ which does not exist by assumption. But then, just as above, $w'\notin x_u(\omega,f)$, in contradiction to $w'\in y \subseteq x_u(\omega,f)$, as $y\in B_x$. Hence, $\mf t(x) \le t$ which implies that $f'|_{[0,\mf t(x))_\T} = f|_{[0,\mf t(x))_\T}$. Hence, whether $f=f'$ or not, it follows that $w'=(\omega,f')\in x$. We infer that $y\subseteq x$. Thus, $x$ is a maximum of $B_x$ and an infimum of $\uparrow x \setminus \{x\}$.    
\end{proof}